\documentclass[aps,pra,reprint,superscriptaddress,nofootinbib]{revtex4-2}
\usepackage{graphicx}
\usepackage{mathpazo}
\usepackage{dcolumn}
\usepackage{xcolor}
\usepackage{bm}
\usepackage{braket}
\usepackage{lipsum}
\usepackage{amsmath, amsthm, amsfonts, amssymb, bbm, dsfont, quantikz, ragged2e} 
\usepackage{enumitem}
\usepackage{mathtools}
\usepackage{nicematrix}
\usepackage{pifont}
\usepackage[most]{tcolorbox}
\usepackage[colorlinks,linkcolor=blue,urlcolor=blue, citecolor=blue]{hyperref}
\usepackage{tikz}

\usetikzlibrary{decorations}

\newcommand{\ketbra}[2]{|#1\rangle\! \langle #2|}

\definecolor{teal}{RGB}{42, 157, 143}
\definecolor{yellow}{RGB}{233, 196, 106}
\definecolor{red}{RGB}{210, 66, 66}
\definecolor{lred}{RGB}{222,94,100}
\definecolor{lblue}{RGB}{179, 235, 242 }

\newcommand{\fref}[1]{\textcolor{blue}{\hyperref[#1]{Fig.$\,$\bfseries\ref{#1}}}}
\newcommand{\lemref}[1]{\textcolor{blue}{\hyperref[#1]{Lemma$\,$\bfseries\ref{#1}}}}
\newcommand{\thmref}[1]{\textcolor{blue}{\hyperref[#1]{Thm.$\,$\bfseries\ref{#1}}}}

\newtheorem{theorem}{Theorem}[section]

\newtheorem{proposition}{Proposition}[section]

\newtheorem{lemma}[theorem]{Lemma}
\theoremstyle{definition}
\newtheorem{definition}{Definition}[section]

\tikzset{
    every node/.style={font=\small},
    arrow/.style={-{Stealth}, thick},
    implies/.style={->, double equal sign distance, thick}
}

\newtcolorbox[auto counter]{mybox}[2][]{
    breakable = false,
    enhanced,
    sharp corners,
    colback=violet!3!white,
    colframe=violet!40!white,
    fonttitle=\bfseries,
    title={\centering \strut #2}, 
    enlarge bottom at break by=5mm,
    enlarge top at break by=5mm,
    overlay first={%
        \draw[black, line width=0.5mm](frame.south west)--(frame.south east);
        \node[anchor=north east] at (frame.south east) {continued on next page};
    },
    overlay middle={%
        \draw[black, line width=0.5mm](frame.south west)--(frame.south east);
        \draw[black, line width=0.5mm](frame.north west)--(frame.north east);
        \node[anchor=north east] at (frame.south east) {continued on next page};
        \node[anchor=south west] at (frame.north west) {continued from next page};
    },
    overlay last={%
        \draw[black, line width=0.5mm](frame.north west)--(frame.north east);
        \node[anchor=south west] at (frame.north west) {continued from last page};},
    #1
}

\let\oldaddcontentsline\addcontentsline
\renewcommand{\addcontentsline}[3]{}

\begin{document}

\preprint{APS/123-QED}

\title{Cooling a Qubit using $n$ Others}

\author{Jake Xuereb}
\email{jake.xuereb@tuwien.ac.at}

\affiliation{Vienna Center for Quantum Science and Technology, Atominstitut, TU Wien, 1020 Vienna, Austria}
\author{Benjamin Stratton}
\email{ben.stratton@bristol.ac.uk}
\affiliation{Quantum Engineering Centre for Doctoral Training, H. H. Wills Physics Laboratory and Department of Electrical \& Electronic Engineering, University of Bristol, BS8 1FD, UK}
\affiliation{H.H. Wills Physics Laboratory, University of Bristol, Tyndall Avenue, Bristol, BS8 1TL, UK}
\author{Alberto Rolandi}
\affiliation{Vienna Center for Quantum Science and Technology, Atominstitut, TU Wien, 1020 Vienna, Austria}
\author{Jinming He}
\affiliation{School of Physical Sciences, University of Science and Technology of China, Hefei 230026, China}
\author{Marcus Huber}
\affiliation{Vienna Center for Quantum Science and Technology, Atominstitut, TU Wien, 1020 Vienna, Austria}
\affiliation{Institute for Quantum Optics and Quantum Information - IQOQI Vienna, Austrian Academy of Sciences, Boltzmanngasse 3, 1090 Vienna, Austria}
\author{Pharnam Bakhshinezhad}
\affiliation{Vienna Center for Quantum Science and Technology, Atominstitut, TU Wien, 1020 Vienna, Austria}

\date{\today}

\begin{abstract}
In the task of unitarily cooling a quantum system with access to a larger quantum system, known as the machine or reservoir, \textit{how does the structure of the machine impact an agent's ability to cool and the complexity of their cooling protocol}? Focusing on the task of cooling a single qubit given access to $n$ separable, thermal qubits with arbitrary energy structure, we answer these questions by giving two new perspectives on this task. Firstly, we show that a set of inequalities related to the energetic structure of the $n$-qubit machine determines the protocol which cools the qubit to the coldest reachable state, which parts of the machine contribute to this protocol and give rise to a Carnot-like bound. Secondly, we show that cooling protocols can be represented as perfect matchings on bipartite graphs enabling the optimization of cost functions e.g., gate complexity or dissipation. Our results generalize the algorithmic cooling problem, establish new fundamental bounds on quantum cooling and offer a framework for designing novel autonomous thermal machines and cooling algorithms.
\end{abstract}

\maketitle

\section{Introduction} 

The task of cooling a quantum system using another~\cite{Reeb_2014,ralph_swap,wilming_third_law,Masanes2017,clivaz_pre_2019,clivaz_prl_2019,taranto_23,bassman_campisi_24} is a foundational point of inquiry in quantum thermodynamics~\cite{Goold_2016, binder2018thermodynamics, strasberg2022quantum}. The problem has been studied in a multitude of settings including both closed and open system dynamics and has inspired practical computational subroutines such as algorithmic cooling~\cite{Park2016,schulman_vazirani,schulman_limits,raeisi_mosca,naye_prl,Alhambra2019heatbathalgorithmic,naye_comparison} and autonomous quantum refrigerators~\cite{skrzypczyk_10,ronnie_12,Mitchison_review}. Both have applications to passive error correction and state preparation, with the latter recently experimentally realised to create an autonomous reset mechanism for a superconducting qubit~\cite{Aamir2025}.

In this work, we consider the task of cooling a qubit via a unitary interaction with $n$ thermal qubits which will often be referred to as the \textit{machine} throughout this work. By cooling, we mean increasing the ground state population of the system. Here, we go beyond prior work~\cite{bassman_campisi_24,schulman_vazirani,schulman_limits} by allowing the $n$ thermal qubits used for cooling to be non-identical; the qubits can have varying energy gaps, effectively corresponding to different temperatures. This enables an understanding of how the energetic structure of the $n$-thermal qubits impacts the cooling scenario. Moreover, it brings the scenario closer to practical reality where physical qubits are unlikely to be identical in real hardware, whether intentionally or by design.

\begin{figure}[t]
\begin{mybox}{Main Results}
\justifying
When unitarily cooling a qubit with energy gap $\omega$ at temperature $T_S$ with access to $n$ qubits at temperature $T_M$ and increasing gaps $\Gamma = (\gamma_1, \gamma_2, \dots, \gamma_n) \, : \,E_\text{Max} = \sum^n_{j=1} \gamma_j$ the optimal cooling protocol (maximally increasing the ground state population of $S$) is specified by; \textit{A ground state subspace energy level $\ket{0_S i_M}$  should be exchanged with an energy level from the excited state subspace $\ket{1_S j_M}$, for some $i_M, j_M \in \{0,1\}^{\times n}$ if} 
     \begin{gather*}
        \frac{1}{2}\left(\frac{T_M}{T_S}\omega + E_\text{Max}\right) < E(i_M)
    \end{gather*} 
where $E(i_M)$ is the energy corresponding to the level $\ket{0_S i_M}$ given by $i_M \cdot \Gamma$. Intuitively, cooling $S$ depends on the energetic structure of $M$--whose individual energies can be ordered. The macroscopic quantities forming the L.H.S of the above inequality determines the midpoint in this order. If an energy falls above this midpoint, a corresponding energy may be found below it such that exchanging these energy levels cools $S$. Using these inequalities we are able to obtain the following in this work.
\begin{itemize}
    \item[Sec. III] A single condition to determine the minimal set of energy level exchanges needed for optimal cooling, leading to a set of inequalities that characterise a cooling scenario.

    \item[Sec. III] A Carnot-\textit{like} bound on any energy level exchange that cools the system. 

    \item[Sec. IV] Conditions determining when sets of machine qubits are energetically irrelevant for cooling, characterising when a machine is reducible (some machine qubits do not contribute) or irreducible (all the machine contributes).

    \item[Sec. IV] A framework to optimise cooling unitaries relative to a cost using minimum weight perfect matching on bipartite graphs~\cite{godsil01,comb_opt}.  
\end{itemize}
\end{mybox}
\end{figure}

Previous works on cooling qubit systems have, in addition to using identical qubits, allowed the machine qubits to be refreshed between successive rounds of unitary cooling~\cite{Alhambra2019heatbathalgorithmic,clivaz_pre_2019,clivaz_prl_2019}, with asymptotic bounds on the cooling found in this case. Alternately, fixed machines have been considered but have been either asymptotically large~\cite{karen_dynamical_2011,memory_taranto,taranto_23,taranto2024efficientlycoolingquantumsystems}, very small~\cite{skrzypczyk_10,ralph_swap} and/or used identical qubits~\cite{schulman_vazirani, schulman_limits,raeisi_mosca,bassman_campisi_24}. In this work, we remove the assumptions of identical machine qubits and the ability to rethermalise the machine, focusing on finite machines. By removing the ability to rethermalise qubits, we allow the machine structure's effect on cooling to be isolated and studied. In addition, we ensure that our cooling protocols can be physically implemented as a single quantum circuit and examined in terms of their gate complexity, given they contain no non-unitary steps nor asymptotic considerations.

However, even in the case of non-identical machine qubits, this is a \textit{heuristically} solved problem. It is part of the folklore of quantum thermodynamics that \textit{to cool a system $S$ given access to an auxiliary machine $M$, one need only take their joint system and machine density and reorder its eigenvalues in decreasing order}~\cite{Allahverdyan_2004,Horodecki2013,clivaz_prl_2019}. Surprisingly, little attention has been paid to the fine-grained details of these cooling
dynamics, namely \textit{what form does this sorting take and how does this depend on the energetic structure of the machine}? This is despite the potential for vast differences in the complexity of processes which generate these cooling operations. For instance, the energy level exchanges of a cooling protocol could simplify down to a \texttt{SWAP} between the system and one of the $n$ machine qubits. Alternatively, it could require an operation controlled on all $n$ machine qubits. Recently, it was shown that in a single-shot state transformation under unitary dynamics, the energetically optimal trajectory can be identified on a polytope of potential population vectors the system can take~\cite{silva2024optimalunitarytrajectoriescommuting}. In this work we address this gap by presenting a framework that not only identifies the necessary energy level permutations for optimal cooling in terms of a set of physically motivated inequalities, but also evaluates the difficulty of implementing them. Specifically, we establish criteria which we term \textit{reducibility conditions} to detect whether a cooling interaction involves all of the machine or merely a part of it. Shockingly, we find that some cooling scenarios only make use of a single qubit from an $n$ qubit machine. But our criteria can be used as tools to construct \textit{irreducible} machines whose parts are always useful. We further explore these dynamics by presenting a method for optimising over the set of optimal cooling unitaries for some given constraint, here considering the example of minimising gate complexity.

\begin{figure*}[t]
    \centering
\includegraphics[width=\linewidth]{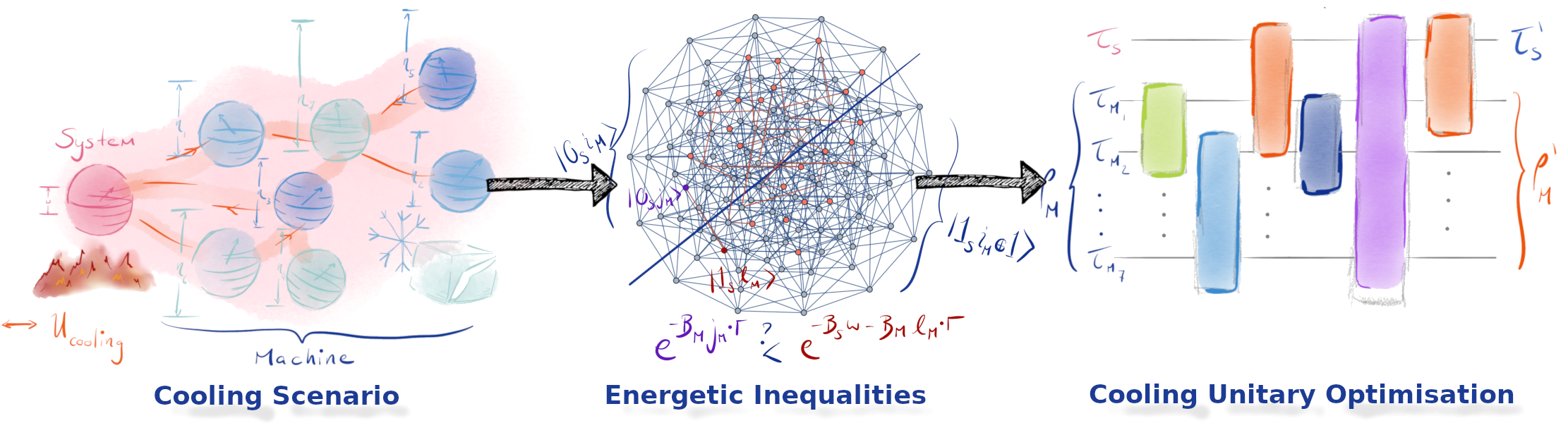}
    \caption{\textit{Summary Illustration --} We examine unitary cooling scenarios where a qubit at temperature $\beta_S$ with gap $\omega$ is coupled to an $n$ qubit machine at temperature $\beta_M$ with energy gap structure $\Gamma = (\gamma_1 ,\gamma_2, \dots, \gamma_n)$. Since the temperature of a two-level system is determined by the ratio of its ground and excited state populations, to cool the qubit system an agent must exchange disordered populations on the global population vector from the ground state subspace $\ket{0_S i_{M}}$ to the excited state subspace $\ket{1_S j_{M}}$. We find that such exchanges only cool if the energy $E(i_M)$ is larger than $1/2(T_M\omega/T_S + E_\text{Max})$ exposing a set of inequalities which determine various properties of the cooling scenario. By finding a representation for cooling unitaries as perfect matchings on bipartite graphs we are able to optimise cooling protocols relative to a cost function by using minimum weight perfect matching techniques.}
\end{figure*}

This manuscript is organised as follows; in Sec.~\ref{sec:example} we present a simple three qubit example which exposes the main insights of this work. In Sec~\ref{sec:ineq} we introduce a set of inequalities that fully characterise a cooling scenario. In Sec~\ref{sec:unitary} we begin by presenting the reducibility conditions for a cooling scenario and move on to the problem optimising cooling unitaries relative to a cost function which we solve by using a novel graph representation. We follow with closing discussions and conclude after which technical matter is provided in the appendices.
\section{Setting \& Illustrative Example}
\label{sec:example}

\textit{What does it mean to cool a qubit}? The state of a qubit always has a well defined temperature. This is due to qubit states always being describable by a Gibbs state with respect to a Hamiltonian $H = E_g\ketbra{g}{g} + E_e\ketbra{e}{e}$ at inverse temperature $\beta$, \begin{gather}\tau_\beta = \frac{\ketbra{g}{g} + e^{-\beta (E_e-E_g)}\ketbra{e}{e} }{1 + e^{-\beta (E_e-E_g)}}, ~ ~ \beta = \frac{1}{k_BT}~,
\end{gather}
where $T$ is the temperature and $k_B$ is the Boltzmann constant. The goal of cooling can be rephrased as increasing the ground state population of the qubit, which effectively lowers $T$ when $H$ is fixed. As such, we may, at points, also refer to qubits with larger energy gaps as being ``\textit{colder}'', with this alluding to a lower effective temperature with respect to the qubit with the smallest energy gap. Using this approach, one can apply the results we present to qubits in general states by \emph{mathematically} assigning such a Hamiltonian to the individual system and machine qubits. Throughout the work, without loss of generality, we take $\hbar = 1$ and fix the computational basis to be $\ketbra{g}{g} = \ketbra{0}{0}, \ketbra{e}{e} = \ketbra{1}{1}$. We denote the partition function of $\tau_\beta$ as $\mathcal{Z} = \textrm{tr}\{e^{-\beta H}\}$ and when working with bit strings use superscripts, e.g., $i^n \in \{0,1\}^{\times n}$ to denote the length of a string and subscripts e.g. $i_2$ in $i_1i_2i_3$, to denote the position of the index of a bit.

Before introducing the general case we present an expository example. Consider three qubits, all with $\beta~=~1$: a system qubit, $\tau_S = \frac{1}{1+e^{\omega}}(\ketbra{0}{0} + e^{-\omega}\ketbra{1}{1})$, and two machine qubits, $\tau_{M_1} = \frac{1}{1+e^{- \gamma_1}}(\ketbra{0}{0} + e^{-\gamma_1}\ketbra{1}{1})$ and $\tau_{M_2} = \frac{1}{1+e^{-\gamma_2}}(\ketbra{0}{0} + e^{-\gamma_2}\ketbra{1}{1})$. We assume that $\omega \leq \gamma_1 \leq \gamma_2$, such that $M_2$ could be colder than $M_1$ that could be colder than $S$. Such an ordering is always possible if one aims to cool the hottest qubit $S$. We now seek to understand \textit{by how much can $S$ be cooled via a joint unitary on $S, M_1$, and $M_2$ for different $M_1$ and $M_2$, and what operations are required to achieve this cooling}?

The state of the system and the machine can be expressed as a diagonal matrix of 8 populations
\begin{align}
    \tau_S \hspace{-0.05cm} \otimes \hspace{-0.05cm} \tau_M &= \hspace{-0.85cm} \sum_{i_S, i_{M_1}, i_{M_2} \in \{0,1\}} \hspace{-0.75cm} \frac{e^{-(i_S \omega + i_{M_1}\gamma_1 + i_{M_2}\gamma_2)}}{ \mathcal{Z}_{S}\mathcal{Z}_{M}} \ketbra{i_S i_{M_1} i_{M_2}}{i_S i_{M_1} i_{M_2}}\nonumber\\
    &=\frac{1}{\mathcal{Z}_{S}\mathcal{Z}_{M}}\text{diag}\underbrace{\left(1, e^{- \gamma_2}, e^{-\gamma_1}, e^{-(\gamma_1 + \gamma_2)}, \right.}_{\ket{0_S i_M}}  \\
    &\hspace{2cm}\underbrace{\left. e^{- \omega}, e^{-(\omega + \gamma_2)}, e^{-(\omega+ \gamma_1)}, e^{-(\omega+\gamma_1 + \gamma_2)}\right)}_{\ket{1_S i_M}}, \nonumber
\end{align}
where $\mathcal{Z}_S = (1+e^{-\omega})$ and $\mathcal{Z}_{M} = (1+e^{-\gamma_1})(1+e^{-\gamma_2})$. It proves useful to adopt a bitwise product to denote these populations. Specifically, we introduce the machine gap vector $\Gamma = (\gamma_1 , \gamma_2)$, and define its product with the 2-bit-string $i_M = i_{M_1}i_{M_2}$ that labels a given basis state, as $i_M \cdot \Gamma = i_{M_1}\gamma_1 + i_{M_2}\gamma_2$ giving the energy associated to this level. With this, we are in a position to set the scene. Initially, the ground state population of S is $p_{0} = \sum_{i_M}\frac{e^{- i_M \cdot \Gamma}}{\mathcal{Z}_{SM}}$ obtained by summing over entries $\ket{0_S i_M} \, : i_M \in \{0,1\}^{\times 2}$ in a partial trace and similarly the initial excited state population of $S$ is $1~-~p_0~=~\sum_{i_M} \frac{e^{-(\omega + i_M \cdot \Gamma)}}{\mathcal{Z}_{SM}}$. To optimally cool $S$, our aim is therefore to maximise the sum over the entries $\ket{0_S i_M}$ under the action of a joint unitary on $S, M_1$ and $M_2$. Any unitary cooling operation $U(\tau_S \otimes \tau_M)U^\dagger$ must preserve these global populations as they are eigenvalues. Thus, to cool $S$ we can only permute the $SM$ basis kets to ensure that the populations contributing to the ground state population after cooling ( $p_0'$)  are the largest. That is, to perform this optimal cooling we seek a permutation that assigns the $4$ largest populations of $SM$ to the basis states $\ket{0_s i_M}$.

To achieve this, we first note that all permutations of the global state can be broken down into a product of exchanges featuring only two energy levels of $SM$, which we will often term two-level permutations (TLP). In particular, we are interested in cooling TLPs which exchange $\ket{0_S i_M} \leftrightarrow \ket{1_S j_M}$ if their associated populations satisfy
\begin{equation}
    e^{-{ i_M \cdot \Gamma}} < e^{-{ (\omega + j_M\cdot \Gamma)}}, \label{eq:which_swaps_cool}
\end{equation}
as the new population of the basis state $\ket{0_S i_M}$ will now be larger. Using Eq.~\eqref{eq:which_swaps_cool} one can find a set of TLP that are necessary to ensure that the $4$ largest populations are associated to the $\ket{0_s i_M}$ basis states.

Due to the non-decreasing gap structure, $\omega \leq \gamma_1 \leq \gamma_2$, the $\ket{0_Si_M}$ populations in non-increasing order are $1 \geq e^{- \gamma_1} \geq e^{- \gamma_2} \geq e^{-( \gamma_1 + \gamma_2)}$. Given the $\ket{1_Si_M}$ populations are the $\ket{0_Si_M}$ populations multiplied by the constant $e^{-\omega}$, they too are ordered in the same way. The TLP of the basis kets associated to the largest and smallest populations is therefore cooling if $e^{-(\gamma_1 + \gamma_2)} < e^{-\omega}$, which holds if $\omega < \gamma_1 + \gamma_2$. This always holds due to the gap structure, and hence the TLP $\ket{0_S11} \leftrightarrow \ket{1_S00}$ is always cooling and leads to the largest amount of cooling. Considering the next largest population in $\ket{1_Sj_M}$ against the next smallest in $\ket{0_Si_M}$, the associated TLP is cooling if $e^{-\gamma_2} < e^{-(\omega + \gamma_1)}$, which holds if $\omega < \gamma_2 - \gamma_1$. This is not always true and depends on the energetic structure of the machine. The remaining potential exchanges $\ket{0_S10} \leftrightarrow \ket{1_S01}$ and $\ket{0_S00} \leftrightarrow \ket{1_S11}$ can never cool as $\ket{0_S00}$ and $\ket{0_S01}$ already hold two of the four largest populations in the state. As we verified, $\ket{1_S00}$ holds one of the four largest, and the TLP $\ket{0_S11} \leftrightarrow \ket{1_S00}$ always ensures its correct placement. This leaves just one final population, with two possible scenarios a) if $\omega \geq \gamma_2 - \gamma_1$, $\ket{0_S01}$ holds the $4$th largest population already, b) if  $\omega < \gamma_2 - \gamma_1$, $\ket{1_S10}$ holds the $4$th largest population and the TLP $\ket{0_S10} \leftrightarrow \ket{1_S01}$ is needed.

In scenario a) the only TLP that must be applied is $\ket{0_S11} \leftrightarrow \ket{1_S00}$. A unitary $S_{011,100}= \ketbra{0_{S}11}{1_{S}00} + \ketbra{1_{S}00}{0_{S}11} + \mathbb{1}_{\rm Rest}$ optimally cools $S$ and can be broken down into a series of \texttt{Toffoli} gates as in Fig.~\ref{fig:3_qubit}. The final ground population of the system is then \hbox{$p'^{(a)}_0 = \frac{1}{\mathcal{Z}_S \mathcal{Z}_M}\left(1 + e^{-\gamma_2} + e^{-\gamma_1}  + e^{-\omega} \right)$}. It can be seen that $p^{(a)}_0$ is larger than the initial ground state population of $M_2$, the coldest qubit, if $p^{(a)}_0 > 1/(1 + e^{-\gamma_2})$ which occurs when $\omega > \gamma_2 - \gamma_1$, which is the criteria for scenario a). In scenario b) both $\ket{0_S11} \leftrightarrow \ket{1_S00}$ and $\ket{0_S01} \leftrightarrow \ket{1_s10}$ must be implemented. With both conditions fulfilled, $i_{M_1}$ becomes a free index such that all exchanges $\ket{0_S i_{M_1} 1} \leftrightarrow \ket{1_S i_{M_1} 0} : i_{M_1} \in \{0,1\}$ must be carried out which is achieved by $\mathtt{SWAP}_{S, M_2}$. The final achievable ground state population after cooling is therefore implicitly that of $M_2$ i.e. $p'^{(a)}_0 = 1/(1 +e^{- \gamma_2})$.

This is a curious outcome. In b) where the energetic gap of $S$ is smaller than the difference between the energy gaps of the machine qubits, $M_1$ does not contribute to the optimal cooling! The operation for optimal cooling is purely bipartite, with one machine qubit left unused. On the other hand, in a), where the energetic gap of $S$ is larger than the difference between the energy gaps of the machine qubits, an agent is able to cool further than the temperature of the coldest qubit. However, the dynamics are necessarily tripartite. We can better understand these two cooling scenarios by visualising them as vertex weighted graphs. Let each eigenket $\ket{i_S i_{M_1}i_{M_2}}$ correspond to a vertex weighted by $e^{-\beta(i_S\omega + i_M \cdot\Gamma)}$ and let two vertices $v_i$ and $v_j$ be connected by an edge if they correspond to bitstrings $i_S i_{M_1}i_{M_2}$ and $j_S j_{M_1}j_{M_2}$ that have Hamming distance 1 i.e. differ by one bit. This is the weighted 3 bit hypercube graph $Q_3$--the graph of bitstrings of length 3 with edges between bitstrings differing by one bit (Hamming distance 1)~\cite{godsil01}--and permutations of the $SM$ eigenkets for cooling can be thought of as operations on $Q_3$. An example of each of the two cooling scenarios described is represented on $Q_3$ in \fref{fig:3_qubit}. Even cooling a single qubit with access to two other qubits seems to lead to an intricate set of scenarios which whilst complex is seemingly captured by a simple set of relationships between the energy gaps of the system and machine.
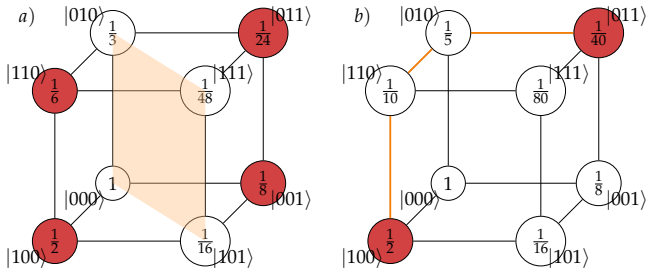
\begin{figure}[t]
\label{fig:eg_2}
\scalebox{0.8}{
\begin{tikzpicture}[scale=2.5]
\def\radius{0.85}
\node[anchor=south east] at (115:\radius + 0.25) {$a)$};
    \node[draw, circle] (C) at (0, 0, 0) {$1$}; 
    \node[anchor=north east] at (0, 0, 0) {$\ket{0_S00}$};
    \node[draw, circle, fill = red] (D) at (1, 0, 0) {$\frac{1}{8}$}; 
    \node[anchor=north west] at (1, 0, 0) {$\ket{0_S01}$};
    \node[draw, circle] (E) at (0, 1, 0) {$\frac{1}{3}$}; 
    \node[anchor=south east] at (0, 1, 0) {$\ket{0_S10}$};
    \node[draw, circle, fill = red] (F) at (1, 1, 0) {$\frac{1}{24}$}; 
    \node[anchor=south west] at (1, 1, 0) {$\ket{0_S11}$};
    \node[draw, circle, fill = red] (G) at (0, 0, 1) {$\frac{1}{2}$}; 
    \node[anchor=north east] at (0, 0, 1) {$\ket{1_S00}$};
    \node[draw, circle] (H) at (1, 0, 1) {$\frac{1}{16}$}; 
    \node[anchor=north west] at (1, 0, 1) {$\ket{1_S01}$};
    \node[draw, circle, fill = red] (A) at (0, 1, 1) {$\frac{1}{6}$}; 
    \node[anchor=south east] at (0, 1, 1) {$\ket{1_S10}$};
    \node[draw, circle] (B) at (1, 1, 1) {$\frac{1}{48}$}; 
    \node[anchor=south west] at (1, 1, 1) {$\ket{1_S11}$};
    \draw (C) -- (D); 
    \draw (C) -- (E); 
    \draw (C) -- (G); 
    \draw (D) -- (F); 
    \draw (D) -- (H); 
    \draw (E) -- (F); 
    \draw (E) -- (A); 
    \draw (F) -- (B); 
    \draw (G) -- (H); 
    \draw (G) -- (A); 
    \draw (H) -- (B); 
    \draw (A) -- (B); 
    \fill[orange!50, opacity=0.4] (0, 1, 0) -- (1, 1, 1) -- (1, 0, 1) -- (0, 0, 0) -- cycle;
\end{tikzpicture}
\hspace{0.1cm}
\begin{tikzpicture}[scale=2.5]
\def\radius{0.85}
\node[anchor=south east] at (115:\radius + 0.25) {$b)$};
    \node[draw, circle] (C) at (0, 0, 0) {$1$}; 
    \node[anchor=north east] at (0, 0, 0) {$\ket{0_S00}$};
    \node[draw, circle] (D) at (1, 0, 0) {$\frac{1}{8}$}; 
    \node[anchor=north west] at (1, 0, 0) {$\ket{0_S01}$};
    \node[draw, circle] (E) at (0, 1, 0) {$\frac{1}{5}$}; 
    \node[anchor=south east] at (0, 1, 0) {$\ket{0_S10}$};
    \node[draw, circle, fill = red] (F) at (1, 1, 0) {$\frac{1}{40}$}; 
    \node[anchor=south west] at (1, 1, 0) {$\ket{0_S11}$};
    \node[draw, circle, fill = red] (G) at (0, 0, 1) {$\frac{1}{2}$}; 
    \node[anchor=north east] at (0, 0, 1) {$\ket{1_S00}$};
    \node[draw, circle] (H) at (1, 0, 1) {$\frac{1}{16}$}; 
    \node[anchor=north west] at (1, 0, 1) {$\ket{1_S01}$};
    \node[draw, circle ] (A) at (0, 1, 1) {$\frac{1}{10}$}; 
    \node[anchor=south east] at (0, 1, 1) {$\ket{1_S10}$};
    \node[draw, circle] (B) at (1, 1, 1) {$\frac{1}{80}$}; 
    \node[anchor=south west] at (1, 1, 1) {$\ket{1_S11}$};
    \draw (C) -- (D); 
    \draw (C) -- (E); 
    \draw (C) -- (G); 
    \draw (D) -- (F); 
    \draw (D) -- (H); 
    \draw (E) -- (F); 
    \draw (E) -- (A); 
    \draw (F) -- (B); 
    \draw (G) -- (H); 
    \draw (G) -- (A); 
    \draw (H) -- (B); 
    \draw (A) -- (B); 
    \draw[orange, thick] (G) -- (A);
    \draw[orange, thick] (A) -- (E);
    \draw[orange, thick] (E) -- (F);
\end{tikzpicture}
}
\caption{In a) we have three qubits at temperature $\beta = 1$ and with gaps $\omega = \log(2), \gamma_1 = \log(3)$ and $\gamma_2 =  \log(8)$ with the populations factored by $\frac{4}{9}$ for ease of readability. In orange we see a plane denoting the axis of reflection corresponding to a $\mathtt{SWAP}$ between the system qubit and the machine qubit 2 (bipartite) which results in the optimal distribution. In b) we have three qubits at temperature $\beta = 1$ but with gaps $\omega = \log(2)$ and $\gamma_1  = \log(5)$ and $\gamma_2 = \log(8)$ giving one less disordered pair, with populations factored by $\frac{81}{40}$. Here the edges in orange correspond to a series of two-level swaps or $\mathtt{Toffolis}$ in the order $\ket{1_S00} \leftrightarrow \ket{1_{S}10}$ then $\ket{0_{S}10} \leftrightarrow \ket{0_{S}11}$ and $\ket{0_{S}10} \leftrightarrow \ket{
1_{S}10}$ (tripartite) that give the optimal distribution.}
\label{fig:3_qubit}
\end{figure}

\section{A Set of Inequalities}
\label{sec:ineq}
\paragraph*{Finding the disordered populations.}Turning our attention to the general problem, imagine a qubit with gap $\omega$ at temperature $\beta_S$ as our system $S$ and $n$-qubits at temperature $\beta_M$ with an energetic gap structure $\Gamma = (\gamma_1, \gamma_2, \dots, \gamma_n)$ such that $\gamma_i \leq \gamma_{i+1} ~ \forall~i \in \{1,N\}$, as the machine $M$ in this cooling scenario. The system and machine have the initial joint state 
\begin{gather}
    \tau_S \otimes \tau_M = \hspace{-0.5cm}\sum_{\substack{i_S \in \{0,1\}\\ i_M \in \{0,1\}^{\times n }}} \frac{e^{-(\beta_S i_S \omega + \beta_M i_M \cdot \Gamma)}}{\mathcal{Z}_{S}\mathcal{Z}_{M}} \ketbra{i_S i_M}{i_S i_M},
\end{gather}
where $\mathcal{Z}_{S} =  (1 + e^{-\beta_S \omega})$ and $\mathcal{Z}_{M} = \prod^{n}_{j = 1}(1+e^{-\beta_M \gamma_j})$. As before, any TLP $\ket{0_S i_M} \leftrightarrow \ket{1_s j_M}$ on $\tau_S \otimes \tau_M$ that an agent performs is cooling on $S$ if their associated populations satisfy $e^{-\beta_M i_M \cdot \Gamma} < e^{-(\beta_S\omega + \beta_M j_M \cdot \Gamma)}$. This implies that any pair of energy eigenkets in the $SM$ Hilbert space that cool down $S$ when permuted must satisfy 
\begin{gather}
    \frac{T_M}{T_S}\omega < (i_M - j_M)\cdot \Gamma, \label{eq:gen_ineq}
\end{gather}
that is the difference in energy associated with this exchange must be greater than $\frac{T_M}{T_S}\omega$. In this case we say that the populations of the eigenkets $\ket{0_Si_M}$ and $\ket{1_Sj_M}$ are disordered. To find the TLP's needed for performing optimal cooling, one can then consider the strategy of iteratively swapping the largest populations $\ket{1_S j_M}$ with the smallest populations $\ket{0_S i_M}$, as in the example above. Conveniently, in the general case this can be achieved by always considering the TLP's $\ket{0_S i_M} \leftrightarrow \ket{1_S i_M \oplus 1}$, where $\oplus 1$ is a bit-wise addition $\!\!\!\mod \!2$. See Appendix.~\ref{app:Hamming_weight_conj} for the proof. Using Eq.~\eqref{eq:gen_ineq}, a given TLP is therefore necessary for optimal cooling if 
\begin{gather}
    \frac{T_M}{T_S}\omega < (i_M - (i_M \oplus 1) )\cdot \Gamma, \label{eq:hamming_weight_con_ineq}
\end{gather}
which, given that $(i_M \oplus 1) \cdot \Gamma  = \sum^{n}_{l = 1} (i_l \oplus 1 )\gamma_l = \sum^{n}_{l = 1} \gamma_l - \sum_{l} i_l \gamma_l = E_\text{Max} - (i_M \cdot \Gamma)$, where $E_{\rm Max} = \sum_{l=1}^{n} \gamma_l$, simplifies to 
\begin{gather}
    \frac{1}{2}\left(\frac{T_M}{T_S} \omega + E_\text{Max}\right) < i_M\cdot\Gamma. \label{eq:ineq_simp}
\end{gather}

This is our first main result, an inequality that detects every disordered population that must be exchanged to cool $S$ in terms of whether a change in energy is greater than the product of the ratio of temperatures of system and machine and the energy gap of $S$. The quantity $\frac{1}{2}\left(\frac{T_M}{T_S}\omega + E_\text{Max}\right)$ will appear several times in this manuscript and appears to be a constant that dictates many properties of the cooling scenario such as the order of interaction required to cool and how much a given machine is able to cool.

We can note that if the energy level $i_M$ of the machine has an energy $i_M\cdot \Gamma$ above the threshold $\frac{1}{2}\left(\frac{T_M}{T_S}\omega + E_\text{Max}\right)$ and another energy level $l_M$ has energy larger than $i_M$ (i.e. $l_M \cdot \Gamma > i_M \cdot \Gamma$) then $\ket{0_S l_M}$ should also be exchanged for cooling if $\ket{0_S i_M}$ is to be exchanged. In fact, considering the energies of the machine we have the largest energy $E_\text{Max} = 1^n\cdot\Gamma$, smallest energy $0 = 0^n\cdot\Gamma$ and every energy in between given by $i_M\cdot \Gamma$ for $i_M \in \{0,1\}^{\times n}$. Since $(i_M \oplus 1)\cdot \Gamma~=~E_\text{Max}- i_M\cdot \Gamma$, if $i_M\cdot \Gamma$ is above $E_\text{Max}/2$ then $(i_M \oplus 1) \cdot \Gamma$ is below $E_\text{Max}/2$. Therefore we find that $E_\text{Max}/2$ is the median of the energies. This allows us to reinterpret Eq.~\eqref{eq:gen_ineq} as; \textit{if the energy $i_M\cdot \Gamma$ is above the shifted median value $\frac{1}{2}\left(\frac{T_M}{T_S}\omega + E_\text{Max}\right)$ then there exists at least one energy $j_M\cdot \Gamma$ below the median such that exchanging $\ket{0_S i_M}$ and $\ket{1_S j_M}$ cools  $S$.} For more details and visualisations see Appendix~\ref{app:Hamming_weight_conj}.

\paragraph*{A Carnot-\textit{like} bound.} Manipulating Eq.~\eqref{eq:gen_ineq} we surprisingly find that the change in effective energy gap of the system due to a single energy level exchange $\ket{0_S i_M} \leftrightarrow \ket{1_S j_M}$ in the system-machine Hilbert space is upper-bounded by the Carnot efficiency of the system and machine $\eta_c~=~1 - T_M/T_C$. To expose this fundamental bound, divide Eq.~\eqref{eq:gen_ineq} by $\omega$, multiply each side by $-1$ and add 1 to both sides giving 
\begin{gather}
   \frac{\Delta E_S}{E_S} = \frac{\omega - (i_M - j_M)\cdot \Gamma}{\omega} < 1 - \frac{T_M}{T_S}
\end{gather}
where $\Delta E_S = \omega - (i_M - j_M)\cdot \Gamma$ and $E_S = 0 + \omega$. This bound suggests that the change in effective energy gap of the system induced by a global energy ket exchange is upper bounded by the Carnot efficiency of $SM$ which is independent of the energetic structure $\Gamma$ of the $n$-qubit machine. A conventional Carnot efficiency bound cannot be obtained in this scenario as the reduced state of the machine after the cooling interaction will no longer be at its initial temperature nor thermal w.r.t the machine Hamiltonian. It is interesting to see that despite this the Carnot efficiency still bounds the population changes a single TLP can induce. This bound is saturated when $(i_M - j_M)\cdot \Gamma = \frac{T_M}{T_S}\omega$ which corresponds to energy level exchange which results in the smallest change in energy that still cools $S$.

\begin{figure*}[t]
    \includegraphics[width = \textwidth]{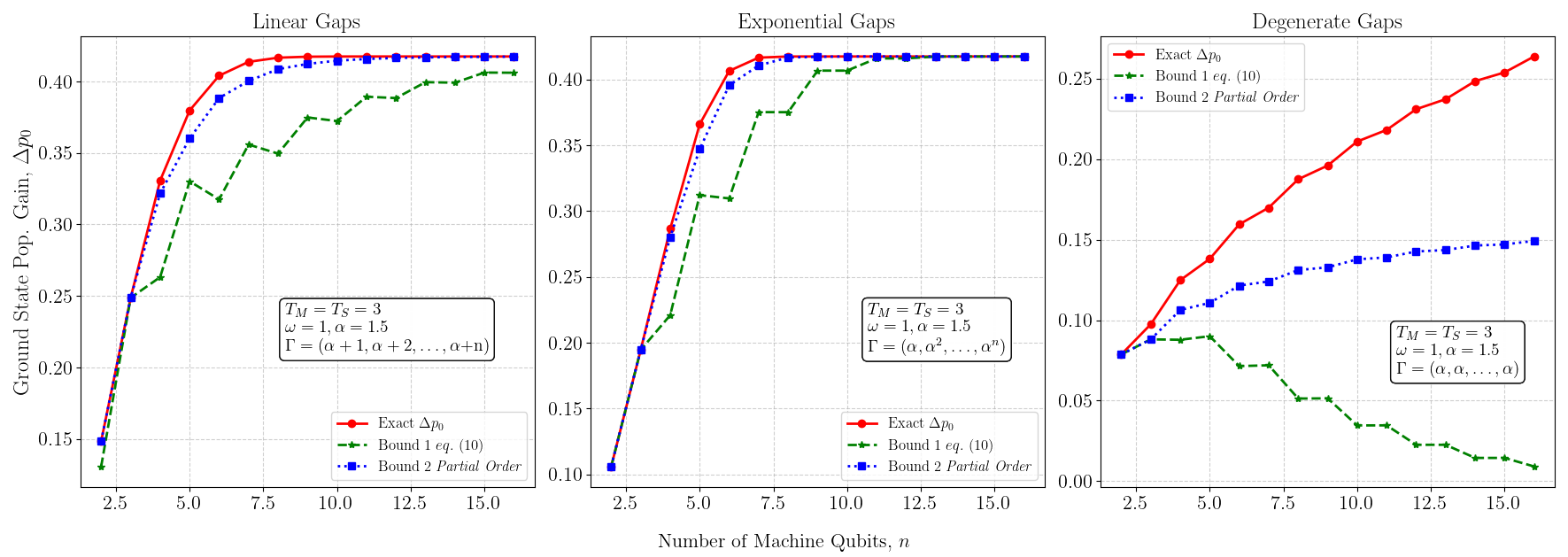}
\caption{In this figure we compare the change in ground-state population of $S$,  $\Delta p_0$, induced by the cooling protocol for different machine energy gap structures as the number of qubits in the machine $n$ increases. We also compare the lower bounds derived to approximate the change in population, in Green we see the performance of the bound given in Eq.~\eqref{eq:lower_bound_pop_change} and in Blue we see the bound calculated algorithmically using the partial order presented in Appendix~\ref{app:swappable_set}. We see that whilst the fixed bound Eq.~\eqref{eq:lower_bound_pop_change} performs well for machines with energetic structures that grow quickly it performs poorly for machines with fixed gaps. This is because Eq.~\eqref{eq:lower_bound_pop_change} captures a fixed set of strings to be exchanged which in the linear and exponential cases seem to capture most exchanges even as $n$ grows but in the degenerate case we see that this fixed set of strings contributes less and less to the growing number of strings which need to be exchanged as $n$ grows. In contrast, the Blue bound is adaptive constructed via the partial order and so performs better in all scenarios.}
\label{fig:pops}
\end{figure*}

\paragraph*{How much can a machine cool?} Armed with our bitwise notation and hierarchy of inequalities we can investigate how much a given machine can cool our system qubit. The figure of merit we are interested in is the optimal change in ground state population of $S$, $\Delta p_0 = p'_0 - p_0$, where $p'_0$ is the ground state population after cooling. Examining $\Delta p_0$ we begin by noticing that the ground state population after cooling can be expressed as a sum of populations of unchanged energy levels and swapped energy levels
\begin{align}
    p'_0 = \hspace{-0.2cm}\underbrace{\sum_{i_M \in \{0,1\}^{\times n} \, \setminus \, \mathbb{S}} \hspace{-0.1cm} \frac{e^{-\beta_M i_M \cdot \Gamma}}{{\mathcal{Z}_S\mathcal{Z}_{M}}}}_\text{unchanged} + \underbrace{\sum_{j_M \in \mathbb{S}} \frac{e^{-\beta_S \omega - \beta_M (j_M \oplus 1)\cdot \Gamma}}{\mathcal{Z}_S\mathcal{Z}_{M}}}_\text{swapped}, \label{eq:post_pop}
\end{align}
where $\mathbb{S}$ is the set of $n$-bit strings satisfying Eq.~\eqref{eq:ineq_simp}. The initial ground state population can be expressed as $p_0 = 1/\mathcal{Z}_S$ and so when subtracting it from $p_0'$ acts as $\mathcal{Z}_M$ when taking the common denominator $\mathcal{Z}_S \mathcal{Z}_M$. Then, it can be be noted that $\mathcal{Z}_M$ can be expanded to be written as a sum over all $n$-bit strings $\mathcal{Z}_M = \sum_{i_M \in \{0,1\}^{\times n}} e^{-\beta_M i_M \cdot \Gamma}$. From here, it can be seen that the $p'_0$ contribution in Eq.~\eqref{eq:post_pop} is nothing more than a sum of two disjoint sets whose union is $\{0,1\}^{\times n}$. Splitting $\mathcal{Z}_M$ into a contribution from $\mathbb{S}$ and its set complement the unchanged populations in Eq.~\eqref{eq:post_pop} cancel to give a sum over $\mathbb{S}$
\begin{align}
    \Delta p_0 &= \sum_{j_M \in \mathbb{S}} \frac{e^{-\beta_S \omega - \beta_M(E_\text{Max} - j_M \cdot \Gamma)} -  e^{-\beta_{M} j_{M} \cdot \Gamma}}{{\mathcal{Z}_S\mathcal{Z}_{M}}},
\end{align}
when simplified. This quantity can always be found if the set $\mathbb{S}$ is known, as discussed below for specific examples. If $\mathbb{S}$ is not completely known, but instead a subset of $\mathbb{S}$ is known, a lower bound can be at least be placed on $\Delta p_0$. For example, take the case that $T_M\omega/T_S < \gamma_1$, and then consider the bit-string $b_M = 10^{n-k-1}1^{k}$ where $k=n/2$ if $n$ is even and $k=(n-1)/2$ if $n$ is odd. Given the assumption that $\Gamma$ is ordered in non-decreasing order, this is a bit-string such that all bit-strings including and above this (in terms of the lexicographic order, cf. Appendix~\ref{app:swappable_set}) are guaranteed to satisfy Eq.~(\ref{eq:ineq_simp}) for all $\Gamma$. By considering performing the swaps associated to just these bit-strings, one finds that 
\begin{align}
        \Delta p_0 \geq \frac{1}{\mathcal{Z}_S \mathcal{Z}^{\Tilde{k}:n}_M}& \big( e^{- \beta_S \omega} - e^{-\beta_M E(\Tilde{k})} \big) ~+ \label{eq:lower_bound_pop_change}\\ 
        &\frac{1}{\mathcal{Z}_S \mathcal{Z}_M} \big( e^{-\beta_S E(\Tilde{k})} - e^{-\beta_S\omega - \beta_M(E_{\rm max} - E(\Tilde{k}))} \big),  \nonumber
\end{align}
where $E(\Tilde{k}) = \sum_{i = \Tilde{k}}^{n} \gamma_i$, with $\Tilde{k}= n/2 + 1$ if $n$ is even and $\Tilde{k}= (n + 3)/2$ if $n$ is odd, and $\mathcal{Z}^{\Tilde{k}:n}_M = \prod_{i=\Tilde{k}}^{n} (1 + e^{-\beta_M \gamma_i})$ is the partition function of the $n - \Tilde{k}$ machine qubits with the largest energy gaps (see Appendix~\ref{app:how_cold} for more details). From this example, we can further this line of thought by considering what is the largest subset of $\mathbb{S}$ that always satisfies Eq.~(\ref{eq:ineq_simp}). We find that this set can be characterized with the help of the lexicographic partial order of bit-strings (cf. Appendix~\ref{app:swappable_set}). Depending on the value of $\omega$ it consists of all the bit-strings larger than $a_M = 1^{2\ell-n}(01)^{n-\ell}$ when $\sum_{i=1}^{2(\ell-1)-n} \gamma_i \leq T_M\omega/T_S  < \sum_{i=1}^{2\ell-n} \gamma_{i}$ (with $\ell = \lceil \frac{n+1}{2}\rceil,\dots,n$). As opposed to the previous bound with $b_M$ we cannot obtain a formula as explicit as Eq.~\eqref{eq:lower_bound_pop_change} for the set composed of the elements larger than $a_M$. However, since the set is fully characterizable as a list of bit-strings, we compare both of the bounds in Fig.~\ref{fig:pops} (Green and Blue respectively) against the actual value of $\Delta p_0$ calculated numerically. 

\section{Cooling Unitaries : Reducibility \& Matching on Graphs}
\label{sec:unitary}
Having developed a framework to detect which populations must be exchanged in the joint system-machine state to optimally cool $S$, what remains to be considered is \textit{how}? As we've seen, optimally cooling $S$ given $M$ is simple; \textit{find the machine energy levels $i_M \in \{0,1\}^{\times n}$ satisfying $\frac{1}{2}(\frac{T_M}{T_S}\omega + E_\text{Max}) < i_M \cdot \Gamma$ and then perform a set of TLPs to exchange all $\ket{0_S i_M} \leftrightarrow \ket{1_S i_M \oplus 1}$}. However, denoting the sets of energy levels to be swapped as $\ket{A} = \{\ket{0_S i_M} \in \mathbb{S}\}$ and $\ket{B} = \{\ket{1_S i_M \oplus 1}\}$, it can be seen that we have the freedom to pair each $\ket{0_S i_M} \in \ket{A}$ with any $\ket{1_S j_M} \in \ket{B}$ and obtain the same increase in ground state population, $\Delta p_0$. There are therefore at least $|\mathbb{S}|!$ unitaries of the form $U_\text{cool} = \ketbra{A}{B} + \ketbra{B}{A} + \mathbb{1}_\text{Rest}$ that can cool $S$ to the same temperature.\footnote{This is because we have two lists of size $|\mathbb{S}|$. Starting with one element in $\ket{A}$ we have $|\mathbb{S}|$ choices to pair it with in $\ket{B}$, whereas the next element has $|\mathbb{S}|-1$ possibilities and so on giving $|\mathbb{S}|\times |\mathbb{S}|-1 \times |\mathbb{S}|-2 \times \dots \times 1 = |\mathbb{S}|!$ possible pairings in general. Other unitaries may be considered that cool optimally but also act on levels outside $\ket{A}$ without any effect on the cooling.}
But, these unitaries, whilst all valid for cooling, are far from equivalent in implementation or thermodynamic cost. In this section we will begin by examining the order of interaction necessary to carry out $U_\text{cool}$ for different machines and then move on to represent the problem of choosing a $U_\text{cool}$ under some cost constraints via a minimum weight perfect matching on bipartite graphs~\cite{godsil01}.

\subsection{Reducibility of Machines} 

Within the setting we have described it is perhaps natural to think that increasing the number of qubits in the machine is always helpful. And indeed it is, as increasing the number of qubits in the machine increases $E_\text{Max}$. However, one can ask whether it is only the increased energy that improves the cooling, or also the extra dimensions in the Hilbert space. In other words, we must investigate whether a cooling operation involving $S$ and an $n$-qubit machine $M$ is ever reducible to an operation involving just some $k < n$ coldest qubits of the machine. In which case, the energetic gap structure $\Gamma$ extended to $n$ qubits appears dimensionally \textit{wasteful} and could be reduced to a machine with less qubits involving only a subset of $\Gamma$. To begin this investigation, let's look at the most egregiously wasteful kind of machine, an $n$ qubit machine that is reducible to a bipartite interaction involving the coldest qubit in the machine and $S$. Such a machine can be said to be $(n-1)$-reducible, meaning that its warmest $n-1$ qubits do not participate in the cooling process. Mathematically, this means that $U_\text{cool}$ acts as the identity on these $n-1$ qubits leaving them as free indices. This would mean that the set of energy levels to be exchanged can be expressed $\mathbb{S}~=~\{\ket{0_S i^{n-1}1}~\,|~\, i^{n-1} \in \{0,1\}^{\times n -1}\}$ and so every $\ket{0_S i^{n-1}1}$ satisfies Eq.~\eqref{eq:ineq_simp}, such that
\begin{gather}
    \hspace{-0.25cm}\frac{1}{2}\left(\frac{T_M}{T_S}\omega + E_\text{Max}\right) < i^{n-1}1\cdot \Gamma \,\,\, \forall \, i^{n-1} \in \{0,1\}^{\times n-1}.
\end{gather}
Most stringently, this must also be valid for $i^{n-1} = 0^{n-1}$ where only the qubit with gap $\gamma_n$ contributes to cooling. For this string, the above condition becomes 
\begin{gather}
    \frac{1}{2}\left(\frac{T_M}{T_S}\omega + E_\text{Max}\right) < \gamma_n, \label{eq:n-1_red}
\end{gather}
that is \textit{a cooling scenario involving a machine whose largest gap is larger than half the sum of all the gaps in the machine and the constant $T_M/T_S\omega$ is reducible to a bipartite interaction} (see Appendix~\ref{app:reduc} for more details). This is an unexpected outcome where we see that if the functional form of $\Gamma$ grows steeply enough with $n$ then all that matters for cooling is the energetic gap of the largest qubit, and rest of the machine is \textit{useless}. Whilst surprising, we already encountered this situation in our introductory example in Sec.~\ref{sec:example} when only one of two qubits in the two qubit machine was involved in the cooling operation when $\omega < \gamma_2 - \gamma_1$. A quick calculation shows this is precisely the $(n-1)$-reducibility condition Eq.~\eqref{eq:n-1_red} for $T_M = T_S$ with a $2$ qubit machine. Whilst this condition is quite restrictive, it arises naturally in one of the examples we have considered, exponential machines. An $n$ qubit machine with an exponential gap structure is one whose gap structure grows exponentially with $n$ such that $\Gamma = (\gamma, \gamma^2,\dots, \gamma^n)$. Applying Eq.~\eqref{eq:n-1_red}, we find that a sufficient condition for an exponential machine to be ($n-1$)-reducible is $\gamma > (T_M/T_S)\omega$ with $\gamma > 1$ (cf. Appendix~\ref{app:n-1_red}).
\begin{figure}[h]
    \centering
\tikzset{every picture/.style={line width=0.75pt}} 
\scalebox{0.9}{
 \begin{tikzpicture}
 [x=0.75pt,y=0.75pt,yscale=-1,xscale=1]

\draw [color={rgb, 255:red, 74; green, 144; blue, 226 }  ,draw opacity=1 ]   (284.98,154) -- (296.23,154)(284.98,157) -- (296.23,157) ;
\draw [shift={(304.23,155.5)}, rotate = 180] [color={rgb, 255:red, 74; green, 144; blue, 226 }  ,draw opacity=1 ][line width=0.75]    (10.93,-3.29) .. controls (6.95,-1.4) and (3.31,-0.3) .. (0,0) .. controls (3.31,0.3) and (6.95,1.4) .. (10.93,3.29)   ;
\draw [color={rgb, 255:red, 65; green, 117; blue, 5 }  ,draw opacity=1 ]   (366.23,153.85) -- (381.39,153.77)(366.24,156.85) -- (381.41,156.77) ;
\draw [shift={(389.4,155.23)}, rotate = 179.72] [color={rgb, 255:red, 65; green, 117; blue, 5 }  ,draw opacity=1 ][line width=0.75]    (10.93,-3.29) .. controls (6.95,-1.4) and (3.31,-0.3) .. (0,0) .. controls (3.31,0.3) and (6.95,1.4) .. (10.93,3.29)   ;
\draw [color={rgb, 255:red, 74; green, 144; blue, 226 }  ,draw opacity=1 ]   (196.98,154.16) -- (214.98,154.11)(196.99,157.16) -- (214.99,157.11) ;
\draw [shift={(222.98,155.59)}, rotate = 179.84] [color={rgb, 255:red, 74; green, 144; blue, 226 }  ,draw opacity=1 ][line width=0.75]    (10.93,-3.29) .. controls (6.95,-1.4) and (3.31,-0.3) .. (0,0) .. controls (3.31,0.3) and (6.95,1.4) .. (10.93,3.29)   ;

\draw [color={rgb, 255:red, 74; green, 144; blue, 226 }  ,draw opacity=1 ]   (205.84,190.95) -- (192.66,174.89)(208.16,189.05) -- (194.98,172.98) ;
\draw [shift={(188.75,167.75)}, rotate = 50.64] [color={rgb, 255:red, 74; green, 144; blue, 226 }  ,draw opacity=1 ][line width=0.75]    (10.93,-3.29) .. controls (6.95,-1.4) and (3.31,-0.3) .. (0,0) .. controls (3.31,0.3) and (6.95,1.4) .. (10.93,3.29)   ;
\draw [color={rgb, 255:red, 74; green, 144; blue, 226 }  ,draw opacity=1 ]   (224.84,188.05) -- (237.77,172.24)(227.16,189.95) -- (240.1,174.14) ;
\draw [shift={(244,167)}, rotate = 129.29] [color={rgb, 255:red, 74; green, 144; blue, 226 }  ,draw opacity=1 ][line width=0.75]    (10.93,-3.29) .. controls (6.95,-1.4) and (3.31,-0.3) .. (0,0) .. controls (3.31,0.3) and (6.95,1.4) .. (10.93,3.29)   ;
\draw [color={rgb, 255:red, 65; green, 117; blue, 5 }  ,draw opacity=1 ]   (340.71,191.27) -- (330.54,174.14)(343.29,189.73) -- (333.12,172.61) ;
\draw [shift={(327.75,166.5)}, rotate = 59.3] [color={rgb, 255:red, 65; green, 117; blue, 5 }  ,draw opacity=1 ][line width=0.75]    (10.93,-3.29) .. controls (6.95,-1.4) and (3.31,-0.3) .. (0,0) .. controls (3.31,0.3) and (6.95,1.4) .. (10.93,3.29)   ;
\draw [color={rgb, 255:red, 65; green, 117; blue, 5 }  ,draw opacity=1 ]   (369.32,189.33) -- (383.14,171.63)(371.68,191.17) -- (385.51,173.48) ;
\draw [shift={(389.25,166.25)}, rotate = 128] [color={rgb, 255:red, 65; green, 117; blue, 5 }  ,draw opacity=1 ][line width=0.75]    (10.93,-3.29) .. controls (6.95,-1.4) and (3.31,-0.3) .. (0,0) .. controls (3.31,0.3) and (6.95,1.4) .. (10.93,3.29)   ;
\draw    (270.73,229.85) -- (255.59,215.79)(272.77,227.65) -- (257.63,213.59) ;
\draw [shift={(250.75,209.25)}, rotate = 42.88] [color={rgb, 255:red, 0; green, 0; blue, 0 }  ][line width=0.75]    (10.93,-3.29) .. controls (6.95,-1.4) and (3.31,-0.3) .. (0,0) .. controls (3.31,0.3) and (6.95,1.4) .. (10.93,3.29)   ;
\draw    (311,227.88) -- (325.53,214.95)(313,230.12) -- (327.52,217.19) ;
\draw [shift={(332.5,210.75)}, rotate = 138.32] [color={rgb, 255:red, 0; green, 0; blue, 0 }  ][line width=0.75]    (10.93,-3.29) .. controls (6.95,-1.4) and (3.31,-0.3) .. (0,0) .. controls (3.31,0.3) and (6.95,1.4) .. (10.93,3.29)   ;
\draw [color={rgb, 255:red, 0; green, 0; blue, 0 }  ,draw opacity=1 ][line width=1.5]  [dash pattern={on 1.69pt off 2.76pt}]  (205.67,140.67) -- (205.5,169.75) ;
\draw [line width=1.5]  [dash pattern={on 1.69pt off 2.76pt}]  (376.25,140.75) -- (376.5,170) ;

\draw (291.73,239) node    {$|0_{S} i_{1} i_{2} 1\rangle $};
\draw (223.73,198.5) node  [color={rgb, 255:red, 74; green, 144; blue, 226 }  ,opacity=1 ]  {$|0_{S} i_{1} 01\rangle $};
\draw (355.23,198.5) node  [color={rgb, 255:red, 65; green, 117; blue, 5 }  ,opacity=1 ]  {$|0_{S} i_{1} 11\rangle $};
\draw (253.98,155.5) node  [color={rgb, 255:red, 74; green, 144; blue, 226 }  ,opacity=1 ]  {$|0_{S} 10 1\rangle $};
\draw (335.23,155.5) node  [color={rgb, 255:red, 65; green, 117; blue, 5 }  ,opacity=1 ]  {$|0_{S} 011\rangle $};
\draw (420.4,155.08) node  [color={rgb, 255:red, 65; green, 117; blue, 5 }  ,opacity=1 ]  {$|0_{S} 111\rangle $};
\draw (165.98,155.75) node  [color={rgb, 255:red, 74; green, 144; blue, 226 }  ,opacity=1 ]  {$|0_{S} 001\rangle $};
\draw (391,235) node [anchor=north west][inner sep=0.75pt]   [align=left] {2 - reducible};
\draw (389.5,193.25) node [anchor=north west][inner sep=0.75pt]   [align=left] {1 - reducible};
\draw (454.5,146) node [anchor=north west][inner sep=0.75pt]   [align=left] {irreducible};
\draw (199.17,124.83) node [anchor=north west][inner sep=0.75pt]   [align=left] {{\footnotesize (i)}};
\draw (368.33,126) node [anchor=north west][inner sep=0.75pt]   [align=left] {{\footnotesize (ii)}};
\draw [color={rgb, 255:red, 74; green, 144; blue, 226 }  ,draw opacity=1 ]   (256.23,197) -- (314.73,197)(256.23,200) -- (314.73,200) ;
\draw [shift={(322.73,198.5)}, rotate = 180] [color={rgb, 255:red, 74; green, 144; blue, 226 }  ,draw opacity=1 ][line width=0.75]    (10.93,-3.29) .. controls (6.95,-1.4) and (3.31,-0.3) .. (0,0) .. controls (3.31,0.3) and (6.95,1.4) .. (10.93,3.29)   ;
\end{tikzpicture}
}
    \caption{Diagram of reducibility inequalities for a 3-qubit machine showing the partial order in the inequalities, where $\ket{0_S i_M} \implies \ket{0_S j_M}$ means that if $\ket{0_S i_M} \in \mathbb{S}$, then so is $\ket{0_S j_M} \in \mathbb{S}$. Although, the reverse implication does not necessarily hold. }
    \label{fig:3_qubi_ineq}
\end{figure}
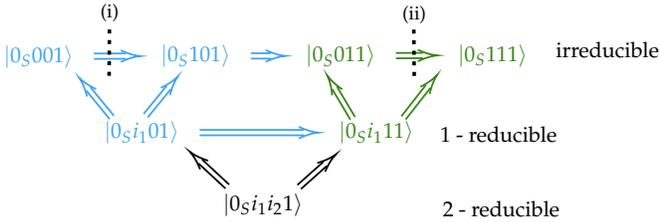

Following this line of reasoning, one can find conditions for a machine to be $k$ - reducible by examining under which conditions all the energy levels $\ket{0_S i^k 1^{n-k}}$ need to be exchanged such that the $k$ bits in $i^k$ are free indices, so that the unitary acts as the identity on the $k$ qubits corresponding to these $k$ bits. This situation is not straightforward, as one must consider $\ket{0_S i^k l^{n-k}}$ for every possible choice of $l^{n-k} \in \{0,1\}^{\times n-k}$. Consider Fig~\ref{fig:3_qubi_ineq}, which visualises the structure of the partially ordered set of inequalities for a 3-qubit machine. This poset has 3 levels corresponding to the machine being 2 - reducible, meaning the optimal cooling interaction involves only the coldest qubit in the machine, 1-reducible, involving the coldest two qubits in $M$, and irreducible, involving all the qubits in $M$. If the machine is 2-reducible, that is $\ket{0_S i_1 i_2 1}$ satisfies Eq.~\eqref{eq:ineq_simp} for any $i_1,i_2 \in \{0,1\}$, then every energy level in the hierarchy is also a member of $\mathbb{S}$ and must be exchanged. Ascended one level of the hierarchy we find two scenarios, only one of which leads to 1-reducibility. It must be the case that $\ket{0_S i_1 11} \in \mathbb{S}$ whilst $\ket{0_S i_1 01} \notin \mathbb{S}$ for $i_1 \in \{0,1\}$, preventing 2-reducibility. If $\ket{0_S i_1 01} \in \mathbb{S}$ then we have $\ket{0_S i_1 11} \in \mathbb{S}$ also, leading again to 2-reducibility. At the final level, two scenarios lead to irreducibility. Firstly in situation (i) in Fig~\ref{fig:3_qubi_ineq} we have $\ket{0_S001} \notin \mathbb{S}$ and $\ket{0_S101} \in \mathbb{S}$ implying that $\ket{0_S011}, \ket{0_S111} \in \mathbb{S}$ also. Whilst this gives $\ket{0_S i_1 11} \in \mathbb{S}$ for $i_1 \in \{0,1\}$, we have $\ket{0_S101} \in \mathbb{S}$ but $\ket{0_S 001} \notin \mathbb{S}$ and so any exchange involving $\ket{0_S101}$ is 4-partite. Secondly, the machine is irreducible if $\ket{0_S011} \notin \mathbb{S}$ and $\ket{0_S111} \in \mathbb{S}$, condition (ii) in Fig~\ref{fig:3_qubi_ineq}. The machine must be irreducible in this case as any exchange involving $\ket{0_S111} \in \mathbb{S}$ cannot involve $\ket{0_S011} \notin \mathbb{S}$ i.e. $i_1$ cannot be a free index. Any TLP that cools scenario (ii) acting on $\ket{0_S111}$ necessarily acts on all qubits in the cooling scenario. This begins to elucidate the complex hierarchical structure of the energetic inequalities we have explored in this manuscript, in Appendix~\ref{app:reduc} we give a visualisation for the 4 qubit machine hierarchy of inequalities and formal definitions for $k$-reducibility in general. The general structure of the partial order of these inequalities is explored in Appendix~\ref{app:swappable_set}.

\paragraph*{Irreducible Machines} While general $k$-reducibility conditions are difficult to obtain, we are able to fully characterise scenarios under which a machine is irreducible. This is arguably the most important reducibility criterion because it teaches us how to design \textit{useful} machines where adding more qubits always leads to more potential for cooling. Such machines exemplify the 3rd Law of Thermodynamics~\cite{taranto_23} as increase in the number of constituents of the machine and so its dimension leads to further cooling.\\
An irreducible machine must use every qubit for cooling, even its warmest qubit. This implies that $i_1$ cannot be a free index in $U_\text{cool}$, meaning that there must exists at least one $l^{n-1} \in \{0,1\}^{\times n - 1}$ such that $\ket{0_S 1 l^{n-1}} \in \mathbb{S}$ whilst $\ket{0_S 0 l^{n-1}} \notin \mathbb{S}$. If not true, then the joint action of all these permutations renders the first qubit of the machine useless, making the machine reducible. The first condition $\ket{0_S 1 l^{n-1}} \in \mathbb{S}$ imposes the upper bound $\frac{1}{2}(\frac{T_M}{T_S}\omega + E_\text{Max}) < 1l^{n-1} \cdot \Gamma$, whilst the second condition $\ket{0_S 0 l^{n-1}} \notin \mathbb{S}$ imposes the lower bound $0l^{n-1} \cdot \Gamma \leq \frac{1}{2}(\frac{T_M}{T_S}\omega + E_\text{Max})$. Taken together, we come to the condition that \textit{an $n$ qubit machine with gap structure $\Gamma$ at temperature $T_M$ acts irreducibly when cooling a qubit at temperature $T_S$ with gap $\omega$ if there exists at least one $l^{n-1} \in \{0,1\}^{\times n - 1}$ such that}
\begin{gather}
0 \leq  \frac{1}{2}\left(\frac{T_M}{T_S}\omega + E_\text{Max}\right) - l^{n-1} \cdot \Gamma_{2:n} < \gamma_1, \label{eq:irreducibility}
\end{gather}
where $\Gamma_{k:n} = (\gamma_k, \gamma_{k+1}, \ldots, \gamma_n)$ such that $E(l^{n-1}) = l^{n-1} \cdot \Gamma_{2:n}$ can be considered an energy of the $n-1$ coldest qubits. Intuitively, one can understand this inequality to mean that there is at least one energy level $l^{n-1}$ of the $n-1$ coldest machine qubits which is only large enough to optimally cool if $\gamma_1$ is included -- exchanging $l^{n-1}$ without a contribution from the warmest qubit would not cool $S$. It is worth noting that is clear from eq.(\ref{eq:irreducibility}) that by investing work and effectively cooling $S$ by increasing $\omega$ and keeping $\Gamma$ constant, one can make any irreducible machine reducible by trading unitary complexity for thermodynamic work. We investigate this inequality in further depth in Appendix~\ref{app:reduc}.

\paragraph*{Degenerate Machines are always irreducible.} A first application of Eq.~\eqref{eq:irreducibility} is to show when a machine of identical qubits, as typically considered in algorithmic and dynamical cooling scenarios~\cite{schulman_vazirani,schulman_limits,karen_dynamical_2011,bassman_campisi_24}, is irreducible. For a machine of $n$ qubits at temperature $T_M$, each with gap $\gamma$, note that the energy of any level $i_M$ in the machine can be expressed as $E(i_M) = t\gamma$ where $t = \#(i_M)$ is the number of 1s in $i_M$. The irreducibility condition Eq.~\eqref{eq:irreducibility} then takes the form 
\begin{gather}
    -(n-2t) \gamma \leq \frac{T_M}{T_S}\omega \leq 2\gamma - (n-2t)\gamma, \label{eq:deg_machine_irreducible}
\end{gather}
where $E_\text{Max} = n \gamma$ and $t \in [0,n-1]$. Each $t$ gives a condition that is some interval on the real-numbers into which $(T_M/T_S) \omega$ must fall. In the degenerate case, these intervals capture the whole number-line $[0, n \gamma]$, meaning whenever $(T_M/T_S) \omega \leq n \gamma$, one can always find a $t$ such that the above condition is satisfied $\square$.  Under the reasonable assumptions that $T_M \leq T_S$ and $\omega \leq \gamma$, degenerate machines are therefore always irreducible. See Appendix~\ref{app:degen_reduc} for more details.
\subsection{Cost Optimised Cooling Unitaries from Minimum Weight Perfect Matchings}
In the prior section, we considered what order of interaction was necessary for optimal cooling given the set $\mathbb{S}$. Here, we present a framework to identify which of the $\vert \mathbb{S} \vert !$ unitaries of the form $U_\text{cool} = \ketbra{A}{B} + \ketbra{B}{A} + \mathbb{1}_\text{Rest}$ should be used to implement the cooling protocol given some additional constraint. An agent may wish to optimise different costs of their cooling protocol such as heat dissipated into their machine, work cost invested to implement $U_\text{cool}$ or the gate complexity of the protocol. While we focus on gate complexity, our methods can be applied to any such cost function. A natural way to characterise implementations of $U_\text{cool}$ is in terms of a decomposition of the unitary into a product of quantum gates. To do so, one needs to consider a gate set into which a given unitary can be decomposed. We focus on the set $\mathcal{G} = \{\mathtt{X}, \mathtt{CNOT}, \mathtt{Toffoli}, \mathtt{C^{3}X}, \dots, \mathtt{C^{n-1}X}\}$, which suffices since the unitaries we are interested in generating are permutations of the energy eigenbasis which w.l.o.g were chosen to be in the $\mathtt Z$ basis. A figure of merit for the complexity of $U_{\rm cool}$, related to the classical notion of time complexity, is then the \textit{gate complexity}. This the minimum number of gates from $\mathcal{G}$ that are required to implement $U_{\rm cool}$.

\paragraph*{Hamming Weight lower bounds Gate Complexity} Let $S_{j,k}$ denote the TLP unitary that exchanges the $n$ qubit eigenket $\ket{j}$ for the eigenket $\ket{k}$ expressed as $S_{j,k} = \ketbra{j}{k} + \ketbra{k}{j} + \mathbb{1}_{\rm Rest}$. For the gate set $\mathcal{G}$, we can show that the minimum number of gates required to implement $S_{j,k}$ using $\mathcal{G}$, denoted $\mathcal{C}_{\mathcal G}(S_{j,k})$, is lower bounded by
\begin{equation}
    \mathcal{C}_{\mathcal G}(S_{j,k}) \geq  D_H(j,k), \label{eq:lower_bound_gate_complexity}
\end{equation}
where $D_{H}(j,k)$ is the Hamming distance between the bit-strings $j$ and $k$. To prove this statement, first note that \hbox{$D(\ket{i}, U_{g} \ket{i}) \leq 1 ~ \forall ~U_{g} \in \mathcal{G}$} and $\ket{i}$, where $D(\ket{i}, \ket{j}) = D(i,j)$. This is due to each element of $\mathcal{G}$ only changing the Hamming weight of a basis state, that is the number of $1$s in its identifying bit-string, by at most one. This then leads to a change in distance between the input and output bit string of at most one. Consider now some $S_{j,k}$, where $D(j,k) = l$. This can be decomposed as $S_{j,k} = U_1 U_2 \ldots U_l$ where $U_i \in \mathcal{G}~\forall~i$. Given that $S_{j,k}\ket{j}=\ket{k}$, it must be the case that $l\geq D(j,k)$, as each $U_i$ can only change the distance by at most one and since $l = \mathcal{C}_{\mathcal G}(S_{j,k})$ by definition, the statement holds. $\square$

\begin{figure}[t]
    \centering
 \begin{tikzpicture}[scale=1, baseline=(current bounding box.center)]
    \def\xleft{0}
    \def\xright{3}

    \node[draw, circle, fill=lred] (A) at (\xleft, 3) {$0011$};
    \node[draw, circle, fill=lred] (B) at (\xleft, 2) {$0111$};
    \node[draw, circle, fill=lred] (C) at (\xleft, 1) {$0101$};
    \node[draw, circle, fill=lred] (E) at (\xright, 3) {$1010$};
    \node[draw, circle, fill=lred] (F) at (\xright, 2) {$1000$};
    \node[draw, circle, fill=lred] (G) at (\xright, 1) {$1100$};
    \foreach \leftnode in {A,B,C} {
        \foreach \rightnode in {E,F,G} {
            \draw[thick] (\leftnode) -- (\rightnode);
        }
    };
    \draw[orange,thick] (A) -- (E);
    \draw[orange,thick] (B) -- (F);
    \draw[orange,thick] (C) -- (G);
    \end{tikzpicture}
    \caption{An example of a complete bipartite graph $G(K,K\oplus 1, E)$ for a cooling scenario involving a 3-qubit machine as detailed in Appendix~\ref{app:example_2} (case 4). Each perfect matching in this bipartite graph corresponds to a choice of $U_\text{cool}$ that will lead to the same final ground state population from the system. The minimum weight perfect matching on this graph (orange) corresponds to the unitary $U^*_\text{cool}$ that minimises the Hamming weight across pairings as seen using the cost matrix Eq.~\eqref{eq:cost_mat}.}
    \label{fig:mwpm}
\end{figure}
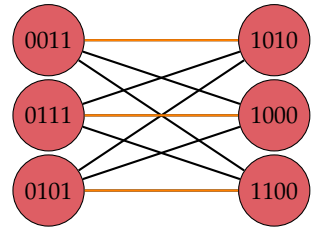
\paragraph*{Minimum Weight Perfect Matching} The above result implies that one could find a unitary that both optimally cools and potentially minimises the gate complexity of $U_{\rm cool}$ by choosing a set of TLPs that satisfy the energetic inequalities but also minimise the Hamming distance between pairs $\ket{0_S i_M} \in \ket{A}$ and $\ket{1_S j_M} \in \ket{B}$. We are therefore interested in finding the unitary $U^*_\text{cool}$ such that,
\begin{gather}
U^*_\text{cool} = \sum_{(\ket{a},\ket{b}) \in \mathcal{M}} \ketbra{a}{b} + \ketbra{b}{a} +\mathbb{1}_\text{Rest} \, \text{with} \\ \mathcal{M} = \bigg\{ (\ket{a}, \ket{b}) \in \ket{A} \times \ket{B} \bigg| \, \min_{\substack{\ket{a} \in \ket{A} \\ \ket{b} \in \ket{B}}} \sum_{\substack{\ket{a} \in \ket{A} \\ \ket{b} \in \ket{B}}} d_H(a,b) \bigg\}, \nonumber
\end{gather}
where $\mathcal{M}$ is a set of 2-tuples or pairs of bit strings labelling energy levels whose Hamming distances sum to give the minimum over all possible pairings. Note, this pairing is the opposite of the strategy we employed to \textit{identify} the smallest set of TLPs necessary for optimal cooling. There, we found the TLPs that maximised the Hamming weight distance. This emphasises the point that once we have identified the necessary TLP using that strategy, we are free to implement any of the $|\mathbb{S}|!$ alterative pairings which all achieve optimal cooling but perform differently relative to different costs.  

To find $U_{\rm cool}^{*}$, consider the complete bipartite graph $G(K,K\oplus 1, E)$, where each node on the left subgraph $K$ represents a bit-string in $\mathbb{S}$, and each node on the right subgraph $K\oplus1$ represents a bit-string in $\mathbb{S} \oplus 1$. Each node in $K$ is then connected to every other node in $K\oplus 1$ such that an edge represent a potential energy level exchange or TLP. See Appendix~\ref{app:hypercube_graph} for full details on constructing this graph. A matching $M$ on such a graph is a collection of edges of the graph where each vertex which intersects with the collection intersects precisely with one edge, a matching is said to be perfect if each vertex in a graph intersects with it. That is, a perfect matching connects every vertex in a graph to exactly one other vertex in the graph. We visualise such a complete bipartite graph for $\vert \mathbb{S} \vert = 3$ in Fig.~\ref{fig:mwpm} for the example of a $3$-qubit machine, as detailed in Appendix~\ref{app:example_2} (case 4). It can be shown that by definition a perfect matching on $G(K,K\oplus 1, E)$ corresponds to a suitable $U_\text{cool}$ as it pairs each disordered ground state energy level $\ket{0_S i_M} \in \mathbb{S}$ to exactly one $\ket{1_S j_M}$ which need be exchanged for cooling, exhausting all of $\mathbb{S}$. See Appendix~\ref{app:hypercube_graph} for the proof. If we add weights to the edges of a graph then a \textit{minimum weight perfect matching} (MWPM) is the perfect matching with pairings that minimise the sum of the edge weights in the perfect matching. By weighting each edge in $G(K,K\oplus 1, E)$ with the Hamming distance between the bit-strings labelling the nodes which intersect it, we find that \textit{a minimum weight perfect matching on the graph corresponds to a $U^*_\text{cool}$}. In achieving this correspondence we reformulate the task of finding an optimal cooling unitary subject to some additional constraint as that of finding minimum weight perfect matchings on complete bipartite graphs. A detailed proof of this correspondence is given in Appendix~\ref{app:hypercube_graph}.

Finding MWPMs is possible through the \textit{Hungarian Algorithm}~\cite{comb_opt}. The algorithm begins by assigning a cost to every possible pairing and creating a matrix $C$ of these costs. In this case, this matrix can be created by assigning $a_j \in \ket{A}$ to the columns of the matrix and $b_i \in \ket{B}$ to the rows of the matrix and computing the entries of the matrix as $c_{ij} = D_H(b_i,a_j)$. For the example in Fig.\ref{fig:mwpm}, this gives
\begin{gather}
C =
\begin{pNiceMatrix}[first-row,first-col]
    & \ket{0011} & \ket{0101} & \ket{0111}       \\
\ket{1100} & 4   & \textcolor{lred}{2}     & 3   \\
\ket{1010} & \textcolor{lred}{2}   & 4     & 3   \\
\ket{1000} & 3  & 3     & \textcolor{lred}{4}   \\
\end{pNiceMatrix}.
\label{eq:cost_mat}
\end{gather}
From C, a minimum perfect weight matching corresponds to picking $3$ entries from this matrix using every row and column once such that the sum of these $3$ entries is as low as possible. Here, this gives the unique choice of the cooling unitary that might also minimise gate complexity as $U^*_\text{cool} = \ketbra{0101}{1100} + \ketbra{1100}{0101} + \ketbra{0011}{1010} + \ketbra{1010}{0011} + \ketbra{0111}{1000} + \ketbra{1000}{0111} + \mathbb{1}_\text{Rest}$. More generally, the MWPM on $G(K,K\oplus1, E, d_H)$ will not be unique. While in this section we have focused on optimising the Hamming distance between energy levels being exchanged, as a proxy for gate complexity, there is nothing which prevent us from optimising other constraints by simply choosing another cost function $C(v_i,v_j)$ such as work or dissipated heat with which to weight the edges intersecting nodes $v_i \in K, v_j \in K \oplus1$. Indeed, adding multiple constraints e.g. Hamming distance and energetic cost can lead to a unique choice of cooling unitary, as we discuss in Appendix~\ref{app:hypercube_graph}. This offers a new way of designing cooling unitaries in quantum thermodynamics which both gives insights; into the design of circuits for cooling due to its decomposed structure considering pairs of energy levels to be exchanged at a time, and to the optimisation of these circuits relative to a cost function of interest.
\subsection{Connections to Other Work \& Applications}\textit{Algorithmic cooling} (AC)~\cite{schulman_vazirani, Park2016, naye_comparison} is a collection of techniques to cool a single thermal qubit with access to $n$ identical thermal qubits and the ability to refresh these qubits with a thermal bath. As such, most AC protocols are divided into two parts; 1. \textit{entropy compression}: the process of unitarily ordering the joint-system machine populations in decreasing order and 2. \textit{relaxation}: where the machine qubits are reset to their initial temperature with access to a thermal bath. These two steps are repeated for many rounds, cooling the target system. In this context, our work gives a general framework for the analysis of entropy compression step. Namely, we give three new insights to this part of the procedure: (i) techniques used for entropy compression such as the partner-pair algorithm~\cite{schulman_vazirani,Park2016} are heuristic and non-constructive whilst our inequalities and minimum weight perfect matching methods are constructive. (ii) The heuristic methods used in algorithmic cooling do not cater for optimisation relative to costs e.g. work cost, gate complexity whilst the MWPM method provided can be used to optimise cooling unitaries for any cost function. (iii) Our setting generalises the scenario from $n$ identical thermal qubits to $n$ thermal qubits with arbitrary energetic structure where the lexicographic order heuristically used in the partner-pair algorithm is no longer enough and the partial order discussed in this work must be considered.

\textit{Quantum Error Correction}--Several works in this field have pointed out a connection between cooling and error correction~\cite{temme2015faststabilizerhamiltoniansthermalize,Reiter2017,stab_ham,deNeeve2022,dissipative_stab,Shtanko2025boundsautonomous}. Despite experimentally implementing cooling protocols for error correction~\cite{Reiter2017,deNeeve2022} and theoretical works in dissipative quantum error correction investigate how cooling to the ground state of the Hamiltonian given by the negative sum of the stabilisers of the code~\cite{temme2015faststabilizerhamiltoniansthermalize,stab_ham,dissipative_stab} naturally restores the system to the codespace-- a connection with the state of the art of quantum thermodynamics theory is lacking. In Appendix~\ref{Ap:quantum_error_correction} we explore how the techniques presented in our work can be applied to this scenario by rephrasing the three qubit error correcting code as a cooling problem, setting the stage for future work.

\textit{Autonomous Thermal Machines}-- Much of the literature in quantum thermodynamics in the past decade has focused on cooling via the \textit{virtual qubit subspace swap}~\cite{skrzypczyk_10,brunner_12,ralph_swap,clivaz_prl_2019,Mitchison_review} where the $\ket{0_S 1_M} \leftrightarrow \ket{1_S 0_M}$ exchange is carried out. This exchange carries the most powerful population transfer which in cooling strategies using refreshed machines, asymptotically determines how much one can cool~\cite{brunner_12,clivaz_prl_2019}. In this work we have clarified using our inequalities, how many other population exchanges $\ket{0 \, i_M}, \ket{1 \, i_M \oplus 1}$ with $E(i_M)$ satisfying our main inequality, play a role in a cooling scenario. This opens the door for more efficient cooling protocols using a small finite number of refreshed machines rather than asymptotically many. Embedding the virtual qubit subspace swap in Lindbladian dynamics where it is the effective steady state interaction has also inspired autonomous thermal machines used for refrigeration~\cite{skrzypczyk_10}, timekeeping~\cite{qtm_timekeeping} and thermometry~\cite{qtm_thermometer}. Recent works~\cite{barosik_sciadv, xuereb2025} evem explore using such an exchange towards a computational goal. Since our work identifies cooling operations beyond the virtual qubit subspace swap, we anticipate that these energy level permutations can also be embedded in Lindbladian dynamics leading to more intricate autonomous thermal machines which can perform more complex tasks than those based around the virtual qubit subspace swap.

\textit{Role of Quantum Coherence \& Purification}--In this research we chose to focus on the simplistic setup of machines consisting of thermal qubits in a separable state, allowing in principle only classical correlations to play a role in cooling. An interesting direction to further this work would be to generalise to more complicated machines featuring higher-dimensional systems such as qudits, resonators, many-body systems and, importantly, quantum correlations~\cite{taranto_23,PhysRevX.5.041011,yada2024measuringmultipartitequantumcorrelations,PhysRevLett.132.140402}. In this direction, in Appendix~\ref{appendix: Generalisation to Arbitrary States}, we show how our framework and techniques can be adapted to the task of amplifying the purity of an arbitrary qubit state with access to $n$ other arbitrary qubit states.

\textit{Cooling \& Preserving Symmetries}--In recent years, the topic of thermalisation of quantum systems which preserve non-Abelian symmetries or charges has been a prominent point of inquiry~\cite{Majidy2023,Majidy2024}. Considering these works provokes one to ponder how preserving symmetries impacts an agent's ability to cool a quantum system within the framework we have presented. This question was pursued in~\cite{silva2024optimalunitarytrajectoriescommuting} where it was observed that respecting symmetries restrists that set of unitaries which can cool a quantum system. In Appendix~\ref{Ap:cooling_w_symmetries} we investigate this question within our framework for energy and total spin conservation and similarly find that the set of cooling unitaries reduces under these constraints.

\section{Conclusion}
In this work, we took the folkloric idea that cooling a quantum system given access to another consists of only reordering the eigenvalues of the joint state in descending order--and examined the structure of this sorting in detail. We found that the intricacies of this rearrangement is determined by a simple set of inequalities which depend on macroscopic properties of the cooling scenario. Namely the global temperature of the machine, its total energy $E_\text{Max}$ and temperature and energy of the system, elucidating the physics behind the heuristic sorting approach to cooling. These inequalities also allow an agent to cool the system without considering the whole system-machine energy density matrix by simply checking whether a given energy level of the machine is above the threshold value $\frac{1}{2}(T_M\omega/T_S + E_\text{Max})$. 

We were also able to show that these inequalities are sensitive to what number of qubits in the machine are contributing to an optimal cooling interaction with $S$. Here, we found that machines whose gaps increase exponentially are generally fully reducible to an interaction involving only their coldest qubit ($n-1$-reducible). And, we were able to verify that degenerate machines are generally irreducible, that is all their qubits contribute to optimally cooling $S$. This irreducibility criterion should prove especially helpful when designing quantum thermal machines for cooling quantum systems as it allows one to ensure that the whole machine they have constructed will be useful at the cooling task at hand. This is especially relevant as autonomous thermal machines are beginning to be experimentally realised~\cite{fridge_16,Aamir2025} with the hope of being deployed to cool quantum systems. 

Lastly, we provided a systematic way of decomposing cooling unitaries into a set of two-level permutations. This led us to provide a novel graph representation for cooling scenarios which is explored further in the appendices. Importantly, this graph representation inspired us to consider the optimisation of cooling unitaries relative to any cost function of interest as a perfect matching problem on bipartite graphs. In doing so we have further solidified the ability to represent quantum cooling protocols as quantum circuits~\cite{bassman_campisi_24,bassmann} and developed a new tool for their optimisation. 

The study of the theoretical limits to cooling a system has been central to every era of thermodynamic inquiry in the past centuries, be that the Carnot cycle for heat engines~\cite{Carnot1824} or the Doppler limit for laser cooling~\cite{laser_cooling}. We are now in the qubit era, forced to imagine if a quantum thermodynamic system could be fully controlled and acted upon using a quantum computer what the limits and methods of cooling would be in this scenario. Here we have offered a further step towards contending with the challenge our era of thermodynamic inquiry faces.

\label{sec:discussion}

\paragraph*{Acknowledgements}
\paragraph*{Acknowledgements}
J.X. acknowledges Josef Lauri who taught graph theory at the University of Malta and Xandru Mifsud \& Adriana Baldacchino for their hospitality and pointing out that the graphs in earlier versions of this work were hypercube graphs. The authors thank Nayeli Briones Rodriguez for enlightening discussions on the state of the art in algorithmic cooling and Paul Skrzypczyk, Chung-Yun Hsieh \& Fabien Clivaz for insightful discussions. The authors acknowledge TU Wien Bibliothek for financial support through its Open Access Funding Programme. B.S.~acknowledges support from UK EPSRC (EP/SO23607/1). J.X., M.H. and P.B. acknowledge funding from the European Research Council (Consolidator grant ‘Cocoquest’ 101043705) and financial support from the Austrian Federal Ministry of Education, Science and Research via the Austrian Research Promotion Agency (FFG) through the project FO999914030 (MUSIQ) funded by the European Union – NextGenerationEU.
A.R. acknowledges support from the the Swiss National Science Foundation for funding through Postdoc.Mobility (Grant No. P500PT225461). The code for generating the presented plots can be found at~\cite{code}.

\providecommand{\noopsort}[1]{}\providecommand{\singleletter}[1]{#1}%

\clearpage

\onecolumngrid
\appendix
\begin{center}
    \Large \bfseries Technical Matter
\end{center}
\vspace{1em}
\let\addcontentsline\oldaddcontentsline %

\begingroup
\parskip=0pt
\setcounter{tocdepth}{2}
\tableofcontents
\endgroup

\section{Hamming Weight Conjugate Energy Exchanges}
\label{app:Hamming_weight_conj}

The central mathematical object of this work is the inequality 
\begin{gather}
    \frac{T_M}{T_S}\omega < (i_M - j_M)\cdot\Gamma, \label{eq:gen_ineq_app}
\end{gather}
which when valid indicates that exchanging the energy levels $\ket{0_S i_M}$ and $\ket{1_S j_M}$ in the joint state of system at temperature $T_M$ with gap $\omega$ and $n$ qubit machine at temperature $T_M$ with gaps $\Gamma = (\gamma_1, \gamma_2, \dots, \gamma_n)$, cools $S$. In much of the manuscript we set $j_M = i_M \oplus 1$, firstly because it allows us to index the ground and excited state subspaces using a single index and secondly because it is the exchange of $\ket{0_S i_M}$ and $\ket{1_S i_M \oplus1}$ which determines whether there is any other $\ket{1_S j_M}, \, j_M \in \{0,1\}^{\times n}$ which when exchanged with $\ket{0_S i_M}$ cools $S$ also. In this Appendix we will explain why this is the case, both from the perspective of the energetic structure of the problem and from the mathematical perspective of swapping the largest entry with the smallest entry across the median of a vector. 

\subsection{A Strategy for detecting Disordered Populations}
\tikzset{
	ncbar angle/.initial=90,
	ncbar/.style={
		to path=(\tikztostart)
		-- ($(\tikztostart)!#1!\pgfkeysvalueof{/tikz/ncbar angle}:(\tikztotarget)$)
		-- ($(\tikztotarget)!($(\tikztostart)!#1!\pgfkeysvalueof{/tikz/ncbar angle}:(\tikztotarget)$)!\pgfkeysvalueof{/tikz/ncbar angle}:(\tikztostart)$)
		-- (\tikztotarget)
	},
	ncbar/.default=0.05cm,
}
\tikzset{square left brace/.style={ncbar=0.05cm}}
\tikzset{square right brace/.style={ncbar=-0.05cm}}

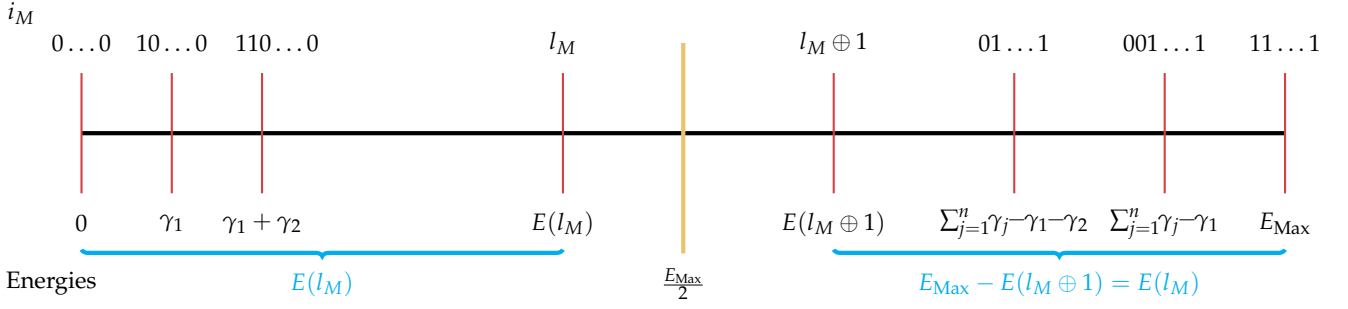
\begin{figure}[h]
	\begin{tikzpicture}[scale = 4]
	\centering
	\draw[ultra thick] (0,0) -- (4,0);
	\draw[red, thick] (4,-0.2) -- (4,0.2);
    \draw [red, thick] (0,-0.2) -- (0,0.2);
    \draw [red, thick] (0.3,-0.2) -- (0.3,0.2);
	\draw [red, thick] (3.6,-0.2) -- (3.6,0.2);
	\draw [red, thick] (0.6,-0.2) -- (0.6,0.2);
	\draw [red, thick] (3.1,-0.2) -- (3.1,0.2);
	\draw [red, thick] (2.5,-0.2) -- (2.5,0.2);
	\draw [red, thick] (1.6,-0.2) -- (1.6,0.2);
	\draw[ultra thick, yellow] (2,-0.4) -- (2,0.3);
    \node[] at (-0.1, -0.5) {Energies};
    \node[] at (-0.2, 0.4) {$i_M$};
    \node[] at (0, 0.3) {$0\dots0$};
    \node[] at (0.3, 0.3) {$10\dots0$};
    \node[] at (0,-0.3) {$0$};
    \node[] at (0.65,0.3){$110\dots0$};
	\node[] at (0.3,-0.3) {$\gamma_1$};
    \node[] at (0.6,-0.3) {$\gamma_1 + \gamma_2$};
    \node[] at (1.6,0.3) {$l_M$};
    \node[] at (1.6,-0.3) {$E(l_M)$};
	\node[] at (3.6,0.3) {$001\dots1$};
    \node[] at (3.6,-0.3) {$\sum^{n}_{j = 1} \hspace{-0.1cm} \gamma_j \hspace{-0.1cm}- \hspace{-0.1cm} \gamma_1$};
	\node[] at (3.1,0.3) {$01\dots1$};
    \node[] at (3.1,-0.3) {$\sum^{n}_{j = 1} \hspace{-0.1cm}\gamma_j \hspace{-0.1cm} - \hspace{-0.1cm}\gamma_1 \hspace{-0.1cm} - \hspace{-0.1cm} \gamma_2$};
	\node[] at (2.5,0.3) {$l_M \oplus 1$};
    \node[] at (2.5,-0.3) {$E(l_M \oplus 1)$};
    \node[] at (4,-0.3) {$E_\text{Max}$};
    \node[] at (4,0.3) {$11\dots1$};
	\draw[decoration={brace,raise=1ex,}, decorate,ultra thick, cyan] (1.6,-0.35) -- (0,-0.35);
    \draw[decoration={brace,raise=1ex,mirror}, decorate,ultra thick, cyan] (2.5,-0.35) -- (4,-0.35);
    \node[] at (3.25,-0.5) {\textcolor{cyan}{$E_\text{Max} - E(l_M\oplus1) =E(l_M)$}};
    \node[] at (0.8,-0.5) {\textcolor{cyan}{$E(l_M)$}};
    \node[] at (2,-0.5) {$\frac{E_\text{Max}}{2}$};
	\end{tikzpicture}
\caption{In this figure, we visualise the energy levels of the machine seeing that i) Hamming weight conjugate pairs are equally distant from the extremal energies i.e. $E(l_M)$ is equally away from $0$ as $E(l_M \oplus 1)$ is from $E_\text{Max}$  and ii) Hamming weight conjugate pairs fall on one side or the other of the median energy value $E_\text{Max}/2$.}
\label{fig:energetic_line}
\end{figure}
\paragraph*{An Informal Argument.} Considering Hamming weight conjugate pairs by setting $\ket{1_S j_M} = \ket{0_S i_M \oplus 1}$ gives us a systematic way to go through all energy levels of the system and machine to find the disordered populations which must be swapped. In the inequality Eq.~\eqref{eq:gen_ineq_app} which detects which energy levels should be exchanged, the right-hand side only features the machine energy levels $i_M$ and $j_M$. By making this choice in indexing and recalling that the entries of $\Gamma$, $\gamma_j \leq \gamma_{j+1}$ are in non-decreasing order, we are exposing ourselves to a choice in the energetic ordering of these machine energy levels. As visualised in Fig~\ref{fig:energetic_line} the bitstring $\ket{l_M}$ with energy $E(l_M) = l_M\cdot\Gamma = \sum^{n}_{j = 1} x_j \gamma_j$ where $x_j$ is the $j$th bit in $l_M$, is energetically $E(l_M)$ away from 0 while the energy of its Hamming weight conjugate $E(l_M \oplus 1)$ is also $E(l_M)$ from the maximum energy $E_\text{Max} = \sum^{n}_{j = 1} \gamma_j$ as 
\begin{gather}
    E(l_M \oplus 1) = E_\text{Max} - E(l_M).
\end{gather}
With this observation we note that if $E(l_M)$ is lower than the median energy value $E_\text{Max}/2$ then $E(l_M \oplus 1)$ must be above it that is $E(l_M) <  E_\text{Max}/2$ implies $E(l_M \oplus 1)  > E_\text{Max}/2$ since $E(l_M \oplus 1) = E_\text{Max} - E(l_M)$.

Perhaps most importantly, if one chooses to exchange $\ket{0_S i_M}$ with $\ket{1_S i_M \oplus 1}$ proceeding lexicographically through $i_M \in \{0,1\}^{\times n}$ pairing Hamming weight conjugate kets in the inequality \textit{one always pairs the energetically furthest eigenkets available}. We begin with $|1^n - 0^n|\cdot \Gamma = E_\text{Max}$ then $|1^{n-1}0 - 0^{n-1}1|\cdot \Gamma = E_\text{Max} - \gamma_n$ and so on moving in decreasing size of energy difference until $|10^{n-1} - 01^{n-1}|\cdot \Gamma = -E_\text{Max} + \gamma_1$ and $|0^{n} - 1^{n}|\cdot \Gamma = -E_\text{Max}.$\\
Thus if $\ket{0_S i_M}$ does not lead to a valid inequality when paired with $\ket{1_S i_M \oplus 1}$ there can be no other energy level $\ket{1_S j_M}$ of a lower energy in the excited state subspace worth swapping with as $|i_M - j_M|\cdot\Gamma >|i_M - i_M\oplus1|\cdot \Gamma$ since $j_M$ corresponds to a lower energy than $i_M\oplus1$ but $\frac{T_M}{T_S}\omega \geq |i_M - i_M\oplus1|\cdot \Gamma$ by assumption. Whilst machine eigenkets of higher energy would have already been paired up. In this way, pairing Hamming weight conjugate energy levels exhausts the set of machine levels and detects all disordered populations required to be exchanged for cooling. For a mathematically clearer and complete exposition to this strategy see below. 

Lastly, as we mention in Sec~\ref{sec:unitary} while this pairing is natural and useful to consider for detecting disordered population it might not lead to the pairing for obtaining a gate complexity optimal cooling unitary.

\subsubsection*{Heuristic Cooling Algorithm}

Given that the initial state of the system and machine is diagonal in the energy eigenbasis, all the populations can be collected into two vectors $\bm{q}_0$ and $\bm{q}_1$ such that 
\begin{equation}
    p_0 = \sum_{i=1}^{2^{n}} q_i, ~ ~  q_i \in \bm{q}_0, \hspace{1cm} p_1 = 1 - p_0 = \sum_{i=1}^{2^{n}} l_i, ~ ~  l_i \in \bm{q}_1,
\end{equation}
where $p_0$ and $p_1$ are the initial ground and excited state populations of the system qubit. A heuristic way to perform optimal cooling would be to first order $\bm{q}_0$ and $\bm{q}_1$ in non-decreasing order, giving $\bm{q}^{\downarrow}_0$ and $\bm{q}^{\downarrow}_1$ respectively. One can then compare the first component of $\bm{q}^{\downarrow}_1$ with the final component of $\bm{q}^{\downarrow}_0$ and perform a permutation between the two levels if the first component of $\bm{q}^{\downarrow}_1$ is greater than the final component of $\bm{q}^{\downarrow}_0$. If valid, this will perform the optimal exchange for increasing the ground state population of the system qubit in a single two level permutation. Next, one compares the second component of $\bm{q}^{\downarrow}_1$ with the second to last component of $\bm{q}^{\downarrow}_0$ and performs a permutation between the two levels if the second component of $\bm{q}^{\downarrow}_1$ is greater than the second to last component of $\bm{q}^{\downarrow}_0$, and so on. Each comparison of components will lead to a two-level permutation that represents the new optimal energy level exchange for increasing the ground state population (given that the previous permutations have been performed). If one reaches the $k$th component of $\bm{q}^{\downarrow}_1$ and it is not greater than the $(2^{n-1} - k)$th component of $\bm{q}^{\downarrow}_0$, one has successfully moved the populations such that the ground state population of the system qubit is maximised and, hence, they have performed optimal cooling via unitary dynamics. See Fig.~\ref{fig:heuristic_cooling_alg} for a pictorial depiction of the heuristic cooling algorithm

In total, one swaps the $k$ largest elements of $\bm{p}_1$ with the $k$ smallest elements of $\bm{p}_0$, where the size of $k$ is dependent on the system and machine being used. Note, once the $k$ elements that must be swapped are known, it does not matter for obtaining the largest attainable ground state population which of the $k$ elements of $\bm{q}_0$ are swapped for the $k$ elements of $\bm{q}_1$. All that is important is that all $k$ elements are swapped.  

\begin{figure}
    \centering
    \includegraphics[scale=0.8]{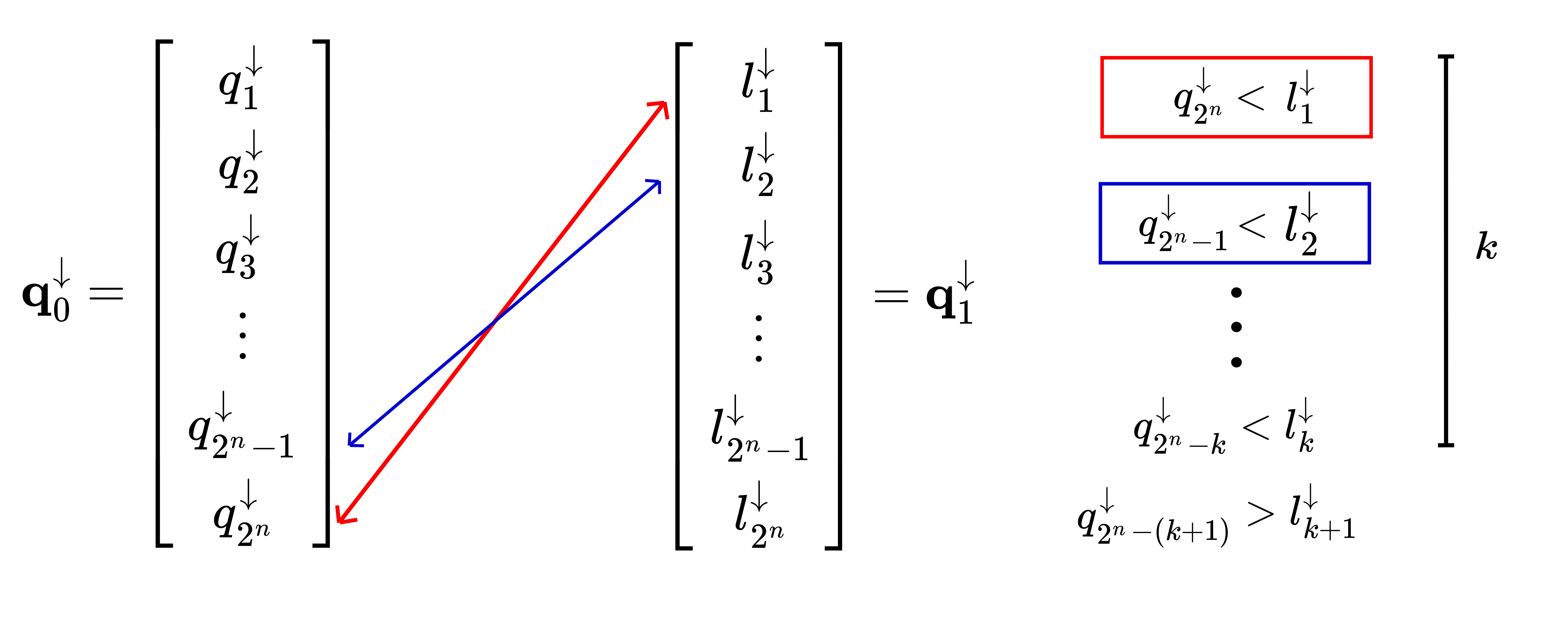}
    \caption{A diagram of the heuristic cooling algorithm. The vectors of populations are first placed in non-decreasing order. The final component of the ground state population vector is compared to the first component of the excited state population vector and a swap between the levels is perform if the former is larger. This process is repeated until a point is reached where the component from the ground state population vector is larger than the component from the excited state population vector. Given at each stage the largest remaining component from the ground state population vector is swapped with the smallest remaining component of the excited state population vector, it represents the optimal (in terms of increasing the system ground state) two-level permutation that can be performed.}
    \label{fig:heuristic_cooling_alg}
\end{figure}

\subsubsection*{Cooling Without Populations}  

\begin{figure}
    \centering
    \includegraphics[scale=0.6]{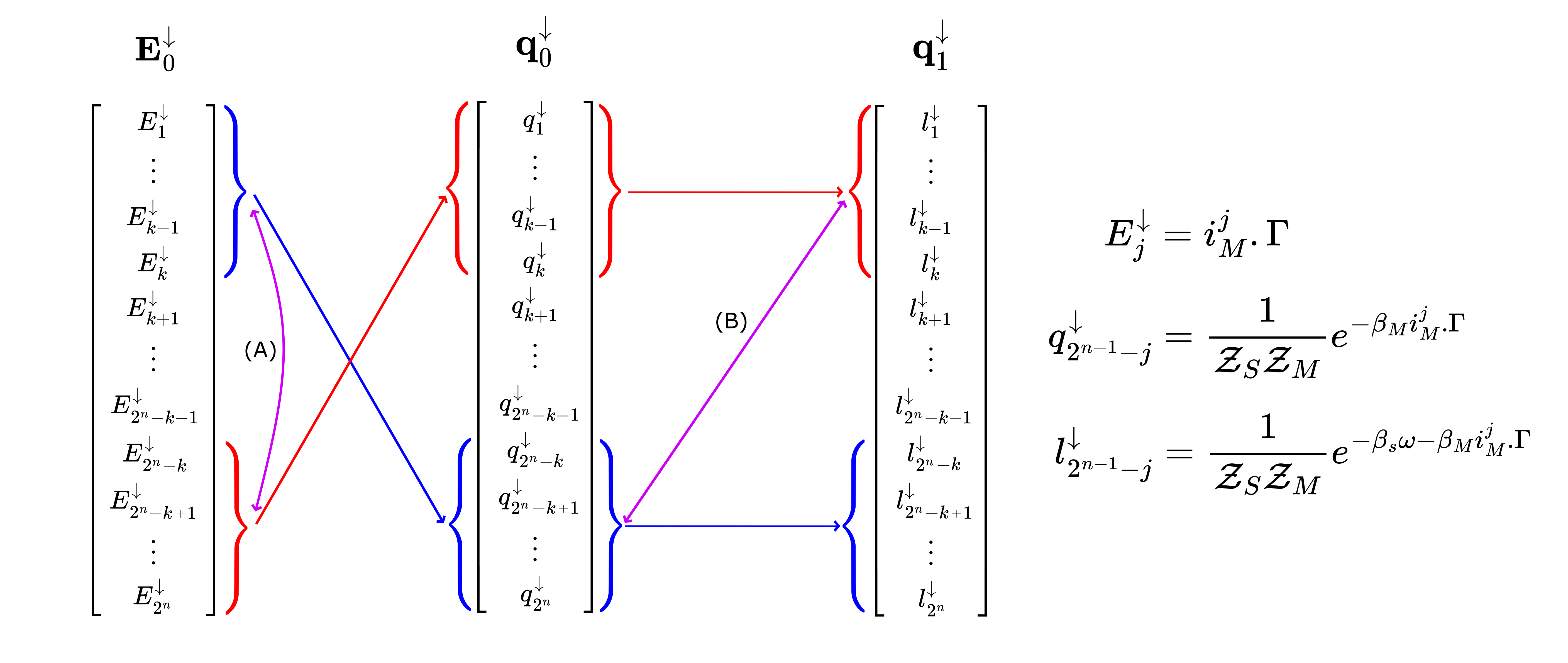}
    \caption{A figure detailing the proof strategy used in Lemma.~\ref{Lemma:coolingSets}. One first considers a vector of energies, $\bm{E}_0$, ordered in non-decreasing order and generated from the ordered population vectors. Firstly, it can be shown that if $i_M$ identifies the $k$th largest energy in $\bm{E}_0$, then $(i_M \oplus 1)$ identifies the $(2^{n}-k)$ largest energy in $\bm{E}_0$. Then, given that $\bm{E}_0$ and $\bm{p}_0$ have reverse orderings, if $i_M$ identifies the $k$th largest energy in $\bm{E}_0$, it identifies the $(2^{n}-k)$th largest population in $\bm{p}_0$. Next, as $\bm{p}_0$ and $\bm{p}_1$ have the same ordering, it can be concluded that if $i_M$ identifies the $(2^{n}-k)$th largest population in $\bm{p}_0$, it identifies the $(2^{n}-k)$th largest population in $\bm{p}_1$ as well. Hence, the relationships found between the $n$-bit-strings identifying the $k$ largest and smallest energies of $\bm{E}_0$ also apply to the $n$-bit-strings identifying the $k$ smallest populations of $\bm{p}_0$ and the $k$ largest populations of $\bm{p}_1$. Red and blue lines represent components of the respective vectors which are identified by the same $n$-bit-strings. The (A) purple line represents the relationship proved between the sets of $n$-bit strings identifying the $k$ largest and smallest energies of $\bm{E}_0$ that can then be used to infer the relationship between the sets of $n$-bit strings identifying the $k$ smallest populations of $\bm{p}_0$ and $k$ largest populations of $\bm{p}_1$, represented by the (B) purple line.}
    \label{fig:proof_stratergy_figure}
\end{figure}

Now, given some qubit at temperature $\beta_S$ with energy gap $\omega$, and access to some machine at temperature $\beta_M$ with energetic gap structure $\Gamma = (\gamma_1, \gamma_2, \ldots, \gamma_n)$, \textit{is it possible to find which two level swaps must be performed without first calculating all the populations and then ordering them}? 

To do this, we first note that each element of $\bm{q}_0$ and $\bm{q}_1$ can be identified with an $n$-bit-string, $i_M$, as
\begin{equation}
        \bm{q}_0 = \bigg\{ \frac{e^{-\beta_M i_{M} \cdot \Gamma}}{\mathcal{Z}_{S}\mathcal{Z}_M} : i_M \in \{0,1\}^{\times n} \bigg\}, \hspace{0.5cm} \bm{q}_1 = \bigg\{ \frac{e^{-\beta_S \omega-\beta_M i_{M} \cdot \Gamma}}{\mathcal{Z}_{S}\mathcal{Z}_M} : i_M \in \{0,1\}^{\times n} \bigg\}.
\end{equation}
Therefore, the $k$ elements of $\bm{q}_0$ that must be swapped can be identified with a set of $n$-bit-strings, as can the $k$ elements of $\bm{q}_1$ that they must be swapped with. Moreover, from the above definitions it can easily be seen that $\bm{q}_1 = e^{-\beta_S \omega} \bm{p}_0$, and hence that $\bm{q}^{\downarrow}_1 = e^{-\beta_S \omega} \bm{p}^{\downarrow}_0$ --- multiplying each element of an ordered list by a positive constant does not change the ordering. Therefore, finding the set of $n$-bit-strings that identify the $k$ largest elements of $\bm{q}_1$ is equivalent to finding the set of $n$-bit-strings that identify the $k$ largest elements of $\bm{q}_0$. One can therefore consider only $\bm{q}_0$ and aim to find the set of $n$-bit-strings that identify the smallest $k$ elements of $\bm{q}_0$, denoted by $\mathbb{S}$, and the set of $n$-bit-strings that identify the largest $k$ elements $\bm{q}_0$, denoted by $\overline{\mathbb{S}}$. From the above reasoning, the bit-strings that identify the $k$ largest elements of $\bm{q}_0$, $\overline{\mathbb{S}}$, will also then identify the $k$ largest elements of $\bm{q}_1$.  We can then consider the following Lemma, which states that if the bit-string $i_M$ is in the set of $n$-bit-strings that identify one of the $k$ smallest populations of $\bm{q}_0$, then its conjugate, $i_M \oplus 1$, is in the set of $n$-bit-strings that identify one of the $k$ largest populations of $\bm{q}_1$.

\begin{lemma} \label{Lemma:coolingSets}
    If $i_M \in \mathbb{S}$, then $i_M \oplus 1 \in \overline{\mathbb{S}}$ where $i_M \in \{0,1\}^n$.
\end{lemma}
\begin{proof}
    To begin, we first convert the vector of populations to a vector of energies 
    \begin{equation}
        \bm{E}_0 = - \frac{1}{\beta_S} \ln{\big(\mathcal{Z}_S \mathcal{Z_M} \bm{q}_0 \big)}  = \{ i_M \cdot \Gamma: i_M \in \{0,1\}^n \}. \label{equ:populationToEnergy}
    \end{equation}
    The constant $-1/ \beta_M \ln{(\mathcal{Z}_S \mathcal{Z_M})}$ does not effect the ordering. Then, given that $1 / \beta_M$ is a positive constant and $\ln({a}) \leq \ln{(b)}$ if $a \leq b$, $-\bm{E}_0$ has the same ordering as $\bm{q}_0^\downarrow$; the negative then flips the ordering. Therefore, the set of $n$-bit-strings that identify the $k$ smallest values of $\bm{q}_0$ identify the $k$-largest values of $\bm{E}_0$, and vice versa for the $k$ largest values of $\bm{q}_0$. Hence, $\mathbb{S}$ now identifies the largest $k$ values of $\bm{E}_0$, and $\overline{\mathbb{S}}$ the $k$ largest of $\bm{E}_1$ (defined as in (Eq.~\ref{equ:populationToEnergy}) but using $\bm{q}_1$). However, such labelling is irrelevant given one only cares about the relationship between the $n$-bit-strings in the sets. 

    Next, we note that 
    \begin{equation}
        i_M \cdot \Gamma + (i_M \oplus 1) \cdot \Gamma = E_{\rm Max}, ~ ~ ~ {\rm where} ~ ~ ~ E_{\rm Max} = \sum_{i=1}^{n} \gamma_i,
    \end{equation}
    meaning that the component of $\bm{E}_0$ identified by the bit-string $i_M$ is paired to the component of $\bm{E}_0$ identified by the bit-string $(i_M \oplus 1)$. Evidently, the larger $i_M \cdot \Gamma$, the smaller $(i_M \oplus 1) \cdot \Gamma$ and vice versa. From here, it can already be seen that if a $n$-bit-string leads to one of the $k$ largest values of $\bm{E}_0$, its conjugate will lead to one of the $k$ smallest. More formally, one can define 
    \begin{equation}
        \begin{split}
            E_{\sup} &= \min_{i_M \in \mathbb{S}} i_M \cdot \Gamma = i_{\sup} . \Gamma, \\
            E_{\inf} &= \max_{i_M \in \{0,1\}^n \setminus \mathbb{S}} i_M \cdot \Gamma = i_{\inf} \cdot \Gamma, 
        \end{split}
    \end{equation}
    such that $E_{\sup}$ is the smallest component of $\bm{E}_0$ of the $k$ largest components (acting as a kind of supremum), and $E_{\inf}$ is the largest component of $\bm{E}_0$ not in the $k$ largest (acting as a kind of infimum), giving $E_{\sup} \geq E_{\inf}$. Therefore, 
    \begin{equation}
        E_{\rm Max} - E_{\inf} \geq E_{\rm Max} - E_{\sup} ~ ~  \implies ~ ~ (i_{\inf} \oplus 1) \cdot \Gamma \geq (i_{\sup} \oplus 1) \cdot \Gamma. 
    \end{equation}
    All other $n$-bit-strings in $\mathbb{S}$ give components of $\bm{E}_0$ larger than $E_{\sup}$ and hence their conjugates lead to components smaller than $E_{\rm Max} - E_{\sup}$. Thus far, we have shown that if $i_M$ is an $n$-bit-string that represents one of the $k$ largest energies of $\bm{E}_0$, $(i_M \oplus 1)$ is an $n$-bit-string that represents one of the $k$ smallest energies of $\bm{E}_0$. 

    Now, given that the ordering of $\bm{E}_0$ is that of $\bm{p}_0$ reversed, if $i_M$ is an $n$-bit-string that represents one of the $k$ largest energies of $\bm{E}_0$, it is an $n$-bit-string that represents one of the $k$ smallest populations of $\bm{p}_0$. Hence, $(i_M \oplus 1)$ is an $n$-bit-string that represents one of the $k$ largest populations of $\bm{p}_0$. To complete the proof, one utilises the fact that the ordering of $\bm{p}_0$ in non-decreasing order is the same as that of $\bm{p}_1$. Hence, the $n$-bit-strings that identify the $k$ largest populations of $\bm{p}_0$ also identify the $k$ largest populations of $\bm{p}_1$. Finally, given they are the same set of bit-strings, any relationship between the $n$-bit strings identifying the $k$ largest and smallest energies of $\bm{E}_0$ must also hold for those $n$-bit-strings representing the $k$ largest and smallest populations of $\bm{p}_0$ and $\bm{p}_1$. They then must hold between the $n$-bit strings identifying the $k$ smallest populations of $\bm{p}_0$ and $k$ largest populations of $\bm{p}_1$ for the same reason. Using the relationship proved between the $n$-bit strings representing the $k$ largest and smallest energies of $\bm{E}_0$, it can therefore be concluded that if $i_M \in \mathbb{S}$ then $(i_M \oplus 1) \in \overline{\mathbb{S}}$. In other words, if $E(i_M)$ is larger than the energetic median, $E(i_M \oplus 1)$ is smaller than it as visualised in Fig.~\ref{fig:energetic_line}. Interestingly, one can go further and note that if the $n$-bit-string $i_M'$ identifies the $k'$ largest component of $\bm{p}_0$, then $(i_M' \oplus 1)$ identifies the $k'$ smallest component (the $2^{n-1}-k'$ largest) of $\bm{p}_1$. 
\end{proof}

Returning now to the heuristic cooling algorithm.
Lemma.~\ref{Lemma:coolingSets} states that if $\mathbb{S}$ represents the set of $n$-bit-strings that identify the smallest components of $\bm{q}_0$, then the set of $n$-bit-strings that identity the largest components of $\bm{q}_1$, $\overline{\mathbb{S}}$, are given by $\overline{\mathbb{S}} = \{ i_M \oplus 1 : i_M \in \mathbb{S}\}$. Given that 
\begin{equation}
    \min_{j_M \in \overline{\mathbb{S}}} j_M\cdot \Gamma > \max_{i_M \in \mathbb{S}} i_M\cdot \Gamma, 
\end{equation}
the smallest component of $\bm{q}_1$ to be swapped is still larger than the largest component of $\bm{q}_0$ to be swapped (where this ceases to be true defines the end of the heuristic cooling algorithm i.e. what is the value of $k$). Which of these $k$ components are paired with each other to be exchanged does not matter --- they all lead to positive changes in the ground state population of the system. Therefore, without loss of generality, one can compare the population of $\bm{p}_0$ identified by the $n$-bit-string $i_M$ with the population of $\bm{p}_1$ identified by the $n$-bit-string $(i_M \oplus 1)$ to see if the population identified by $i_M$ should be swapped. In fact, given that if the $n$-bit-string $i_M'\oplus1$ identifies the $k'$ largest component of $\bm{p}_1$, then $i_M'$ identifies the $k'$ smallest component of $\bm{p}_0$ (the $2^{n}-k'$ largest component), the strategy of swapping opposite Hamming weight pairs exactly outputs the same set of two level swaps as performing the heuristic cooling algorithm, in which at each stage the optimal two level swap to increase the ground state population of the system was performed. 

\begin{figure}[h]
\centering
\begin{tikzpicture}
    \draw[gray!20, step=0.5] (-0.9,-0.9) grid (10.9,6.9);
    \draw[-stealth, thick] (0,0) -- (10,0) node[right] {$i_M$};
    \node[] at (10,-0.75) {Machine Energy Levels};
    \draw[-stealth, thick] (0,0) -- (0,6) node[above] {Energies};
    \node[below] at (1,0) {$10\cdots0$};
    \node[below] at (3,0) {$110\cdots0$};
    \node[below] at (6,0) {$i_l$};
    \node[below] at (9,0) {$11\cdots1$};

    \draw[ thick] (1,1) node[left] {$\gamma_1$}  +(0,0.1) -- +(0,-0.1) +(-0.1,0) -- +(0.1,0);
    \draw[thick] (3,2) node[left] {} +(0,0.1) -- +(0,-0.1) +(-0.1,0) -- +(0.1,0);
    \node[left] at (2.8,2.35) {$\gamma_1 + \gamma_2$};
    \draw[ thick] (6,4) node[left] {$\sum_{\gamma_j \in E(i_l)} \gamma_j$}  +(0,0.1) -- +(0,-0.1) +(-0.1,0) -- +(0.1,0);
    \draw[thick] (9,5) node[left] {$\sum^{n}_{j = 1} \gamma_j$} +(0,0.1) -- +(0,-0.1) +(-0.1,0) -- +(0.1,0);
    
    \draw[red, thick, dotted] (0,1.5) -- (10,1.5);
    \node[red, right] at (10,1.5) {$\frac{T_{M}}{T_S} \omega$};
    \draw[blue, dashed] (1,1) -- (6,1);
    \node[blue, below] at (3,1) {$\tiny{\underbrace{|i_l - 10\dots0|}_{\Delta x}}$};
    \draw[blue, dashed] (1,1) -- (6,4);
    \draw[blue, dashed] (6,1) -- (6,4);
    \node[blue, left] at (3.5,2.75) {$\frac{\Delta E}{\Delta x}$};
    \draw[red, dashed] (1,1) -- (6,2.5);
    \node[red, above] at (3.75,1.75) {$\frac{T_M \omega}{T_S\Delta x}$};
    \draw[red, thick, dashed] (6,1) -- (6,2.5);
    \node[right, red] at (6,1.75) {$\frac{T_{M}}{T_S}\omega$};
    \node[blue,right] at (6,3) {$\tiny{ \underbrace{\sum_{\gamma_j \in E(i_M)} \gamma_j - \gamma_1}_{\Delta E}}$};
\end{tikzpicture}
\caption{The inequality Eq.~\eqref{eq:gen_ineq} can be visualised graphically on a graph whose $x$-axis represents the bitstrings labelling machine energy levels and the Hamming distance between these strings. Whereas the $y$-axis represents the energies of these levels and so their energy differences.}
\end{figure}
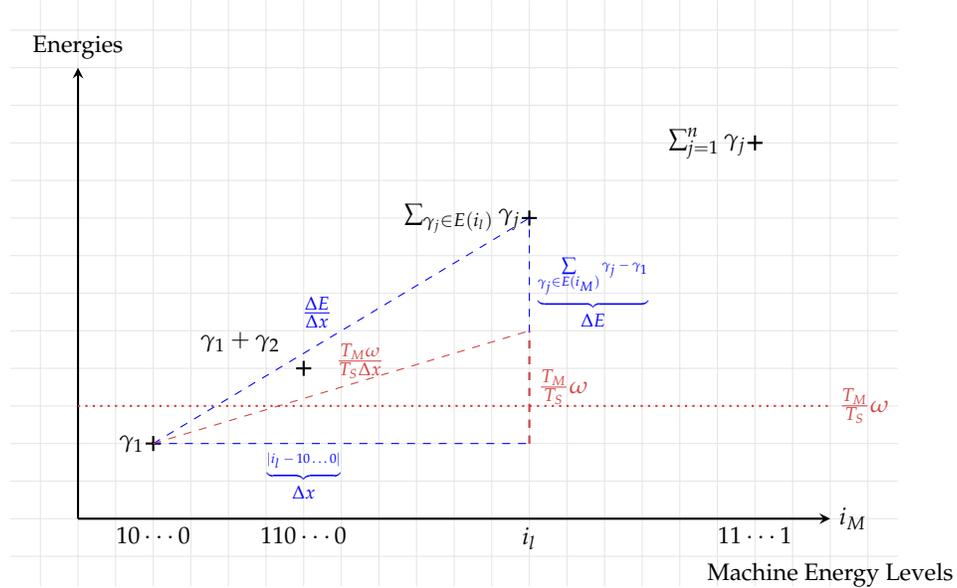

\subsection{A Graphical Interpretation}

The inequality Eq.~\eqref{eq:gen_ineq} also has a geometric interpretation when diving both sides by the absolute difference of Hamming distance between the bitstrings $i_M$ and $j_M$, $d_H(i_M, j_M) = |i_M - j_M| = \sum^{n}_{l =1} |i_l - j_l|$ where $i_l$ and $j_l$ are the $l$th bits in each string. This gives
\begin{gather}
    \frac{T_M \omega}{T_S |i_M - j_M|} < \frac{(i_{M} - j_{M})\cdot \Gamma}{|i_M - j_M|},
\end{gather} 
Here we see that the right hand side takes on the form of a slope $\Delta y / \Delta x$ where $\Delta x = |i_M - j_M|$ is the Hamming distance between machine energy levels $i_M$ and $j_M$ and $\Delta y = \Delta E = (i_{M} - j_{M})\cdot \Gamma$ is the energy difference between these machine energy levels. The slope the captures a sense of how far away the kets one exchanges have to be to obtain a given energy change.

If this slope is larger than the slope $T_M \omega /T_S |i_M - j_M|$ given by the cooling scenario constant and the same $\Delta x$ then the exchange of the levels $\ket{0_S i_M}$ and $\ket{1_S j_M}$ will cool $S$. Note that pairing inverse Hamming weight kets and setting $\ket{1_M j_M} = \ket{0_S i_M \oplus 1}$ set $\Delta x$ to a fixed maximal value of $n$ for any pairing. In this sense we again see the logic of this strategy that if the slope $\frac{(i_{M} - i_M \oplus 1)\cdot \Gamma}{n}$ is not greater than $ \frac{T_M \omega}{T_S n}$ there can be no other slope formed by pairing with a $j_M$ of energy lower than $i_M \oplus 1$ which results in a slope that exceeds $\frac{T_M \omega}{T_S n}$. 

\section{Example 2 -- Cooling using a 3-Qubit Machine}
\label{app:example_2}

To continue to understand the tools we have developed in this manuscript and the setting at hand let us consider an exemplary task of cooling a thermal qubit $S$ with gap $\omega$ with $\beta_S =1$ with access to three thermal qubits also with $\beta_M =1$ and gaps $\Gamma = (\gamma_1, \gamma_2, \gamma_3)$ such that  $\omega \leq \gamma_1 \leq \gamma_2 \leq \gamma_3$. We will use our tools to access which swaps should be performed for optimal cooling, which machines are irreducible or reducible, and evaluate a lower bound on optimal achievable ground state temperature using a $n=3$ qubit machine. 

\begin{figure}[h]
    \centering    \includegraphics[width=\linewidth]{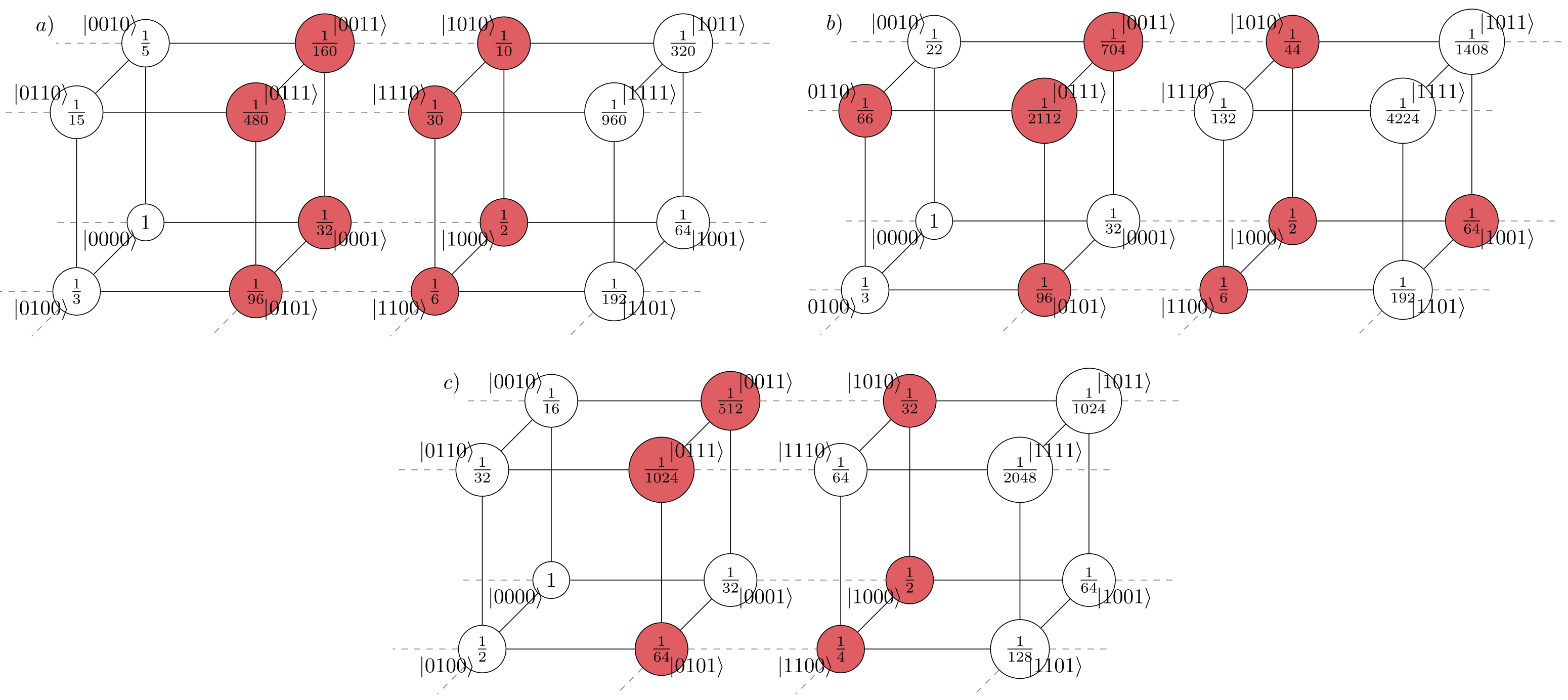}
    \caption{Here we illustrate the various cooling scenarios given a 3-qubit machine by giving a 3D projection of a hypercube. In a) we see the situation given by case 1 $\gamma_1 + \gamma_2 < \gamma_3$ and $\omega < \gamma_3 - \gamma_1 - \gamma_2$ for an example with $\omega = \log 2, \gamma_1 = \log 3, \gamma_2 = \log 5$ and $\gamma_3 = \log 32$, where the population vector has been factored by $\frac{40}{90}$. In b) we see case 2 for $\gamma_1 + \gamma_2 > \gamma_3$ and $\omega < \gamma_1 + \gamma_2 - \gamma_3$ for $\omega = \log 2, \gamma_1 = \log 3, \gamma_2 = \log 22 $ and $\gamma_3 = \log 32$ factored by $\frac{32}{69}$. Finally in c) we see the situation for case 3,4 and 5 with values satisfying case 5 explicitly $\omega = \log 2, \gamma_1 = \log 2, \gamma_2 = \log 16$ and $\gamma_3 = \log 32$ with population vector factored by $\frac{2048}{5049}$.}
    \label{fig:3_qubit_machine}
\end{figure}

The population vector for the ground and excited state of the system, $\bm{p}_0$ and $\bm{p}_1$ respectively, are 
\begin{align}
        \bm{p}_0 &= \frac{1}{\mathcal{Z}_S \mathcal{Z}_M}\begin{bmatrix}
        1, & e^{-\gamma_3}, & e^{-\gamma_2}, & e^{-(\gamma_2+\gamma_3)}, & e^{-\gamma_1}, & e^{-(\gamma_1+\gamma_3)}, & e^{-(\gamma_1+\gamma_2)}, & e^{-(\gamma_1+\gamma_2+\gamma_3)} \end{bmatrix}, \\
        \bm{p}_1 &= \frac{1}{\mathcal{Z}_S \mathcal{Z}_M}\begin{bmatrix}
        e^{-\omega}, & e^{-(\omega+\gamma_3)}, & e^{-(\omega+\gamma_2)}, & e^{-(\omega+\gamma_2+\gamma_3)}, & e^{-(\omega+\gamma_1)}, & e^{-(\omega+\gamma_1+\gamma_3)}, & e^{-(\omega+\gamma_1+\gamma_2)}, & e^{-(\omega+\gamma_1+\gamma_2+\gamma_3)} 
        \end{bmatrix},
\end{align}
in this scenario where $\mathcal{Z}_S = (1+e^{-\omega})$ and $\mathcal{Z}_M = (1+e^{-\gamma_1})(1+e^{-\gamma_2})(1+e^{-\gamma_3}).$ The heuristic algorithm given in the prior section can be used to determine a set of two level swaps which will lead to optimal cooling of the system by comparing elements of $\bm{p}_0$ and $\bm{p}_1$. Alternatively, it was also shown in the main text and Appendix~\ref{app:Hamming_weight_conj} that this set of two level swaps can be found by considering eq~\eqref{eq:ineq_simp} for all $i_M \in \{0,1\}^{\times 3}$. In this case, eq~\eqref{eq:ineq_simp} takes the form 
\begin{gather}
    \frac{1}{2}\left(\omega + \gamma_1 + \gamma_2 + \gamma_3\right ) < i_1\gamma_1 + i_2\gamma_2 + i_3\gamma_3,\label{eq:3_qubit_swap_condition}
\end{gather}
with the two level swap $\ket{0_s i_1 i_2 i_3} \leftrightarrow \ket{1_s(i_1 \oplus 1)(i_2 \oplus 1)(i_3 \oplus 1)}$ cooling the system if Eq.~\eqref{eq:3_qubit_swap_condition} is statised for a given $i_M = i_1i_2i_3$. Specifically, the following two level swaps can be performed to achieve optimal cooling if the related conditions are satisfied 
\begin{align}
        \ket{0_S000} \longleftrightarrow \ket{1_S111} ~ ~ &{\rm if} ~ ~ \omega < - \gamma_1 - \gamma_2 - \gamma_3, \label{cond_1}\\
        \ket{0_S100} \longleftrightarrow \ket{1_S011} ~ ~ &{\rm if} ~ ~ \omega <  \gamma_1 - \gamma_2 - \gamma_3, \label{cond_2} \\
        \ket{0_S010} \longleftrightarrow \ket{1_S010} ~ ~ &{\rm if} ~ ~ \omega <  - \gamma_1 + \gamma_2 - \gamma_3, \label{cond_3}\\
        \ket{0_S110} \longleftrightarrow \ket{1_S001} ~ ~ &{\rm if} ~ ~ \omega <  \gamma_1 + \gamma_2 - \gamma_3, \label{cond_4} \\
        \ket{0_S001} \longleftrightarrow \ket{1_S110} ~ ~ &{\rm if} ~ ~ \omega <  - \gamma_1 - \gamma_2 + \gamma_3, \label{cond_5} \\
        \ket{0_S101} \longleftrightarrow \ket{1_S010} ~ ~ &{\rm if} ~ ~ \omega <  \gamma_1 - \gamma_2 + \gamma_3, \label{cond_6} \\
        \ket{0_S011} \longleftrightarrow \ket{1_S100} ~ ~ &{\rm if} ~ ~ \omega <  - \gamma_1 + \gamma_2 + \gamma_3, \label{cond_7} \\
        \ket{0_S111} \longleftrightarrow \ket{1_S000} ~ ~ &{\rm if} ~ ~ \omega <  \gamma_1 + \gamma_2 + \gamma_3.  \label{cond_8}
\end{align}
Due to the assumed gap structure of $\omega \leq \gamma_1 \leq \gamma_2 \leq \gamma_3$, it can be seen that $\ref{cond_8}, \ref{cond_7}$ and $\ref{cond_6}$ are always satisfied. This therefore means that their associated two level swaps always cool the system. Note, from the perspective of the heuristic cooling algorithm, these conditions always being satisfied are equivalent to the fact that the following conditions on the populations from $\bm{p}_0$ and $\bm{p}_1$ are always satisfied 
\begin{alignat}{3}
         \omega &<  \gamma_1 - \gamma_2 + \gamma_3 &&\Longleftrightarrow  e^{-(\omega+\gamma_2)} > e^{-(\gamma_1+\gamma_3)}, \\
         \omega &<  - \gamma_1 + \gamma_2 + \gamma_3 &&\Longleftrightarrow  e^{-(\omega+\gamma_1)} > e^{-(\gamma_2+\gamma_3)}, \\
          \omega &<  \gamma_1 + \gamma_2 + \gamma_3 &&\Longleftrightarrow  e^{-\omega} > e^{-(\gamma_1+\gamma_2+\gamma_3)}.
\end{alignat}
Meanwhile, $\ref{cond_3}, \ref{cond_2}$ and $\ref{cond_1}$ are never satisfied, meaning their associated two level swaps never cool the system. As above, these conditions directly correspond to conditions on the populations from $\bm{p}_0$ and $\bm{p}_1$ that are never satisfied.  

This leaves two possible conditions, namely $\omega < \gamma_1 + \gamma_2 - \gamma_3$ (\ref{cond_4}) and $\omega < - \gamma_1 - \gamma_2 + \gamma_3$ (\ref{cond_5}) that one must consider. Given these, there are four possible cases to consider: 1) $\ref{cond_4}$ and $\ref{cond_5}$ are satisfied, 2) 
$\ref{cond_4}$ is satisfied and $\ref{cond_5}$ is not, 3) $\ref{cond_5}$ is satisfied and $\ref{cond_4}$ is not, and 4) $\ref{cond_4}$ and $\ref{cond_5}$ are not satisfied.
\\
\\
\textbf{\textit{Case 1:}}~$\ref{cond_4}$ and $\ref{cond_5}$ are satisfied.

It can easily be seen that this is never possible. This is due to $\omega \geq 0$ and the right hand sides of $\ref{cond_4}$ and $\ref{cond_5}$ being the negative of each other.
\\
\\
\textbf{\textit{Case 2:}}~$\ref{cond_4}$ is satisfied and $\ref{cond_5}$ is not.
\\
In addition to the three always performed swaps, the additional swap
\begin{equation}
    \ket{0_S110} \longleftrightarrow \ket{1_S001} ~ ~ {\rm as} ~ ~ \omega < \gamma_1 + \gamma_2 - \gamma_3 ~ ~ \Longleftrightarrow ~ ~  e^{-(\omega+\gamma_3)} > e^{-(\gamma_1+\gamma_2)} \\
\end{equation}
is needed for optimal cooling in this case. If all swaps are performed then the final ground state population in this case is
\begin{equation}
    p^{*}_{0}(2) = \frac{1 + e^{-\gamma_1} + e^{-\gamma_2} + e^{-\gamma_3}}{\mathcal{Z}_M},
\end{equation}
where 
\begin{equation}
     \omega < \gamma_1 + \gamma_2 - \gamma_3 ~ ~ \implies ~ ~ p_{0}^*(2) > \frac{1}{1+e^{-\gamma_3}}
\end{equation}
meaning one cools the system further than the coldest qubit.
\\
\\
\textbf{\textit{Case 3:}}~$\ref{cond_5}$ is satisfied and $\ref{cond_4}$ is not. 
\\
In addition to the three always performed swaps, the additional swap
\begin{equation}
     \ket{0_S001} \longleftrightarrow \ket{1_S110} ~ ~ {\rm as} ~ ~ \omega < - \gamma_1 - \gamma_2 + \gamma_3 ~ ~ \Longleftrightarrow ~ ~  e^{-(\omega+\gamma_1+\gamma_2)} > e^{-\gamma_2} \\
\end{equation}
is needed for optimal cooling in this case. If all swaps are performed then the final ground state population in this case is
\begin{equation}
    p^*_{0}(3) = \frac{1}{1+e^{-\gamma_3}},
\end{equation}
meaning the optimal one can achieve is the temperature of the coldest qubit. Hence, the optimal unitary for cooling in this case is simply the $\mathtt{SWAP}$ between the system and coldest qubit in the machine. 

In the main text and Appendix.~\ref{app:reduc} the notion of machine reducibility is discussed. Case 3 fits the definition of an $(n-1)$-reducible machine given that two of the three machine qubits are not involved in the cooling process --- the final ground state populations does not depend on $\gamma_1$ or $\gamma_2$.  In Eq.~\eqref{eq:n-1_red} the follow sufficient condition for $(n-1)$-reducibility is given 
\begin{equation}
    \frac{1}{2}\left(\frac{T_M}{T_S}\omega + E_\text{Max}\right) < \gamma_n.
\end{equation}
In this case, this takes the form 
\begin{equation}
    \frac{1}{2}(\omega + \gamma_1 + \gamma_2 + \gamma_3 ) < \gamma_3,
\end{equation}
which can equivalently be expressed as 
\begin{equation}
    \omega < - \gamma_1 - \gamma_2 + \gamma_3.
\end{equation}
Hence, $\ref{cond_5}$ is exactly the condition for $(n-1)$-reducibility when $n=3$. It being satisfied therefore immediately tells us that the optimal cooling that can be achieved is to the temperature of the coolest machine qubit, and that a full swap between the coldest machine qubit and system is an optimal cooling unitary.  
\\
\\
\textbf{\textit{Case 4:}}~$\ref{cond_4}$ and $\ref{cond_5}$ are not satisfied. 
\\
This leads to no additional swaps beyond the three always performed swaps. Performing these gives 
\begin{equation}
    p^{*}_{0}(4) = \frac{1+ e^{-\omega} + e^{-\gamma_1}  + e^{-\gamma_2}+e^{-\gamma_3} +e^{-(\omega+\gamma_2)} 
 +e^{-(\omega+\gamma_3)} + e^{-(\gamma_1+\gamma_2)}  }{\mathcal{Z}_S \mathcal{Z}_M}.
\end{equation}
From here, it can be seen that  
\begin{equation}
    \omega > - \gamma_1 - \gamma_2 + \gamma_3 ~ ~ \implies ~ ~ p_{0}^*(3) > \frac{1}{1+e^{-\gamma_3}},
\end{equation}
where the first condition comes from $\ref{cond_5}$ not being satisfied. Hence, one can cool past the temperature of the coldest qubit in this case. 
\\
\\
It was seen in case $3$ that the machine was $(n-1)$-reducible. We now consider whether the machines in case $2$ and $4$ are irreducible, meaning all qubits in the machine contribute to the optimal cooling. In the Appendix.~\ref{app:reduc}, it is shown that a necessary and sufficient condition for a machine to be irreducible is for there to exists an $l_{n-1} \in \{0,1\}^{ \times n-1}$ such that 
\begin{equation}
     0 \leq \frac{1}{2}\left(\frac{T_M}{T_S}\omega + E_\text{Max}\right) - E(l_{n-1}) < \gamma_1, \label{eq:irreducibility_appendix}
\end{equation}
where $E(l_{n-1})$ is an energy of the $(n-1)$ coldest qubits. In our case of $n=3$ machine qubits this becomes 
\begin{align}
0 &\leq \frac{1}{2}\left(\omega + \gamma_1 + \gamma_2 + \gamma_3 \right) - E(0_S0i_2i_3) < \gamma_1. \label{eq:3_qubit_irreducible}
\end{align}

We first consider case 2, where $\ref{cond_4}$ is satisfied and $\ref{cond_5}$ is not. For $i_2 = 1$ and $i_3 = 0$ Eq.~\eqref{eq:3_qubit_irreducible} becomes 
\begin{align}
    - \gamma_1 + \gamma_2 - \gamma_3 < \omega < \gamma_1 + \gamma_2 - \gamma_3. 
\end{align}
Given $\ref{cond_4}$ is satisfied the right hand side of this equation is satisfied. Given the non-decreasing energy structure of the machine, the left hand side is always statised. Hence, the machine in case 2 is irreducible. 

when then consider case 4, where $\ref{cond_4}$ is satisfied and $\ref{cond_5}$ is not. For $i_2 = 0$ and $i_3 = 1$ Eq.~\eqref{eq:3_qubit_irreducible} becomes 
\begin{align}
    - \gamma_1 - \gamma_2 + \gamma_3 < \omega < \gamma_1 - \gamma_2 + \gamma_3. 
\end{align}
Once again, given the non-decreasing energy structure of the machine the right hand side is always statised. Then, given $\ref{cond_5}$ is not satisfied, the left hand side is satisfied. Therefore, the machine in case 4 is also irreducible. Given that the optimal achievable ground state populations in case 2 and case 4, $p^*_0(2)$ and $p^*_0(4)$ respectively, depend on all elements in $\Gamma$, intuitively one would assume the machines are irreducible. However, the above calculations confirm this. 

Lastly, one can lower bound the optimal achievable change in population of the system, $\Delta p_0^*$, using Eq.~\eqref{eq:lower_bound_pop_change}, which, in this case gives
\begin{equation}
    \begin{split}
        \Delta p_0^* \geq \frac{1}{\mathcal{Z}_S (1+e^{-\gamma_3})} \big( e^{-\omega} - e^{-\gamma_3} \big) + \frac{1}{\mathcal{Z}_S \mathcal{Z}_M} \big( e^{-\gamma_3} - e^{-(\omega + \gamma_1 + \gamma_2)} ~ \big). 
    \end{split}
\end{equation}
where it is noted that the second term is negative if $\ref{cond_5}$ is satisfied. 

\section{How cold can we go?}
\label{app:how_cold}

Given a thermal qubit at temperature $T_S$ with gap $\omega$ and a machine $M$ of $n$-qubits with gaps $\Gamma = (\gamma_1, \gamma_2, \dots, \gamma_n)$ a natural question to ask is \textit{how far we can cool $S$ with a single unitary interaction across $S$ and $M$}? 

Addressing this question within the language of the hierarchy of inequalities we have developed, let's assume that an agent has permuted the energy levels within the joint system-machine Hilbert space such that system cannot be cooled further i.e. none of the inequalities Eq.~\eqref{eq:gen_ineq} remain valid.

In doing so, the agent obtains a ground state population for their system qubit after cooling that can be expressed as 

\begin{align}
    p'_0 = \frac{1}{\mathcal{Z}_S\mathcal{Z}_{M}}\underbrace{\left(\sum_{i_M \in \{0,1\}^{\times n} \, \setminus \, \mathbb{S}}  e^{-\beta_M i_M \cdot \Gamma}\right.}_{\text{unchanged}} + \underbrace{\left.\sum_{j_M \in \mathbb{S}} e^{-\beta_S \omega - \beta_M (j_M \oplus 1)\cdot \Gamma}\right)}_{\text{swapped}} , \label{eq:gs_pop}
\end{align}
where $\mathbb{S} = \left\{ j_M \in \{0,1\}^{\times n} | \frac{1}{2}(\frac{T_{M}}{T_S}\omega + E_{\text{Max}}) < E(j_M) \right\}$ is the set of machine energy levels corresponding to levels in the joint system-machine Hilbert space that will cool the system if swapped $\ket{0_S j_{M}} \longleftrightarrow \ket{1_S j_M \oplus 1}$ and $\mathcal{Z}_{S} = (1 + e^{-\beta_S \omega}), \, \mathcal{Z}_{M } =\prod^{n}_{j = 1} (1 + e^{-\beta \gamma_j})$ are the partition functions of the system and machine, respectively. In Eq.~\eqref{eq:gs_pop} we have represented the groundstate population after cooling as a sum of two contributions. The first term involving energy levels that did not satisfy Eq.~\eqref{eq:gen_ineq} and so would warm the system if swapped and the second term involving the populations obtained by swapping the levels that satisfied Eq.~\eqref{eq:gen_ineq}. Note that the choice of swapping Hamming weight conjugate levels is non-unique but is useful as it allows us to express changes in energy in terms of one parameter over the set of bit strings.

To investigate the performance of a machine in the task of optimally cooling $S$ we wish to calculate the increase in the ground state population of $S$ as a result of our protocol. The initial population was set to be $p_0 = \frac{1}{1+e^{-\beta_S \omega}}$ allowing us to express the difference in populations as 
\begin{align}
    \Delta p_0 &= p'_0 - p_0 = \frac{1}{\mathcal{Z}_S\mathcal{Z}_{M}}\left(\sum_{i_M \in \{0,1\}^{\times n} \, \setminus \, \mathbb{S}} e^{-\beta_M i_M \cdot \Gamma} + \sum_{j_M \in \mathbb{S}} e^{-\beta_S \omega - \beta_M (j_M \oplus 1)\cdot \Gamma}\right) - \frac{1}{\mathcal{Z}_{S}}\\
    &= \frac{1}{\mathcal{Z}_S\mathcal{Z}_{M}}\left(\sum_{i_M \in \{0,1\}^{\times n} \, \setminus \, \mathbb{S}} e^{-\beta_M i_M \cdot \Gamma} + \sum_{j_M \in \mathbb{S}} e^{-\beta_S \omega - \beta_M (j_M \oplus 1)\cdot \Gamma} - \mathcal{Z}_{M}\right).
\end{align}
But we can also express the machine's partition function using our bitwise notation as \begin{align}\mathcal{Z}_{M} &= \prod^{n}_{j=1} (1 + e^{-\beta_{M} \gamma_j}) = \sum_{l_{M} \in \{0,1\}^{\times n}}e^{-\beta_{M} l_{M} \cdot \Gamma}\\ \label{partitionFunctionBitStrings}
&= \sum_{l_{M} \in \{0,1\}^{\times n} \setminus \mathbb{S}}e^{-\beta_{M} l_{M} \cdot \Gamma} + \sum_{l_{M} \in \mathbb{S}}e^{-\beta_{M} l_{M} \cdot \Gamma},\end{align}
where we have split the partition function into a contribution from $S$ and $\{0,1\}^{\times n} \setminus \mathbb{S}$ since these set complements form natural disjoint subsets whose union is $\{0,1\}^{\times n}.$ This reorganisation of the partition function allows us to see that the populations of the groundstate subspace that were \textit{unchanged} during cooling cancel out with a contribution of the partition to give   
\begin{align}
    \Delta p_0 = \frac{1}{\mathcal{Z}_S\mathcal{Z}_{M}}\left(\sum_{j_M \in \mathbb{S}} \,e^{-\beta_S \omega - \beta_M (j_M \oplus 1)\cdot \Gamma} -  \sum_{l_{M} \in \mathbb{S}}e^{-\beta_{M} l_{M} \cdot \Gamma}\right).
\end{align}
We now have a difference of two sums over the same set meaning $l_{S}$ is a free index which we can set to $l_{S} = j_{M}$ and recalling that $E(j_M \oplus 1) = (j_M \oplus 1) \cdot \Gamma = E_\text{Max} - j_M \cdot \Gamma$ we obtain
\begin{align}
    \Delta p_0 &= \frac{1}{\mathcal{Z}_S\mathcal{Z}_{M}}\left(\sum_{j_M \in \mathbb{S}} \left[e^{-\beta_S \omega - \beta_M (E_\text{Max} - j_M \cdot \Gamma)} -  e^{-\beta_{M} j_{M} \cdot \Gamma}\right]\right) \label{changeInGroundPopulation}
\end{align}

\subsection{Change in Population Lower Bounds}
The most basic lower bound on this quantity is given by the contribution $j_M = 1^n$ corresponding to the $\ket{0_S 1 \dots 1}$ ket which is swapped in every cooling scenario and is the largest contributor to the change in population with $1^n \cdot \Gamma = \sum^n_{j=1} \gamma_j = E_\text{Max}$ giving the lower bound 
\begin{align}
    \frac{e^{-\beta_S\omega}-e^{-\beta_ME_{\text{Max}}}}{\mathcal{Z}_S\mathcal{Z}_{M}} \leq \Delta p^*_0,
\end{align}
where $\Delta p^*_0$ is the optimal change in ground state population of the system qubit. This contribution is due to what is known as the \textit{virtual qubit subspace swap} in the literature~\cite{brunner_12,clivaz_pre_2019,clivaz_prl_2019}.

As mentioned in the main text, if one considers the case that $(T_M/T_S) \omega \leq \gamma_1$, then a tighter lower bound can be found by noting that all bit-strings above and including $i_M = 10^{n-k-1}1^{k}$, where $k=n/2$ if $n$ is even and $k=(n-1)/2$ if $n$ is odd, are guaranteed to satisfy Eq.~(\ref{eq:ineq_simp}) for all $\Gamma$. To see this, inputting $i_M$ into Eq.~(\ref{eq:ineq_simp}) says that the population of $i_M$ should be swapped with $i_M \oplus 1$ if 
\begin{equation}
    \frac{T_M}{T_S}\omega < \sum_{j=\Tilde{k}}^n \gamma_j - \sum_{j=2}^{\Tilde{k}-1} \gamma_j + \gamma_1, ~ ~ ~  \Tilde{k} =\begin{cases}
 \frac{n}{2}+1, & n ~ ~ {\rm even} \\
  \frac{n}{2}+\frac{3}{2}, & n ~ ~ {\rm odd} \label{firstDisoreredBitString}
  \end{cases}
\end{equation}
Consider now the assumption that $\Gamma$ is ordered in non-decreasing order, so that $\gamma_i \leq \gamma_j$ if $i \leq j ~ \forall ~ i,j$. From this it can be seen that 
\begin{equation}
    \gamma_1 \leq \sum_{j=\Tilde{k}}^n \gamma_j - \sum_{j=2}^{\Tilde{k}-1} \gamma_j + \gamma_1.
\end{equation}
As we have also assumed that $(T_M/T_S) \omega \leq \gamma_1$, Eq.~(\ref{firstDisoreredBitString}) must be satisfied. All bit-strings above $i_M$ then lead to an Eq.~(\ref{eq:ineq_simp}) with a right hand side larger than Eq.~(\ref{firstDisoreredBitString}), as for all other bit-strings above $i_M$ parts of the negative sum in Eq.~(\ref{firstDisoreredBitString}) are turned positive and hence it can only be larger. Therefore, under the assumption that $(T_M/T_S) \omega \leq \gamma_1$, at least all the bit-strings above and including $i_M$ should be swapped. 

To find a lower bound using Eq.~\ref{changeInGroundPopulation} and this set of bit-strings, one can first consider performing all swaps from $\ket{0^{n-k}1^{k}} \rightarrow \ket{1^n}$, and then removing the contribution arising from the inclusion of $\ket{0^{n-k}1^{k}}$. Initially, one therefore has 
\begin{equation}
   \begin{split}
        \Delta p^{\rm all}_0 &= \frac{1}{\mathcal{Z}_S \mathcal{Z}_M} \bigg( \sum_{i_M \in \{0,1\}^{\times (n-k)}} e^{-\beta_S\omega -\beta_M(E_{\rm Max} - i_M^{n-k} 1^k \cdot \Gamma)} - e^{-\beta_M i_{M}^{n-k} 1^k \cdot \Gamma} \bigg), \\
        &= \frac{1}{\mathcal{Z}_S \mathcal{Z}_M} \bigg( e^{-\beta_S \omega} \sum_{i_M \in \{0,1\}^{\times (n-k)}} e^{-\beta_M(i_M^{n-k} 1^k \oplus 1) \cdot \Gamma} ~ - ~  \sum_{i_M \in \{0,1\}^{\times (n-k)}} e^{-\beta_M i_{M}^{n-k} 1^k \cdot \Gamma} \bigg), 
    \end{split}
\end{equation}
where we have used the fact that $E(i_M) + E(i_M \oplus 1) = E_{\rm Max}$. We now introduce the notation $\Gamma_{1:l} = (\gamma_1, \gamma_2, \ldots, \gamma_l)$ to select parts of $\Gamma$. Using this, the above expression becomes 
\begin{equation}
    \begin{split}        
        \Delta p^{\rm all}_0 &= \frac{1}{\mathcal{Z}_S \mathcal{Z}_M} \bigg( e^{-\beta_S \omega} \sum_{i_M \in \{0,1\}^{\times (n-k)}} e^{- \beta_M(i_M^{n-k} \oplus 1) \cdot \Gamma_{1:n-k}}  ~ - ~  e^{-\beta_M \sum_{l = \Tilde{k}}^n \gamma_l} \sum_{i_M \in \{0,1\}^{\times (n-k)}} e^{-\beta_M i_{M}^{n-k} \cdot \Gamma_{1:n-k}} \bigg), \\
          &= \frac{1}{\mathcal{Z}_S \mathcal{Z}_M} \big( e^{-\beta_S \omega}- e^{-\beta_M \sum_{l = \Tilde{k}}^n \gamma_l} \big) \sum_{i_M \in \{0,1\}^{\times (n-k)}} e^{-\beta_M (i_{M}^{n-k} \cdot \Gamma_{1:n-k})}, 
   \end{split}  
\end{equation}
where the constants from the fixed $1$s have been factorised out and, the fact that the sets $\{i_M \oplus 1 : i_M \in \{0,1\}^{\times n}\}$ and $\{0,1\}^{\times n}$ are equivalent means that 
\begin{equation}
    \sum_{i_M \in \{0,1\}^{\times (n-k)}} e^{-\beta_M(i_M^{n-k} \oplus 1) \cdot \Gamma_{1:n-k}} = \sum_{i_M \in \{0,1\}^{\times (n-k)}} e^{-\beta_M(i_M^{n-k} \cdot \Gamma_{1:n-k})}. 
\end{equation}
From here, it is then noted that this is the partition function over the $(n-k)$ hottest qubits and hence $\Delta p_0^{\rm all}$ can be written as 
\begin{equation}
    \begin{split}  
          \Delta p^{\rm all}_0 &= \frac{1}{\mathcal{Z}_S \mathcal{Z}_M} \big( e^{-\beta_S \omega} - e^{-\beta_M \sum_{l = \Tilde{k}}^n \gamma_l} \big) \prod_{j=1}^{\Tilde{k}-1} (1+e^{-\beta_M \gamma_j}),  \\
        &= \frac{1}{\mathcal{Z}_S\prod_{j=\Tilde{k}}^{n} (1 + e^{-\beta_M \gamma_j})} \big( e^{-\beta_S \omega} - e^{-\beta_M \sum_{l=\Tilde{k}}^n \gamma_l} \big), \\
        &= \frac{1}{\mathcal{Z}_S\mathcal{Z}^{\Tilde{k}:n}_M}\big( e^{-\beta_S \omega} - e^{-\beta_M \sum_{\Tilde{k}:n}^n \gamma_i} \big),  ~ ~ \mathcal{Z}^{\Tilde{k}:n}_M = \prod_{j=\Tilde{k}}^{n} (1 + e^{-\beta_M \gamma_j}) \label{allSwapLowerBound}
   \end{split}  
\end{equation}

One now needs to take into account the change in population included in Eq.~(\ref{allSwapLowerBound}) from swapping $\ket{0^{n-k}1^{k}}$. Performing this swap would give a change in population of 
\begin{equation}
    \Delta p^{\rm corr}_0 = \frac{1}{\mathcal{Z}_S \mathcal{Z}_M} \big( e^{-\beta_S\omega-\beta_M(E_{\rm max} - \sum_{\tilde{k}}^n \gamma_i)} - e^{-\beta_M \sum_{\tilde{k}}^n \gamma_i} \big).
\end{equation}
A lower bound on $\Delta p^*_0$ in this case can then be found as
\begin{equation}
    \begin{split}
        \Delta p^*_0 &\geq \Delta p_0^{\rm all} - \Delta p_0^{\rm corr}, \\
        &\geq \frac{1}{\mathcal{Z}_S\mathcal{Z}^{\Tilde{k}}_M}\big( e^{-\beta_S \omega} - e^{-\beta_M \sum_{\Tilde{k}}^n \gamma_i} \big) + \frac{1}{\mathcal{Z}_S \mathcal{Z}_M} \big( e^{-\beta_M \sum_{\Tilde{k}}^n \gamma_i} - e^{-\beta_S \omega -\beta_M(E_{\rm max} - \sum_{\Tilde{k}}^n \gamma_i)} \big).
    \end{split}
\end{equation}
In the main text we have then introduced the notation that $E(k) = \sum_{i=k}^{n} \gamma_i$ for brevity. 

\section{Reducibility}
\label{app:reduc}

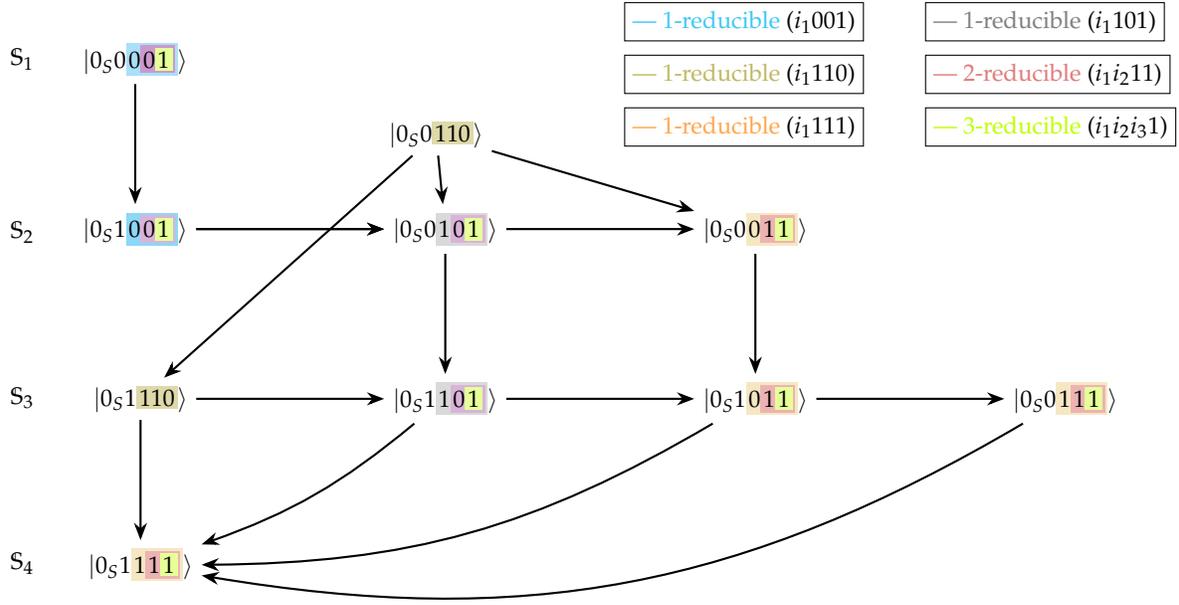
\begin{figure}[h]
    \centering
\scalebox{1}{\begin{tikzpicture}[node distance=1.6cm and 2.5cm]
\setlength{\fboxsep}{1pt} 
\node (s1) at (0,0) {\(|0_S0\colorbox{cyan!30}{0\colorbox{violet!40}{0\colorbox{lime!40}{1}}}\rangle\)};
\node (s2) [below=of s1] {\(|0_S1\colorbox{cyan!40}{0\colorbox{violet!30}{0\colorbox{lime!40}{1}}}\rangle\)};
\node (c1) at (4,-1) {\(|0_S0\colorbox{olive!30}{110}\rangle\)};
\node (a1) [right=of s2] {\(|0_S0\colorbox{gray!30}{1\colorbox{violet!30}{0\colorbox{lime!40}{1}}}\rangle\)};
\node (a2) [right=of a1] {\(|0_S0\colorbox{yellow!40}{0\colorbox{red!40}{1\colorbox{lime!40}{1}}}\rangle\)};
\node (a3) [right=of a2] {};
\node (a4) [below=of a1] {\(|0_S1\colorbox{gray!30}{1\colorbox{violet!30}{0\colorbox{lime!40}{1}}}\rangle\)};
\node (a5) [right=of a4] {\(|0_S1\colorbox{yellow!40}{0\colorbox{red!40}{1\colorbox{lime!40}{1}}}\rangle\)};
\node (a6) [right=of a5] {\(|0_S0\colorbox{yellow!40}{1\colorbox{red!40}{1\colorbox{lime!40}{1}}}\rangle\)};
\node (s3) [left=of a4] {\(|0_S1\colorbox{olive!30}{110}\rangle\)};
\node (s4) [below=of s3] {\(|0_S1\colorbox{yellow!40}{1\colorbox{red!40}{1\colorbox{lime!40}{1}}}\rangle\)};
\node(S1) at (-1.5,0) {$\mathbb{S}_1$};
\node(S2) at (-1.5,-2.3) {$\mathbb{S}_2$};
\node(S3) at (-1.5,-4.5) {$\mathbb{S}_3$};
\node(S4) at (-1.5,-6.7) {$\mathbb{S}_4$};

\draw[arrow]  (s1) -- (s2);
\draw[arrow] (s2) -- (a1);
\draw[arrow] (s2) -- (a1);
\draw[arrow] (a1) -- (a2);
\draw[arrow] (c1) to (a1);
\draw[arrow] (c1) to (a2);
\draw[arrow] (c1) to (s3);
\draw[arrow] (s3) -- (a4);
\draw[arrow] (a4) -- (a5);
\draw[arrow] (a5) -- (a6);
\draw[arrow] (a1) -- (a4);
\draw[arrow] (a2) -- (a5);
\draw[arrow] (s3) -- (s4);
\draw[arrow, bend left=10] (a4) to (s4);
\draw[arrow, bend left=15] (a5) to (s4);
\draw[arrow, bend left=20] (a6) to (s4);

\node[draw, anchor=west] at (6.5, 0.5) {\textcolor{cyan!60}{--- 1-reducible} ($i_1$001)};
\node[draw, anchor=west] at (6.5, -0.2)  {\textcolor{olive!60}{--- 1-reducible} ($i_1$110)};
\node[draw, anchor=west] at (6.5, -0.9) {\textcolor{orange!70}{--- 1-reducible} ($i_1$111)};
\node[draw, anchor=west] at (10.5, 0.5) {\textcolor{gray}{--- 1-reducible} ($i_1$101)};
\node[draw, anchor=west] at (10.5, -0.2) {\textcolor{red!70}{--- 2-reducible} ($i_1i_2$11)};
\node[draw, anchor=west] at (10.5, -0.9) {\textcolor{lime}{--- 3-reducible} ($i_1i_2i_3$1)};

\end{tikzpicture}
}
\caption{The set $\mathbb{S}$ of ground state subspace energy levels worth exchanging with energy levels from the excited state subspace to cool $S$ under different conditions.}
\label{fig:4_qubit_reduc}
\end{figure}

The third law in quantum thermodynamics is often interpreted as a statement that cooling a quantum system to its ground state at zero temperature requires diverging resources. One such resource in our setting is the number of constituent qubits $n$ forming our machine, but as we have observed if the energy of these qubits grows too fast it could render energetically weaker parts of the machine useless. In this appendix we give insight into characterising this situation by introducing the notion of \textit{$k$-reducibility}, the scenario where the $k$ energetically weakest qubits of the $n$ qubit machine do not contribute to the cooling scenario. In particular, here we find that degenerate machines always behave irreducibly for some cooling scenario whereas machines whose energy structure grows exponentially by a factor larger than 2 are always $(n-1)$-reducible resulting in a \textit{wasteful} bipartite interaction. 

\begin{definition}
    A machine is said to be \textit{$k$-reducible} if (i) its first $k$ qubits do not contribute to the cooling interaction (ii) \textit{precisely $k$-reducible} if the machine cannot be further reduced i.e. it is not $(k+1)$-reducible. 
    \begin{align}
    &(i) \quad\quad \ket{0_S 1^k l^{n-k}} \in \mathbb{S} \implies \ket{0_S i^k l^{n-k}} \in \mathbb{S} \,\, \forall i^k \in \{0,1\}^{\times k} \nonumber \\
    &(ii) \quad\quad \exists \,\, \ket{0_S 1^{k+1} l^{n-k-1}} \in \mathbb{S} \,: \, \ket{0_S i^{k+1} l^{n-k-1}} \notin \mathbb{S} \, \, \forall \, i^{k+1} \in \{0,1\}^{\times k+1} \nonumber \end{align}
\end{definition}
Energetically, this definition implies that a $k$-reducible machine is one which results in a cooling scenario where if $\ket{0_s 1^k 1^{n-k}}$ corresponds to a population that should be exchanged so must every other energetically weaker machine energy level $\ket{0_S i^k 1^{n-k}}$. In particular $\ket{0_S 0^k 1^{n-k}} \in \mathbb{S}$ giving \begin{align}
    (i) \quad\quad \frac{1}{2} \left(\frac{T_M}{T_S}\omega + E_\text{Max}\right) < \sum^{n}_{j=k+1} \gamma_j,
\end{align}
which is a necessary condition on the energetic structure of the machine for $M$ to be $k$-reducible. On the other hand from (ii) we find that a precisely $k$-reducible machine would have an energetic structure satisfying 
\begin{align}
    (ii) \quad \quad i^{k+1}l^{n-k-1} \cdot \Gamma\leq \frac{1}{2} \left(\frac{T_M}{T_S}\omega + E_\text{Max}\right) < 1^{k+1}l^{n-k-1}\cdot \Gamma,
\end{align}
for some $i^{k+1} \in \{0,1\}^{\times k+1}$ and $l^{n-k-1}\in \{0,1\}^{\times n-k-1}.$

\subsection{Irreducible Machines} With this definition in mind an irreducible $n$ qubit machine, one whose energetic structure grows with increasing $n$ ensuring that all its $n$ qubits contribute to the cooling scenario, cannot be 1-reducible. This will be ensured if condition $(i)$ is never satisfied for $k=1$ that is an irreducible machine induces a swappable set $\mathbb{S}$ which features at least one $\ket{0_S1 l^{n-1}} \in \mathbb{S}$ and $\ket{0_S 0 l^{n-1}} \notin \mathbb{S}$. Or equivalently, if condition $(ii)$ is always satisfied for $k=0$. Energetically, (ii) being satisfied for $k=0$ implies that there is always at least one energy level $\ket{0_S1l^{n-1}}$ that must be exchanged for cooling whilst its bit conjugate $\ket{0_S0l^{n-1}}$ does not cool when exchanged giving the inequality
\begin{align}
    0l^{n-1} \cdot \Gamma\leq \frac{1}{2} \left(\frac{T_M}{T_S}\omega + E_\text{Max}\right) < 1l^{n-1}\cdot \Gamma.
\end{align}
This inequality can be restated in two ways which give different insight into when irreducibility occurs. A machine is irreducible if there is at least one energy level $l^{n-1} \in \{0,1\}^{\times n -1}$ of the $n-1$ coldest qubits of the machine which satisfies 
\begin{gather}
    0 \leq \frac{1}{2}\left(\frac{T_M}{T_S}\omega + E_\text{Max}\right) - E(l^{n-1}) < \gamma_1, \label{eq:irred_cond}
\end{gather}
where $\gamma_1$ is the energy of the warmest qubit in the machine and 
\begin{equation}
    \gamma_2 \leq E(l^{n-1}) \leq \sum_{i=2}^n \gamma_i.
\end{equation}
Alternatively by noting that $E_\text{Max} - 2 E(l^{n-1}) = E(l^{n-1} \oplus 1) -E(l^{n-1})$ one can restate this inequality as a set of intervals given by energy changes induced by swapping Hamming weight conjugate pairs,
\begin{align}
         E(l^{n-1}) - E(l^{n-1}\oplus 1) &\leq \frac{T_M}{T_S}\omega < E(l^{n-1}) - E(l^{n-1}\oplus 1) + 2\gamma_1. \label{eq:num_line}
\end{align}
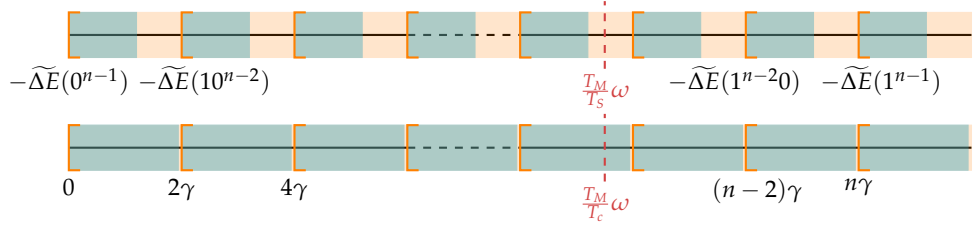
\begin{figure}[t]
\begin{tikzpicture}[scale=1.5]


  \draw[black,thick] (-1,0.5) -- (2,0.5);
  \draw[black,thick,dashed] (2,0.5) -- (3,0.5);
  \draw[black,thick] (3,0.5) -- (7,0.5);

  \foreach \x in {-1,0,1,2,3,4,5,6} {
    \draw[cyan,opacity=0.4, fill opacity = 0.4, fill=cyan] (\x,0.3) -- (\x,0.7) -- ({\x + 0.6},0.7) -- ({\x + 0.6},0.3) -- cycle;
    \draw[orange, thick] ({\x + 0.1},0.7) -- ({\x},0.7) -- ({\x},0.3) -- ({\x + 0.1},0.3);
  }

  \node[below] at (-1,0.3) {$-\widetilde{\Delta E}(0^{n-1})$};
  \node[below] at (0.2,0.3) {$-\widetilde{\Delta E}(10^{n-2})$};
  \node[below] at (4.9,0.3) {$-\widetilde{\Delta E}(1^{n-2}0)$};
  \node[below] at (6.2,0.3) {$-\widetilde{\Delta E}(1^{n-1})$};

  \draw[orange,opacity=0.2, fill opacity = 0.2, fill=orange] (-1,0.3) -- (-1,0.7) -- (7,0.7) -- (7,0.3) -- cycle;
  \draw[red, thick, dashed] (3.75,0.2) -- (3.75,0.8);
  \node[red] at (3.75,0) {\(\frac{T_M}{T_S}\omega\)};


  \draw[black,thick] (-1,-0.5) -- (2,-0.5);
  \draw[black,thick,dashed] (2,-0.5) -- (3,-0.5);
  \draw[black,thick] (3,-0.5) -- (7,-0.5);

  \foreach \x in {-1,0,1,2,3,4,5,6} {
    \draw[cyan,opacity=0.4, fill opacity = 0.4, fill=cyan] 
      (\x,-0.7) -- (\x,-0.3) -- ({\x + 0.97},-0.3) -- ({\x + 0.97},-0.7) -- cycle;
    \draw[orange, thick] ({\x + 0.1},-0.3) -- ({\x},-0.3) -- ({\x},-0.7) -- ({\x + 0.1},-0.7);
  }

  \node[below] at (-1,-0.7) {$0$};
  \node[below] at (0,-0.7) {$2\gamma$};
  \node[below] at (1,-0.7) {$4\gamma$};
  \node[below] at (5.1,-0.7) {$(n-2)\gamma$};
  \node[below] at (6,-0.7) {$n\gamma$};

  \draw[orange,opacity=0.2, fill opacity = 0.2, fill=orange] (-1,-0.7) -- (-1,-0.3) -- (7,-0.3) -- (7,-0.7) -- cycle;

  \draw[red, thick, dashed] (3.75,-0.8) -- (3.75,-0.2);
  \node[red] at (3.75,-1) {\(\frac{T_M}{T_c}\omega\)};

\end{tikzpicture}

\caption{Here we visualise the inequality Eq.~\eqref{eq:num_line} which turns irreducibility into a question of whether $\frac{T_M}{T_S}\omega$ falls within an interval $[-\widetilde{\Delta E}, -\widetilde{\Delta E} + 2\gamma_1)$ or not. Remarkably, for degenerate machines $\widetilde{\Delta E}(l^{n}+1) - \widetilde{\Delta E}(l^{n-1}) = 2\gamma$ so the blue region is equal to the distance between differences, implying that a degenerate machine will always behave irreducibly for some cooling scenario $\frac{T_M}{T_S}\omega,$ as there are no gaps between blue regions.}
\label{fig:num_line}
\end{figure}

This inequality can be visualised by taking the number line, marking the negative energy differences of Hamming weight conjugate energies $- \widetilde{\Delta E}(l^{n-1}) =  E(l^{n-1}) - E(l^{n-1}\oplus 1)$ and considering intervals of $[-\widetilde{\Delta E}, -\widetilde{\Delta E} + 2\gamma_1)$ and asking if $\frac{T_M}{T_S}\omega$ falls in these regions. If it does, we have an irreducible machine. The distance between neighbouring differences say $\widetilde{\Delta E}(l^{n-1})$ and $\widetilde{\Delta E}(l^{n-1}+1)$, where $l^{n-1}+1$ is an abuse of notation to denote the next string above $l^{n-1}$ in lexicographic order, could be much larger than $[-\widetilde{\Delta E}, -\widetilde{\Delta E} + 2\gamma_1)$ giving the possibility of reducibility. 

\subsubsection{When are Degenerate Machines reducible}
\label{app:degen_reduc}

Using the criteria we have derived we are now in a position to prove that degenerate machines, such as those used in algorithmic and dynamical cooling~\cite{schulman_limits,bassman_campisi_24} are irreducible and so increase there cooling potential as the number of constituents forming the machine increases. Firstly the simple condition Eq.~\eqref{eq:irred_cond} takes the form 
\begin{gather}
    0 \leq \frac{1}{2}\left(\frac{T_M}{T_S}\omega + n\gamma - 2t\gamma\right) < \gamma,
\end{gather}
where $E_\text{Max} = n \gamma$, $E(l^{n-1}) = t\gamma$ and $t \in [0,n-1]$. Simplifying, we arrive at
\begin{gather}
    -(n-2t) \gamma \leq \frac{T_M}{T_S}\omega \leq  2\gamma - (n-2t) \gamma.
\end{gather}
As $0<\omega$, it must be the case that the left hand side is always positive, which is true if $n/2 < t.$ Hence, one gets a set of conditions for each integer $n/2 \leq t \leq n-1$, where each $t$ gives a condition that is some interval on the real-numbers. In the degenerate case, these intervals capture the whole number-line $[0, n \gamma]$, meaning whenever $(T_M/T_S) \omega \leq n \gamma$, one can find a $t$ such that the above condition is satisfied. More specifically, under the credible assumptions that $T_M \leq T_S$ and $\omega \leq \gamma$, degenerate machines are therefore always irreducible.

As a second approach we consider energy changes between Hamming weight conjugate pairs and see that $\widetilde{\Delta E}(l^{n-1}) = E(l^{n-1}\oplus 1) -  E(l^{n-1}) = (n-t)\gamma - t\gamma$ where $t$ is the no. of 1s in $l^{n-1}$. Comparing with a neighbouring energy difference $\widetilde{\Delta E}(l^{n-1} + 1) = (n-t - 1)\gamma - (t+1)\gamma$ we find that the distance between two points on the number line
\begin{gather}
    |\widetilde{\Delta E}(l^{n-1} + 1) - \widetilde{\Delta E}(l^{n-1})| = |(n - 2t - 2)\gamma - (n-2\gamma)| = 2\gamma
\end{gather}
for all $l^{n-1}$ indexed by $t$. Thus in the case of the degenerate machine, we have have the remarkable situation that the distance between neighbouring differences is always $2\gamma$ that is $|[-\widetilde{\Delta E}(l^{n-1}), -\widetilde{\Delta E}(l^{n-1}+1))| = 2\gamma$ for any $l^{n-1} \in \{0,1\}^{\times n -1 }$ which grants degenerate machines the property of being always irreducible. This scenario is visualised in Fig.\ref{fig:num_line}.

\subsection{$(n-1)$-reducible machines}
\label{app:n-1_red}
The most wasteful kind of $n$ qubit machine constructed to cool a quantum system would be one where the optimal cooling interaction is merely bipartite, involving only the energetically strongest (coldest) qubit in the machine and rendering the other $n-1$ qubits in the machine useless. In such a machine, these $n-1$ qubits play no role in the cooling interaction and are in this sense only useful in the machines construction, to ramp up to the energy of the last qubit $\gamma_n$. But if one had access to a singular qubit with energy equal to that of the coldest qubit in an $(n-1)$-reducible machine, one could cool equivalently. Let's investigate under what energetic conditions this can happen. Firstly, condition $(i)$ for $(n-1)$-reducibility states that $\ket{0_S 1^n}$ being an energy level whose exchange is cooling implies that every energetically weaker $\ket{0_S i^{n-1}1}$ is also worth swapping $\forall \, i^{n-1} \in \{0,1\}^{n-1}$. In particular even the weakest level $\ket{0_S0^{n-1}1} \in\mathbb{S}$ giving the first condition for $(n-1)$-reducibility
\begin{gather}
    \frac{1}{2}\left(\frac{T_M}{T_S}\omega + E_\text{Max}\right) < \gamma_n.
\end{gather}
A precisely $(n-1)$-reducible machine must still be able to cool, which is to say it is not $n$-reducible and there is at least one energy $\gamma_j$ for which

\begin{gather}
        0 \leq \frac{1}{2}\left(\frac{T_M}{T_S}\omega + E_\text{Max}\right) - \gamma_j < \gamma_1,  \, : j \in [2,n].
\end{gather}
\subsubsection*{Exponential Machines}
We call an $n$-qubit machine with gap vector $\Gamma = (\gamma,\gamma^2, \gamma^3 \dots, \gamma^n)$ where the energetic gap of each qubit is exponentially larger then that of its predecessor, an exponential machine. Exponential    machines are susceptible to being reducible if the growth of their energy structure is so rapid that a single qubit is able to bring the system to its ground state on its own. To avoid this scenario and figure out if exponential machines can ever be irreducible let's investigate the $k$-reducibility of exponential machines.

To begin with, recall that finite progressions of the geometric series sum to give $\sum^n_{i = 0} ar^i = a\left(\frac{1 - r^{n+1}}{1-r}\right)$ for $r \neq 1$. We can use this to express the sum of gaps in the exponential machine as $E_\text{Max} = \frac{1 - \gamma^{n+1}}{1-\gamma} - 1 = \gamma\left(\frac{1-\gamma^n}{1-\gamma}\right)$. Now checking the $(n-1)$ reducibility we find 
\begin{align}
    &\frac{1}{2}\left(\frac{T_M}{T_S}\omega + \gamma\left(\frac{1-\gamma^n}{1-\gamma}\right)\right) < \gamma^n\\
    &\frac{T_M}{T_S}\omega < \frac{\gamma\left(2\gamma^n - 2\gamma^n - 1 \right)}{1 - \gamma } \label{eq:ineq_exp}
\end{align}
which, since $1 < \frac{\left(2\gamma^{n-1} - 2\gamma^n - 1 \right)}{1 - \gamma }$ for all $1 < \gamma$, means Eq.~\eqref{eq:ineq_exp} holds if $\frac{T_M}{T_S}\omega < \gamma$ and $1 <\gamma$. Hence, this is a sufficient condition for an exponential machines to be $(n-1)$-reducible.

\section{Towards Connections between Cooling \& Error Correcting Codes}
\label{Ap:quantum_error_correction}
\begin{figure}[h]
    \centering
\includegraphics[width=0.75\linewidth]{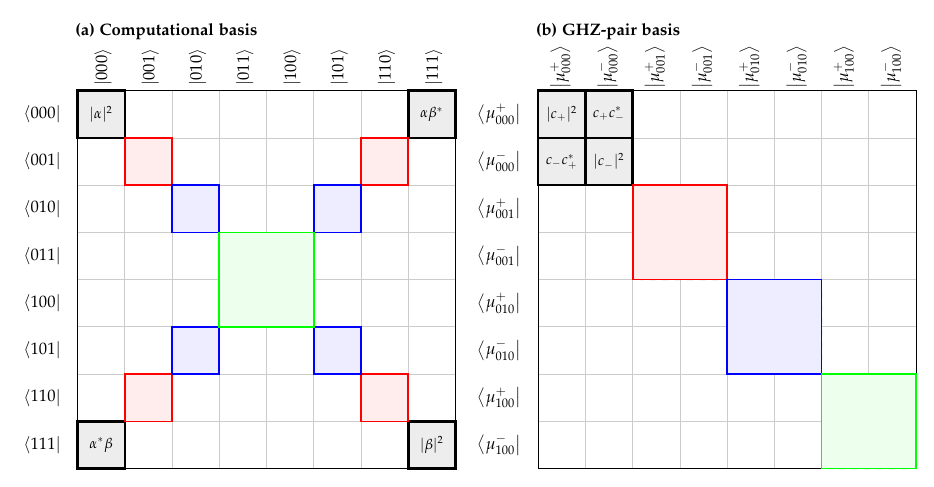}
    \caption{In this figure consider the three qubit repetition code that corrects for single qubit $\sigma_x$ errors.  (a) In the computational basis we have the codewords $\ket{000}, \ket{111}$ giving the encoding $\ket{\psi_L} = \alpha \ket{000} + \beta\ket{111}$. (b) In the GHZ-pair basis we have the codewords $\ket{\mu^+_{000}} = \frac{\ket{000} + \ket{111}}{\sqrt{2}}, \ket{\mu^-_{000}} = \frac{\ket{000} - \ket{111}}{\sqrt{2}}$ giving the encoding $\ket{\psi_L} = c_+ \ket{\mu^+_{000}} + c_- \ket{\mu^-_{000}}$ where $c_\pm = \sqrt{2}(\alpha \pm \beta).$}
    \label{fig:error_cor}
\end{figure} 
Several works in the community have pointed out the connection between error correction and cooling quantum systems. Despite this, explicit connections framing error correcting protocols as quantum cooling protocols, particularly within the framework of thermal machines, are still to be made. In this section we explore this connection using the tools we have developed in this work by examining the most basic quantum error correcting code, the 3 code repetition code, rephrasing it as a cooling problem.

\paragraph*{The Repetition Code} This code encodes a single qubit pure state $\ket{\psi} = \alpha \ket{0} + \beta \ket{1}$ into the 3 qubit subspace spanned by $\ket{000}, \ket{111}$ as $\ket{\psi_L} = \alpha \ket{000} + \beta \ket{111}$. Here, $\ket{000}, \ket{111}$ are known as the codewords of the code and the space they span, the logical subspace or codespace $\mathcal{C}$. The code can detect and correct for any single qubit $\sigma_X$ or bitflip error as any $\sigma^{(i)}_X$ error applied to one of the qubits rotates the encoded state to one of three orthogonal error subspaces as depicted in Fig.\ref{fig:error_cor} (a), $\mathcal{E}_3$ spanned by $\ket{001}, \ket{110}$ if $\sigma^{(3)}_X$ occurs, $\mathcal{E}_2$ given by $\ket{010}, \ket{101}$ if $\sigma^{(2)}_X$ occurs and $\mathcal{E}_1$ spanned by $\ket{100}, \ket{011}$ if $\sigma^{(1)}_X$ occurs. This code has stabilisers $ S_1 =\sigma^{(1)}_z\sigma^{(2)}_z$ and $S_2 = \sigma^{(2)}_z\sigma^{(3)}_z$ as both $S_1$ and $S_2$ leave the codewords invariant but together their expectation values can jointly determine which error has occurred via the syndrome table given below.
\paragraph*{A Change of Basis} In trying to reformulate this protocol into a thermodynamic one, two problems arise. The first, a general state $\ket{\psi}$ encoded for protection against error can have coherences in the code space. This means that the setting we have considered for cooling considering only permutations on the diagonal of a joint system-machine state may not \textit{a priori} suffice. Secondly, the codewords $\ket{000}, \ket{111}$ are very disparate in energy with respect to a local qubit Hamiltonian i.e. the codespace is spanned by the energetically lowest and highest states preventing a straight forward reformulation of the problem into a thermodynamic one. A change of basis may address both issues. Consider the rotation 
\begin{align*}
    \ket{000} \longrightarrow \frac{\ket{000} + \ket{111}}{\sqrt{2}} &&    \ket{111} \longrightarrow \frac{\ket{000} - \ket{111}}{\sqrt{2}},
\end{align*}
on the codewords. In the computational basis, this addresses both issues as is block-diagonalises the coherences of the encoded state and renders the codewords energy degenerate with respect to a local qubit Hamiltonian. For the general unitary, split the computational basis into two disjoint sets of Hamming weight conjugate pairs, in this case $x \in S=\{000,001,010,100\}$ and $\bar{x}_i = x_i \oplus 1 \in \bar{S}$ such that $\{0,1\}^{\times 3} = S \cup \bar{S}$ and let's define $\ket{\mu^\pm_{x}} = \frac{\ket{x} + \ket{\bar{x}}}{\sqrt{2}}$. Then the desired change of basis is obtained via the unitary 
\begin{gather}
    V = \sum_{x \in S} \ketbra{\mu^+_x}{x} + \ketbra{x}{\mu^+_x} + \ketbra{\mu^-_{\bar{x}}}{\bar{x}} + \ketbra{\bar{x}}{\mu^-_{\bar{x}}},
\end{gather}
which can be simply obtained by applying a Hadamard gate to the first qubit, followed by $\mathtt{CNOT}_{12}$ and $\mathtt{CNOT}_{13}$ for any 3 qubit string. This change of basis rotates the code giving the encoding $\ket{\psi_L} = c_+ \ket{\mu^+_{000}} + c_- \ket{\mu^-_{000}}$ where $c_\pm = \sqrt{2}(\alpha \pm \beta),$ whilst preserving the stabilisers as one can see in the below syndrome table.
\begin{table}[h]
\begin{tabular}[t]{|c|c|c|c|}
\hline
Subspace & Span                 & $\langle \sigma^{(1)}_z \sigma_z^{(2)}\rangle$ & $\langle \sigma_z^{(2)}\sigma_z^{(3)} \rangle$ \\ \hline
$\mathcal{C}$& $\ket{000},\ket{111}$ & 1                                              & 1                                              \\ \hline
$\mathcal{E}_1$ & $\ket{100}, \ket{011}$           & -1                                             & 1                                              \\ \hline
$\mathcal{E}_2$&$\ket{010}, \ket{101}$           & -1                                             & -1                                             \\ \hline
$\mathcal{E}_3$&$\ket{001}, \ket{110}$           & 1                                              & -1                                             \\ \hline
\end{tabular}
\quad
\begin{tabular}[t]{|c|c|c|c|}
\hline
Subspace & Span                 & $\langle \sigma^{(1)}_z \sigma_z^{(2)}\rangle$ & $\langle \sigma_z^{(2)}\sigma_z^{(3)} \rangle$ \\ \hline
$\mathcal{C}$& $\frac{\ket{000}\pm\ket{111}}{\sqrt{2}}$ & 1                                              & 1                                              \\ \hline
$\mathcal{E}_1$ & $\frac{\ket{100}\pm\ket{011}}{\sqrt{2}}$           & -1                                             & 1                                              \\ \hline
$\mathcal{E}_2$&$\frac{\ket{010}\pm\ket{101}}{\sqrt{2}}$           & -1                                             & -1                                             \\ \hline
$\mathcal{E}_3$&$\frac{\ket{001}\pm\ket{110}}{\sqrt{2}}$           & 1                                              & -1                                             \\ \hline
\end{tabular}
\end{table}
\paragraph*{Error Correction as a Cooling Problem}
In this basis, an error or correction merely permutes block diagonal subspaces bringing us to a setting closer that considered in the cooling protocols presented in this work. To fully reformulate this problem as a thermodynamic one, we must now introduce a Hamiltonian which energetically discriminates between each subspace. Consider the Hamiltonian 
\begin{gather}
    H_{\textrm{Code}} = \sum_{x \in S} E_x \ketbra{\mu^\pm_x}{\mu^\pm_x},
\end{gather}
where $E_x < E_{x'}$ if $x < x'$ in lexicographic order.  This Hamiltonian energetically sees the states spanning a given subspace as energy degenerate but can energetically discriminate between the codespace and individual error subspaces. In this way, the error free logical state $\ket{\psi_L} = c_+ \ket{\mu^+_{000}} + c_- \ket{\mu^-_{000}}$ is a zero temperature thermal state of this Hamiltonian. The application of an amplitude damping channel i.e. probabilistic bit flip errors results in a thermal distribution over the Hamiltonian $ H_{\textrm{Code}}$ giving a state $\rho_\beta = e^{-\beta H_{\textrm{Code}}}/\mathcal{Z}$. With access to an auxilliary quantum system such as a thermal machine this three qubit system can be unitarily cooled using the methods developed in this work giving 
\begin{gather}
    \rho_{\beta'}=\text{tr}_M\left\{U_{\textrm{Cool}}(\rho_\beta \otimes \rho_M)U^\dagger_{\textrm{Cool}}\right\}
\end{gather}
where $\beta' > \beta$ bringing one closer to the logically encoded state.

\section{Cooling With Symmetries}
\label{Ap:cooling_w_symmetries}

In the main text, we have assumed access to any unitary when cooling the system (although ultimately we have argued that only permutations in the energy eigenbasis are necessary). Here, we consider restricting the allowed dynamics due to the existence of some symmetry, and assess its effect on our ability to optimally cool the system.

Given some observable $Q$, we will now consider unitaries $U$ such that $[U, Q]=0$. Under such dynamics $Q$ is said to be symmetric as the commutator implies that $Q = U Q U^\dagger$. Physically, this means that under evolution via $U$ the probabilities of getting a given measurement outcome of $Q$ will remain a constant. As a consequence, the expectation value of $Q$ will also remain constant. Hence, this symmetry in the dynamics is said to have led to the conservation of $Q$, with this being an example of Noether Theorem.

\subsection{Energy Conservation}

Firstly, we will consider $H_T=H_S + H_M$, where
\begin{equation}
	H_T = \sum_{j_S \in \{0,1\}} \sum_{i_M \in \{0,1\}^n} \big( j_S \cdot \omega + i_M \cdot \Gamma) \ketbra{j_S i_M}{j_S i_M},
\end{equation}
such that $H_S, H_M$ are the Hamiltonians of the system and machine respectively, so that $H_T$ is the total Hamiltonian. The symmetry in the dynamics $U$, i.e., $[U, H_T]=0$, then leads to the total energy of the system and machine being conserved. A unitary $U$ commutes with $H_T$ if it acts an independent unitary within each degenerate subspace of $H_T$,
\begin{equation}
	U = \bigoplus_{c} U_{c}, ~ ~U_{c}U_{c}^\dagger=\mathbb{I}
\end{equation}
such that
\begin{equation}
	U_{c} = \sum_{j_S \in \{0,1\}} \sum_{i_M \in \{0,1\}^n} u_{i_S,i_M}\ketbra{j_S i_M}{j_S i_M} ~ ~ : ~ ~ j_S \cdot \omega + i_M \cdot \Gamma = E_c,
\end{equation}
where $E_c$ is the energy of the $c$th degenerate subspace and $u_{i_S,i_M}$ is some matrix coefficient of $U_c$ in the eigenbasis of $H_T$. Now applying the same logic as in the main text to each degenerate subspace, it can be seen that optimal cooling can be achieved by applying permutations of the energy eigenbasis within each degenerate subspace of $H_T$. As before, a permutation can only change the ground state population of the system qubit if it is of the form
\begin{equation}
	\ket{0_Si_M} \leftrightarrow \ket{1_Sk_M} : i_M, k_M \in \{0,1\}^{\times n}. \label{eq: necessary perms}
\end{equation}
Hence, to achieve any cooling there must exists at least one $i_M, k_M \in \{0,1\}^{\times n}$ such $i_M \cdot \Gamma = \omega + k_M \cdot \Gamma$, as this means there will exists a degenerate subspace of $H_T$ that can facilitate such a TLP. As before, this TLP will be cooling if 
\begin{gather}
	\frac{T_M}{T_S}\omega < (i_M - k_M)\cdot \Gamma. \label{eq: is_TLP_cooling}
\end{gather}
Although, due to the above equality, this will be satisfied if $T_S > T_M$. This result is expected, as the thermal state of $H_T$ is
\begin{equation}
	\tau_{SM} = \frac{1}{\mathcal{Z}_{SM}}e^{-\beta H_T} = \frac{e^{-\beta H_S}}{\mathcal{Z}_S} \otimes \frac{e^{-\beta H_M}}{\mathcal{Z}_M},
\end{equation}
such that the system and machine are locally thermal, each at the some temperature $\beta$. As energy conserving dynamics act invariantly on the thermal state i.e., $U \tau_{SM} U^\dagger = \tau_{SM}$ if $[U,H_T]=0$, no cooling of the system can be achieved if $T_M=T_S$ (as no populations can be moved). However, if the system and machine are at different temperatures, they are not initially in a thermal state of $H_T$. Hence, the initial state is not invariant under energy conserving unitaries and there exists the possibility to cool. Therefore, under energy conserving dynamics, the system can only be cooled using a machine if (a) $H_T$ has degeneracies facilitating TLP of the form eq.~\ref{eq: necessary perms}, and (b) the initial temperature of the machine is less than the initial temperature of the system, $T_S < T_M$.

The change in ground state temperature that can be achieved under energy conserving unitaries will then depend on both the dimension of the degenerate subspaces that can facilitate permutations of the form of eq.~\ref{eq: necessary perms} --- as this dictates the number of TLPs that can be performed within the given degenerate subspace --- and the temperature difference between the system and machine. To see this, it is noted that the change in the ground state population from performing the TLP $\ket{0_Si_M} \leftrightarrow \ket{1_Sk_M}$, within the $E_c$th degenerate subspace, is given by 
\begin{equation}
	\begin{split}
		\Delta p_c &= \frac{e^{-\beta_s \omega-\beta_M k_M \cdot \Gamma} - e^{-\beta_M i_M \cdot \Gamma}}{\mathcal{Z}_S \mathcal{Z}_M} \\
		&= \frac{e^{-\beta_M E_c} \big(e^{\omega(\beta_M - \beta_S)}-1\big)}{\mathcal{Z}_S \mathcal{Z}_M}, \label{eq: TLP in deg space}
	\end{split}
\end{equation}
where we have used the fact that $i_M \cdot \Gamma = \omega + k_M \cdot \Gamma = E_c$. Hence, it can be seen that the larger the disparity between $T_S$ and $T_M$, the larger $\beta_M-\beta_S$, and the larger the possible change in ground state energy.

To solidify the notions introduced here, we will now consider energy conserving unitaries in the motivating example of a two qubit machine considered in the main text. Given the assumption of the gap vector being ordered in non-increasing order, $\omega \leq \gamma_1 \leq \gamma_2$, there are two possible ways in which $H_T$ can have degeneracies that facilitate TLPs of the form of eq.~\ref{eq: necessary perms}: (1) $\omega+\gamma_1 = \gamma_2$, and (2) $\omega = \gamma_1 = \gamma_2 = \gamma$.

In case (1), there is a single degenerate subspace: ${\rm Span}\big\{\ket{0_S01}, \ket{1_S10}\big\}$. Hence, the TLP $\ket{0_S01} \leftrightarrow \ket{1_S10}$ is all that can be performed, leading to an optimal change in ground state population of the system of 
\begin{equation}
	\Delta p^*(1) = \frac{e^{-\beta_M \gamma_2} \big(e^{\omega(\beta_M - \beta_S)}-1\big)}{\mathcal{Z}_S \mathcal{Z}_M}.
\end{equation}
In case (2), there are two degenerate subspaces:
\begin{equation}
	\begin{split}
		{\rm Span}\big\{\ket{0_S01}, \ket{0_S10}, \ket{1_S00}\big\}, \\
		{\rm Span}\big\{\ket{0_S11}, \ket{1_S01}, \ket{1_S10}\big\}. 
	\end{split}
\end{equation}
In the first subspace, either the TLP $\ket{0_S01} \leftrightarrow \ket{1_S00}$ or $\ket{0_S10} \leftrightarrow \ket{1_S00}$ must be performed. An equivalent choice of TLPs then exists for the other degenerate subspace. Taken together, they lead to an optimal change in ground state population of 
\begin{equation}
	\Delta p^*(2) = \frac{\big(e^{-\beta_M \omega} + e^{-2\beta_M \omega} \big)  \big(e^{\omega(\beta_M - \beta_S)}-1\big)}{\mathcal{Z}_S \mathcal{Z}_M}.
\end{equation}

\subsection{Additional Conserved Quantities}

In addition to energy, one can also consider other observables that must be conserved under the cooling dynamics. For example, one can consider symmetry with respect to
\begin{equation}
	J_Z = \sigma_Z^S \otimes \mathbb{I}_M + \mathbb{I}_S \otimes \sum_{i=1}^n \sigma_Z^i, ~ ~ \sigma_Z^i = \mathbb{I}_1 \otimes \mathbb{I}_2 \hdots \mathbb{I}_{i-1} \otimes \sigma_Z \otimes \mathbb{I}_{i+2} \hdots \mathbb{I}_n,
\end{equation}
which is the total spin observable with constants set $1$. As before, the unitaries that conserve total spin are those that commute with $J_Z$, $[U, J_Z]=0$, and hence they act invariantly within the degenerate subspaces of $J_Z$. If both $H$ and $J_Z$ are conserved, one therefore needs to consider unitaries $U$ such that \hbox{$[U,H]=[U,J_Z]=0$}. Such a condition can be very restrictive, allowing arbitrary unitaries to act only within the joint degenerate eigenspaces of $H$ and $J_Z$. Although, if one assumes that the energy eigenvectors and spin eigenvectors align, such that $[H, J_Z]=0$, the problem becomes more tractable. For example, in the case of the two-qubit motivating example, the total spin operator in the joint eigenbasis is under this assumption is
\begin{equation}
	J_Z = {\rm diag}(3,1,1,-1,1,-1,-1,-3).
\end{equation}
Hence, it has the same degenerate subspaces as case (2) detailed above. Therefore, in case (2) both $H_T$ and $J_Z$ can be conserved whilst performing cooling. However, in case (1) there is no overlapping degenerate subspaces between $J_Z$ and $H_T$, meaning it is not possible to perform any cooling whilst conserving both quantities.
\\
\\
An interesting future extension in this direction would be to consider cooling with multiple conserved quantities, as was briefly explored in~\cite{silva2024optimalunitarytrajectoriescommuting}. Whilst having two conserved quantities in the above example of case(1) proved restrictive enough to remove any ability to cool, in larger systems it could still be possible to perform cooling despite added restriction. In addition, it would be interesting to characterise in what situations symmetric dynamics become so restrictive that cooling cannot occur. Moreover, it would be timely and interesting to also consider situations in which the conserved quantities do not commute e.g., $[H_T, J_Z]\neq 0$. Considering non-commuting conserved quantities in the process of thermalisation has already proved interesting \cite{Majidy2023}, it would now pertinent to consider their role in opposite task: cooling.

\section{Cooling a Qubit with $n$ \textit{interacting} Other Qubits -- General Purification}\label{appendix: Generalisation to Arbitrary States}

Throughout this work we have focused on the task of cooling a qubit, which was shown to be equivalent to increasing the probability of being in the ground state of the system Hamiltonian. However, given that we allow the application of arbitrary unitaries (when not considering symmetries), once we have increased the probability of being in the ground state we could arbitrarily change the basis of the state. As cooling is only applicable when considering the energy eigenbasis, and increasing the purity of a qubit just means increase its largest eigenvalue, it would be more appropriate to view the protocol detailed here as a {\em purity enhancement} protocol. If our restriction to unitaries that are permutations of the energy eigenbasis was a physical restriction, rather than just a result of them being sufficient for cooling, our protocol would be strictly cooling.

In the same vein, our protocol can easily be adapted to perform purity enhancement of an arbitrary system qubit given access to a machine of arbitrary qubits i.e., it can be applied more generally than just to thermal states. This can be achieved given that all qubit states can be considered to be a thermal state with respect to some Hamiltonian. Specifically, given some state $\rho$, with spectral decomposition
\begin{equation}
    \rho = \lambda_0 \ketbra{\lambda_0}{\lambda_0} + (1-\lambda_0) \ketbra{\lambda_1}{\lambda_1}, ~ ~ \lambda_0 \geq 2^{-1},
\end{equation}
an effective Hamiltonian can be defined as
\begin{equation}
    H^{\mathrm{eff}} = \gamma \ketbra{\lambda_1}{\lambda_1}, ~ ~ \gamma = -\frac{1}{\beta} \ln{\bigg(\frac{1}{\lambda_0} - 1\bigg)},
\end{equation}
where $\beta$ is some effective temperature that can be chosen to tune the gap. Hence, given some arbitrary set of states 
\begin{equation}
    \rho_S \otimes \rho_M = \rho_S \otimes \rho_{M_1} \otimes \rho_{M_2} \otimes \rho_{M_3} \otimes \ldots \otimes  \rho_{M_n}, 
\end{equation}
one can define an effective Hamiltonian for the system and each machine qubit such that 
\begin{equation}
    \rho_{S} = \frac{e^{-\beta_S H^{\mathrm{eff}}_S}}{Z_S} ~ ~{\mathrm{and}} ~ ~  \rho_{M_i} = \frac{e^{-\beta_M H^{\mathrm{eff}}_i}}{Z_i} ~ \forall~i,
\end{equation}
with total effective Hamiltonian 
\begin{equation}
    H_T^{\mathrm{eff}} = H_{S}^{\mathrm{eff}} + \sum_{i=1}^n H_i^{\mathrm{eff}}.
\end{equation}
From this, an effective system gap, $\omega^{\mathrm{eff}}$, and effect gap vector, $\Gamma^{\mathrm{eff}}$, can be defined. The gaps in the machine can be ordered such that the gap vector is in non-decreasing order. Then, if the chosen system qubit has $\lambda_0 \neq 2^{-1}$, such that the state is not maximally mixed, a $\beta_S$ can always be chosen for any $\rho$ such that $\omega^{\mathrm{eff}} \leq \gamma^{\mathrm{eff}}_1 : \gamma^{\mathrm{eff}}_1 \in \Gamma^{\mathrm{eff}}$. The framework is therefore the same as presented above, meaning that the same protocol can be performed. Although, TLPs of the joint eigenbasis of the system and machine are now considered, rather than TLPs of the energy eigenbasis e.g., TLPs of the form 
\begin{equation}
\ket{\lambda_0^S\lambda_{i_1}^{M_1}\lambda_{i_2}^{M_2} \hdots \lambda_{i_n}^{M_n}} \leftrightarrow \ket{\lambda_1^S\lambda_{j_1}^{M_1}\lambda_{j_2}^{M_2} \hdots \lambda_{j_n}^{M_n}},    
\end{equation}
where $\{ \vert \lambda_0^{S} \rangle \langle \lambda_0^{S} \vert, \vert \lambda_1^{S} \rangle \langle \lambda_1^{S} \vert \}$ and $\{ \vert \lambda_0^{M_k} \rangle \langle \lambda_0^{M_k} \vert, \vert \lambda_1^{M_k} \rangle \langle \lambda_1^{M_k} \vert \}$ are the eigenbasis of the system (the chosen qubit being purified) and the $k$th machine qubit respectively. Therefore, $i_M = i_1i_2 \hdots i_n$ and $j = j_1j_2 \hdots j_n$ are bit-strings labelling the eigenstates of $H_T$ which are tensor products of the local eigenstates. These bit-string are then used in determining which inequalities are satisfied i.e., which TLPs should be performed to purify (cool) the system. The purifying unitary generated by our protocol will then be in this joint eigenbasis, but can of course be decomposed into any basis once it has been found. Finally, we emphasis that the local eigenbasis of each qubit can be different, with the eigenbasis of different qubits not needing to be orthogonal.

In this alternative setup, all of the above results still apply, with the change in ground state population now being reinterpreted as a change in purity. Hence, given access to any number of qubits, our protocol can be used to find the unitary that maximally increases the purity of some chosen target qubit.

\section{Constructing cooling unitaries using Hypercube Graphs}
\label{app:hypercube_graph}
\subsection{Hypercube Graphs}
The graph $G = (V,E)$ with the set of vertices $V$ corresponding to the set of $n$-bit strings $\{0,1\}^{\times n}$ and edges $E$ connecting vertices $v_j$, $v_k$ if the bit strings $j,k$ differ only at one bit i.e. they are Hamming distance 1 from each other-- is known as $Q_n$ the $n$-bit hypercube graph~\cite{hypercube_survey,godsil01}. These graphs have $2^n$ vertices each having $n$ edges giving a total of $n2^{n-1}$ edges on the graph. These graphs are well studied in the field of combinatorics due to their simple structure and ubiquity in classical coding theory. Note that any Hamiltonian path on $Q_n$ corresponds to a Gray code between two bit strings.

\subsection*{Representing Cooling Problems on Hypercube Graphs}
The state of system and machine \begin{gather}
    \tau_S \otimes \tau_M = \hspace{-0.5cm}\sum_{\substack{i_S \in \{0,1\}\\ i_M \in \{0,1\}^{\times n }}} \frac{e^{-(\beta_S i_S \omega + \beta_M i_M \cdot \Gamma)}}{\mathcal{Z}_{SM}} \ketbra{i_S i_M}{i_S i_M},
\end{gather}
exists in a Hilbert space spanned by eigenkets which are indexed by the set of $n+1$ bit strings $i_Si_M \in \{0,1\}^{\times n}$. As such, the state $\tau_S \otimes \tau_M$ can be represented on an $n+1$ Hypercube graph since each eigenket $\ket{i_S i_M}$ corresponds to a vertex $v_j \in V = \{0,1\}^{\times n +1}$. To represent the state we add weights $w(v_j)$ to each vertex to give the vertex set $W$ where $w(v_j) = \frac{e^{i_S\omega + i_M\cdot \Gamma}}{\mathcal{Z}_S \mathcal{Z}_M}$ for $ v_j = i_Si_M \in V$ mapping each bitstring to its corresponding population in to give $\tau_S \otimes \tau_M$, where $Q^{W}_{n+1} = (V,E, W)$ the weighted Hypercube graph of $n+1$ bitstrings.

The disordered ground state subspace energy levels $\ket{0_S k}$ which need be exchanged for cooling correspond to the subgraph $K \subseteq Q_{n+1}$ with vertices $v_{0k}$ with $n$ bitstrings $k\in \mathbb{S}$ satisfying the condition 
\begin{gather}
    \frac{1}{2}\left(\frac{T_M}{T_S}\omega + E_\text{Max}\right) < k\cdot \Gamma.
\end{gather}
This in turn defines the subgraph $K \oplus 1 \subseteq Q_{n+1}$ with vertices $v_{1k\oplus1}$ of disordered excited state energy levels $\ket{1_S k \oplus 1}$ satisifying the property 
\begin{gather}
   (k\oplus 1) \cdot  \Gamma \leq \frac{1}{2}\left(\frac{T_M}{T_S}\omega + E_\text{Max}\right)
\end{gather}
ensuring that whenever the vertices in $Q_{n+1}$ are permuted such that a $v_l \in K$ and $v_m \in K \oplus 1$ are exchanged $S$ is cooled. Note that since the bit flip operation is an automorphism of $Q_{n+1}$, $K \oplus 1$ is a reflection of $K$ through an axis in $Q_{n+1}$ i.e. $K \simeq K \oplus 1$ the subgraphs of disordered populations are isomorphic to each other. 

\subsection{Cooling Unitaries as Minimum Weight Perfect Matchings}

Any \textit{optimal} cooling unitary $U_\text{cool}$ that cools $S$ with access to $M$ to a state with the largest obtainable ground state population \begin{gather} \Delta p^* = \max_{U_{SM} \in \mathcal{SU}(2^{n+1})} \bra{0_S} \text{tr}_M \left\{ U_{SM}(\tau_S \otimes \tau_M) U_{SM}^{\dagger} \right\}\ket{0_S}\end{gather} under unitary interaction must necessarily take the form 
\begin{gather}
    U_\text{cool} = \ketbra{A}{B} + \ketbra{B}{A} + \mathbb{1}_\text{Rest}
\end{gather}
which exchanges the set of disordered ground state subspace energy levels $\ket{A} = \{\ket{0_S i_M} \, | \, i_M \in \mathbb{S}\}$ with the disordered excited subspace energy levels $\ket{B} = \{\ket{1_S i_M \oplus 1} \, | \, i_M \in \mathbb{S}\}$.

As described in Sec.~\ref{sec:unitary} a central problem for constructing cooling unitaries is that there are $|\mathbb{S}|!$ such unitaries that achieve $\Delta p^*$ with potentially different resource costs as quantified by a chosen cost function e.g. gate complexity, dissipation or energetic cost. In this section we also gave a proof that the \textit{the Hamming weight difference between strings $j$ and $k$ lower bounds the gate complexity of a circuit $S_{j,k}$ that permutes their corresponding eigenkets} for the gate set $\mathcal{G} = \{\mathtt{X, CX, CCX, \dots , C^{n}X}\}.$ Minimising the sum of Hamming weight differences $D_H(a,b)$ between all energy levels to be exchanged for cooling would give a cooling unitary $U^*_\text{cool}$ which not only obtains the optimal ground state population $\Delta p^*$ but does so with minimal gate complexity, \begin{gather}
U^*_\text{cool} = \sum_{\substack{\ket{a} \in \ket{A} \\ \ket{b} \in \ket{B}}} \ketbra{a}{b} + \ketbra{b}{a} +\mathbb{1}_\text{Rest} \, : \, \min_{\substack{\ket{a} \in \ket{A} \\ \ket{b} \in \ket{B}}} \sum_{\substack{\ket{a} \in \ket{A} \\ \ket{b} \in \ket{B}}} d_H(a,b).
\end{gather}
But, \textit{how does one pair all the $\ket{a} \in \ket{A}$ and $\ket{b} \in \ket{B}$ to achieve this}?

\paragraph*{Perfect Matching} In graph theory~\cite{godsil01} a \textit{matching} $M \subseteq E$ is a subset of the edge set of a graph where each vertex is incident to only one edge in $M$, that is no two edges in $M$ share a common vertex. A matching is said to be \textit{perfect} if every vertex of the graph is incident once to a single edge in the matching. In other words, a perfect matching $M$ is a collection of edges pairing each vertex to exactly one other vertex in the graph.

If a graph $G = (V, E, W_E)$ also has an edge weight set $W_E$ given by some cost function  $w(E)$ then a \textit{minimum weight perfect matching} is a perfect matching $M^*$ on the graph $G$ where the sum of the weights of the edges in the matching is the minimum out af all possible perfect matchings on $G$ i.e.
\begin{gather}
 M^{*} = \arg  \min_{\substack{M \subseteq E \\ M \text{ perfect}}} \sum_{e \in M} w(e).
\end{gather}

\begin{lemma}
    There is a one-to-one correspondence between edges of the bipartite graph $G(K, K\oplus 1, E)$ and energy level exchanging unitaries $S_{j,k}$, where $K = \{k \in \{0,1\}^{\times n} | \, 0k \in \mathbb{S}\}$ is the subgraph related to disordered ground state subspace energy levels and $K \oplus 1 = \{k' \in \{0,1\}^{\times n} |k' = 1k\oplus1 \, \forall \, k \in K\}$ disordered excited state subspace energy levels.
\label{lemma:lem_graph}
\end{lemma}

\begin{proof}Let $\mathfrak S$ denote the set of unitary basis ket permutations exchanging exactly two levels 
\begin{gather}
\mathfrak{S} = \{S_{j,k} \in \mathcal{SU}(2^n) \, |  S_{j,k} = \ketbra{j}{k} + \ketbra{k}{j} + \mathbb{1}_\text{Rest},\, \, j < k \},\end{gather}
and consider the function $f : \mathfrak{S} \rightarrow E$ which maps the unitary $S(j,k) \in \mathfrak{S}$ to the edge $\{v_j , v_k\} \in E \subseteq \{0,1\}^{\times n}\times\{0,1\}^{\times n}.$ To show that this function is bijective we must show that it is injective and surjective. For injectivity consider two edges $\{v_j,v_k\} = \{v_l,v_m\}$, since $v_j, v_k, v_l, v_m \in K \times K\oplus 1$ then $j,k,l,m \in \{0,1\}^{\times n}$ and so the assumption implies these $n$ bitstrings are equal $j =l$,  $k =m$. The $n$ bit strings index the $2^n$ bit strings of the machine index the eigenkets of the machine so $\ket{j} = \ket{l}, \ket{k} = \ket{m}$ and we find $S_{j,k} = \ketbra{j}{k} + \ketbra{k}{j} + \mathbb{1}_\text{Rest} =\ketbra{l}{m} + \ketbra{m}{l} + \mathbb{1}_\text{Rest} = S_{l,m}$ giving injectivity i.e. $f(x) = f(y) \implies x = y$. For surjectivity we wish to show that $\forall \, e \in E, \exists \,  S_{l,m} \in \mathfrak{S} \, : \, f(S_{l,m}) = e.$ To see this consider an arbitrary $e = \{v_j,v_k\}$ corresponding to $j,k \in \{0,1\}^{\times n}$ then clearly one can construct $S_{j,k} = \ketbra{j}{k} + \ketbra{k}{j} + \mathbb{1}_\text{Rest}$ such that $f(S_{j,k}) = \{v_j, v_k\} = e$ as desired. Therefore $f$ is a bijection. \end{proof}

As a direction consequence of this lemma we have the main statement of this section, that the problem of finding $U^*_\text{cool}$ is equivalent to a minimum weight perfect matching problem.

\begin{proposition}
    The problem of finding an optimal cooling unitary $U^*_\text{cool}$ which minimises the gate complexity relative to $\mathcal{G} = \{\mathtt{X,CX, C^2X, \dots, C^n X}\}$ to cool $S$ with gap $\omega$ at temperature $T_S$ given $n$ qubits with gaps $\Gamma = \{\gamma_1, \gamma_2, \dots, \gamma_n\}$ at temperature $T_M$ is equivalent to finding a minimum weight perfect matching on the bipartite graph $G(K, K\oplus 1, E, d_H)$.
\end{proposition}
\begin{proof}
    Firstly, $U^*_\text{cool}$ can be expressed as a product $U^*_\text{cool} = S_{a_{|\mathbb{S}|}, b_{|\mathbb{S}|}}\dots S_{{a_2, b_2}}S_{a_1, b_1}$ of $S_{j,k}$ since each $\ket{a} \in \ket{A}$ is paired to exactly one $\ket{b} \in \ket{B}$ so none of the pairs intersect. Then, by Lemma~\ref{lemma:lem_graph} each $S_{j,k}$ can be represented as an edge $e = \{v_j, v_k\}$ in the bipartite graph $G(K, K\oplus 1, E)$, giving a collection of edges $M$. From the definition of $U^*_\text{cool}$ we can infer that $M$ is a matching on $G(K, K\oplus 1, E)$ since each $\ket{a} \in \ket{A}$ corresponds to $v_a \in K$ and each $\ket{b} \in \ket{B}$ corresponds to $v_b \in K\oplus 1$ exhausting $K$ and $K \oplus 1$ such that no vertex is incident on more than one edge. Lastly considering $G(K, K\oplus 1, E,d_H)$ where the Hamming distance $D_H(j,k)$ between two strings $j,k$ labelling vertices $v_j,v_k$ is added as a weight $w(e) = D_H(v_j,v_k)$ and the definition of $U^*_\text{cool}$ where the pairing of eigenkets being exchanged satisfies
    \begin{align}
        \min_{\substack{\ket{a} \in \ket{A} \\ \ket{b} \in \ket{B}}} \sum_{\substack{\ket{a} \in \ket{A} \\ \ket{b} \in \ket{B}}} D_H(a,b)
    \end{align}
    which implies by construction $M$ is minimum weight perfect matching on $G$ since it is a matching satisfying $\min_{\substack{M' \subseteq E \\ M \text{ perfect}}} \sum_{e \in M'} w(e)$ since $w(e) = D_H(v_j,v_k)$ and $U^*_\text{cool}$ minimises the Hamming weight across pairs of exchanged energy levels. Therefore any $U^*_\text{cool}$ corresponds to a minimum weight perfect matching.
\end{proof}

\subsection{The Hungarian Algorithm for Perfect Matching \& Complexity of Cooling}

\begin{figure}[h] 
    \centering
    \begin{minipage}[c]{0.48\textwidth}
        \centering
        \begin{tikzpicture}[scale=2, baseline=(current bounding box.center)]
\def\radius{0.85}

\node[draw, circle] (C1) at (0, 0, 0) {$1$}; 
\node[anchor=north east] at (0, 0, 0) {$\ket{0000}$};
\node[draw, circle] (D1) at (1, 0, 0) {$\frac{1}{32}$}; 
\node[anchor=north west] at (1, 0, 0) {$\ket{0001}$};
\node[draw, circle] (E1) at (0, 1, 0) {$\frac{1}{22}$}; 
\node[anchor=south east] at (0, 1, 0) {$\ket{0010}$};
\node[draw, circle, fill = lred] (F1) at (1, 1, 0) {$\frac{1}{704}$}; 
\node[anchor=south west] at (1, 1, 0) {$\ket{0011}$};
\node[draw, circle] (G1) at (0, 0, 1) {$\frac{1}{3}$}; 
\node[anchor=north east] at (0, 0, 1) {$\ket{0100}$};
\node[draw, circle, fill = lred] (H1) at (1, 0, 1) {$\frac{1}{96}$}; 
\node[anchor=north west] at (1, 0, 1) {$\ket{0101}$};
\node[draw, circle, fill = lred] (A1) at (0, 1, 1) {$\frac{1}{66}$}; 
\node[anchor=south east] at (0, 1, 1) {$\ket{0110}$};
\node[draw, circle, fill = lred] (B1) at (1, 1, 1) {$\frac{1}{2112}$}; 
\node[anchor=south west] at (1, 1, 1) {$\ket{0111}$};

\node[draw, circle, fill = lred] (C2) at (2, 0, 0) {$\frac{1}{2}$}; 
\node[anchor=north east] at (2, 0, 0) {$\ket{1000}$};
\node[draw, circle, fill = lred] (D2) at (3, 0, 0) {$\frac{1}{64}$}; 
\node[anchor=north west] at (3, 0, 0) {$\ket{1001}$};
\node[draw, circle, fill = lred] (E2) at (2, 1, 0) {$\frac{1}{44}$}; 
\node[anchor=south east] at (2, 1, 0) {$\ket{1010}$};
\node[draw, circle] (F2) at (3, 1, 0) {$\frac{1}{1408}$}; 
\node[anchor=south west] at (3, 1, 0) {$\ket{1011}$};
\node[draw, circle, fill = lred] (G2) at (2, 0, 1) {$\frac{1}{6}$}; 
\node[anchor=north east] at (2, 0, 1) {$\ket{1100}$};
\node[draw, circle] (H2) at (3, 0, 1) {$\frac{1}{192}$}; 
\node[anchor=north west] at (3, 0, 1) {$\ket{1101}$};
\node[draw, circle] (A2) at (2, 1, 1) {$\frac{1}{132}$}; 
\node[anchor=south east] at (2, 1, 1) {$\ket{1110}$};
\node[draw, circle] (B2) at (3, 1, 1) {$\frac{1}{4224}$}; 
\node[anchor=south west] at (3, 1, 1) {$\ket{1111}$};

\draw (C1) -- (D1); 
\draw (C1) -- (E1); 
\draw (C1) -- (G1);
\draw (D1) -- (F1); 
\draw (D1) -- (H1);
\draw (E1) -- (F1); 
\draw (E1) -- (A1);
\draw (F1) -- (B1); 
\draw (G1) -- (H1);
\draw (G1) -- (A1); 
\draw (H1) -- (B1);
\draw (A1) -- (B1);

\draw (C2) -- (D2); 
\draw (C2) -- (E2); 
\draw (C2) -- (G2);
\draw (D2) -- (F2); 
\draw (D2) -- (H2);
\draw (E2) -- (F2); 
\draw (E2) -- (A2);
\draw (F2) -- (B2); 
\draw (G2) -- (H2);
\draw (G2) -- (A2); 
\draw (H2) -- (B2);
\draw (A2) -- (B2);

\draw[dashed, gray, thin] (D1) -- ++(0.4,0); \draw[dashed, gray, thin] (C2) -- ++(-0.4,0); 
\draw[dashed, gray, thin] (H1) -- ++(0.4,0);
\draw[dashed, gray, thin] (G2) -- ++(-0.4,0);
\draw[dashed, gray, thin] (B1) -- ++(0.4,0);
\draw[dashed, gray, thin] (A2) -- ++(-0.4,0); 
\draw[dashed, gray, thin] (F1) -- ++(0.4,0);\draw[dashed, gray, thin] (E2) -- ++(-0.4,0);

\draw[dashed, gray, thin] (A1) -- ++(-0.5, 0, 0);
\draw[dashed, gray, thin] (G1) -- ++(-0.5, 0, 0);
\draw[dashed, gray, thin] (E1) -- ++(-0.5, 0, 0);
\draw[dashed, gray, thin] (C1) -- ++(-0.5, 0, 0);
\draw[dashed, gray, thin] (H1) -- ++(-0.25, -0.25, 0);
\draw[dashed, gray, thin] (G1) -- ++(-0.25, -0.25, 0);

\draw[dashed, gray, thin] (B2) -- ++(0.5, 0, 0);
\draw[dashed, gray, thin] (D2) -- ++(0.5, 0, 0);
\draw[dashed, gray, thin] (F2) -- ++(0.5, 0, 0);
\draw[dashed, gray, thin] (H2) -- ++(0.5, 0, 0);
\draw[dashed, gray, thin] (H2) -- ++(-0.25, -0.25, 0);
\draw[dashed, gray, thin] (G2) -- ++(-0.25, -0.25, 0);
\end{tikzpicture}
    \end{minipage}
    \begin{minipage}[c]{0.48\textwidth}
        \centering
        \begin{tikzpicture}[scale=1, baseline=(current bounding box.center)]
    \def\xleft{0}
    \def\xright{3}

    \node[draw, circle, fill=lred] (A) at (\xleft, 3) {$\frac{1}{66}$};
    \node[anchor=east] at (\xleft, 3) {$\ket{0110}$};

    \node[draw, circle, fill=lred] (B) at (\xleft, 2) {$\frac{1}{2112}$};
    \node[anchor=east] at (\xleft, 2) {$\ket{0111}$};

    \node[draw, circle, fill=lred] (C) at (\xleft, 1) {$\frac{1}{704}$};
    \node[anchor=east] at (\xleft, 1) {$\ket{0011}$};

    \node[draw, circle, fill=lred] (D) at (\xleft, 0) {$\frac{1}{96}$};
    \node[anchor=east] at (\xleft, 0) {$\ket{0101}$};

    \node[draw, circle, fill=lred] (E) at (\xright, 3) {$\frac{1}{2}$};
    \node[anchor=west] at (\xright, 3) {$\ket{1000}$};

    \node[draw, circle, fill=lred] (F) at (\xright, 2) {$\frac{1}{64}$};
    \node[anchor=west] at (\xright, 2) {$\ket{1001}$};

    \node[draw, circle, fill=lred] (G) at (\xright, 1) {$\frac{1}{44}$};
    \node[anchor=west] at (\xright, 1) {$\ket{1010}$};

    \node[draw, circle, fill=lred] (H) at (\xright, 0) {$\frac{1}{6}$};
    \node[anchor=west] at (\xright, 0) {$\ket{1100}$};

    \foreach \leftnode in {A,B,C,D} {
        \foreach \rightnode in {E,F,G,H} {
            \draw[thick] (\leftnode) -- (\rightnode);
        }
    };
    \draw[orange,thick] (A) -- (E);
    \draw[orange,thick] (B) -- (F);
    \draw[orange,thick] (C) -- (G);
    \draw[orange,thick] (D) -- (H);
    \end{tikzpicture}
    \end{minipage}
\caption{On the left we see the hypercube graph visualisation $Q_4$ for one of the scenarios given in Example 2 App-.\ref{app:example_2} and on the right a visualisation of the complete bipartite graph $G(K, K\oplus 1, E)$ in this scenario with a minimum weight perfect matching in orange.}
\label{fig:matchings}
\end{figure}
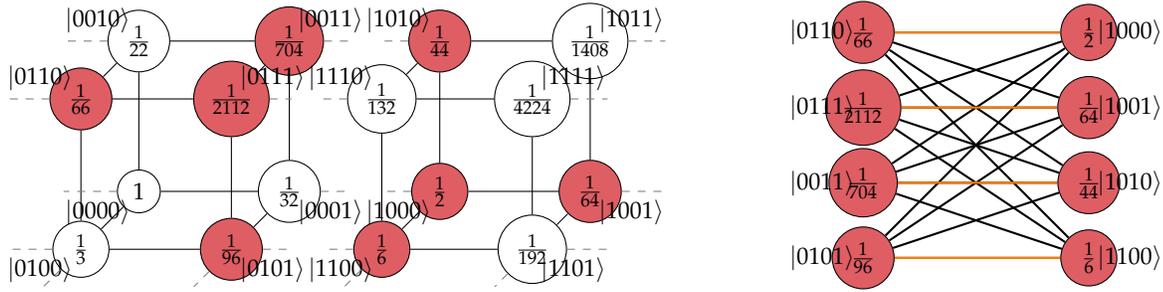

Having now shown that the problem of finding a unitary to cool $S$ given $M$ that minimises the gate complexity given the gate set $\mathcal G$ can be recast as a minimum weight perfect matching problem we now describe how MWPMs can be found algorithmically and discuss complexity theoretic implications on quantum cooling in light of this connection. 

Minimum weight perfect matching is a well studied and seminal problem in the field of combinatorial optimisation with applications in network theory. For a more in-depth introduction see~\cite{comb_opt}. For complete bipartite graphs, like $G(K,K\oplus1, E, d_H)$ the bipartite graph connecting the subgraphs of disordered populations $K, K\oplus 1 \subset Q_{n+1}$ (as visualised in Fig~\ref{fig:matchings}), minimum weight perfect matchings can be found using the Hungarian Algorithm~\cite{comb_opt}.

The matrix formulation of this algorithm for a complete weighted bipartite graph $B(V_1,V_2, E, C)$ with $2m$ vertices is as follows. Consider the $m \times m$ cost matrix $C$ with entries $c_{ij} = c(v^{(1)}_i,v^{(2)}_j) \, \forall \, v^{(1)}_i \in V_1, v^{(2)}_j \in V_2$ where $c(\cdot)$ is cost function determining the weight of an edge $\{v^{(1)}_i,v^{(2)}_j\} \in E$. A perfect matching corresponds to picking $m$ numbers from this matrix picking once from each row and column of the matrix. Without loss of generality one can fix the columns of $C$ and permute its $m$ rows using the permutation matrix $P_{\pi} \, : \, \pi \in \mathcal{S}_m$. This defines a new matrix $\tilde{C} = P_\pi C$ whose diagonal corresponds to the cost of a given perfect matching and so the trace of $\tilde{C} = P_\pi C$, $\text{tr}\{\tilde{C}\} = \sum^m_{i = 1} \tilde{c}_{ii} = \sum_{e \in M} c(e)$ corresponds to the weight of this perfect matching $M$. Minimising this trace over all permutations 
\begin{gather}
 \min_{\pi \in \mathcal{S}_m} \text{tr}\{P_\pi C\} = \text{tr}\{P_{\pi^*}C\}
\end{gather}
one finds the minimum weight perfect matching $M^*$ given by the set of edges $E^* = \{\{ \{v^{(1)}_i, v^{(2)}_{\pi^*(i)} \, : \, i \in (1,2,\dots,m) \}\}$. The time complexity of this algorithm is $\mathcal{O}(m^3)$~\cite{comb_opt}.

\paragraph*{Example} Let's make the connection between cooling unitaries and minimum weight perfect matching more tangible by returning to examples from Appendix~\ref{app:example_2}. In b) we observed case 2 i.e. $\omega < \gamma_1 + \gamma_2 - \gamma_3$ for $T_M = T_S$ which we visualise in Fig.~\ref{fig:matchings} for $\omega = \log 2, \gamma_1 = \log 3, \gamma_2 = \log 22 $ and $\gamma_3 = \log 32$ factored by $\frac{32}{69}$. Here we have the disoredered energy levels $\ket{A} = \{\ket{0_S 011}, \ket{0_S 101}, \ket{0_S110}, \ket{0_S111}\}$ and $\ket{B} = \{\ket{1_S 100}, \ket{1_S 010}, \ket{1_S001}, \ket{1_S000}\}$. Thinking through the Hungarian Algorithm for our example we have the cost matrix $C$ with entries given by the Hamming distance between the bitstrings labelling energy levels 
\begin{gather}
C =
\begin{pNiceMatrix}[first-row,first-col]
    & \ket{0011} & \ket{0101} & \ket{0110} & \ket{0111}       \\
\ket{1100} & 4   & \textcolor{lred}{2}   & 2   & \textcolor{blue}{3}    \\
\ket{1010} & \textcolor{lred}{2}   & 4   & \textcolor{blue}{2}   & 3   \\
\ket{1001} & 2   & \textcolor{blue}{2}   & 4   & \textcolor{lred}{3}    \\
\ket{1000} & \textcolor{blue}{3}   & 3   & \textcolor{lred}{3}   & 4   \\
\end{pNiceMatrix},
\end{gather}
a minimum perfect weight matching now corresponds to picking 4 entries from this matrix using every row and column once such that the sum of these 4 entries is as low as possible. A valid MWPM (highlighted in red) gives the exchanges $\ket{0_S 011}\leftrightarrow\ket{1_S 010} = 2, \ket{0_S 101}\leftrightarrow\ket{1_S 100} = 2, \ket{0_S 111}\leftrightarrow\ket{1_S 001} = 3$ and $\ket{0_S 110}\leftrightarrow\ket{1_S 000} = 3$ and $2 + 2 + 3 + 3 =10$ minimising the sum of interest. This corresponds to the unitary
\begin{gather}
    U^*_\text{cool} = (S_{0\textcolor{red}{01}1, 1\textcolor{red}{01}0} \oplus S_{0\textcolor{blue}{10}1,1\textcolor{blue}{10}0}) \oplus (S_{011\textcolor{red}{1}, 100\textcolor{red}{1}} \oplus S_{011\textcolor{blue}{0},100\textcolor{blue}{0}}) = \mathrm{SWAP}_{S,M_3} \oplus S_{{011,100}_{SM_1M_2}}
\end{gather}
where we see a complexity reduction has been obtained since the first pair of two level swaps features two free indices leading to a bipartite $\rm SWAP$ and the second pair features one free index corresponding to an irreducible 4-partite operation. 

An issue is that there can be multiple minimum weight perfect matchings obtained through the minimisation over $C$ each of which correspond to cooling unitaries with equal capability to potentially minimise gate complexity but do not have equal circuit complexity as some might lead to unitaries with lower order of interaction. This is the case comparing with the other valid solution to this example highlighted in blue given by the exchange $\ket{0_S111} \leftrightarrow \ket{1_S100}, \ket{0_S110} \leftrightarrow \ket{1_S010}, \ket{0_S101} \leftrightarrow \ket{1_S001}, \ket{0_S011} \leftrightarrow \ket{1_S000}$ giving the unitary 

\begin{gather}
U'_\text{cool} = S_{01\textcolor{red}{1}0,10\textcolor{red}{1}0} \oplus S_{01\textcolor{blue}{0}1,10\textcolor{blue}{0}1} \oplus S_{0111,1100} \oplus S_{0011,1000} = S_{{010,100}_{SM_1M_3}} \oplus S_{0111,1100} \oplus S_{0011,1000}
\end{gather}
which is more complex than $U^*_\text{cool}.$ 

\paragraph*{Optimising Energetic Cost} Adding an energetic constraint to reduce the change in energy in the machine to the cost matrix by considering the cost function $c(i,j) = D_H(i,j)(|i-j|\cdot \Gamma)$ one can optimise the \textit{work cost} of the cooling unitary which also turns out to make the \textit{optimal} MWPM unique in this example for $\omega =1, \Gamma= (1,2,4)$ unique giving 

\begin{gather}
C =
\begin{pNiceMatrix}[first-row,first-col]
    & \ket{0011} & \ket{0101} & \ket{0110} & \ket{0111}       \\
\ket{1100} & 16   & \textcolor{lred}{6}   & 2   & \textcolor{blue}{15}    \\
\ket{1010} & \textcolor{lred}{6}   & 8   & \textcolor{blue}{0}   & 12   \\
\ket{1001} & 2   & \textcolor{blue}{0}   & -8   & \textcolor{lred}{6}    \\
\ket{1000} & \textcolor{blue}{15}   & 12   & \textcolor{lred}{6}   & 24   \\
\end{pNiceMatrix},
\end{gather}
 where we see that the cost of $U^*_\text{cool}$ is 24 which is lower than the cost of $U'_\text{cool}$ which is 30. Of course the reduction in the number of viable MWPMs one gets by choosing this additional constraint depends on the structure of $\Gamma$, it is left to future work to determine the fascinating interplay which we allude to here between unitary complexity, energetic cost and minimum weight perfect matchings for cooling unitaries.

\subsubsection*{Complexity of finding an Optimal Cooling Unitary}

The Hungarian method for obtaining a minimum weight perfect matching on a complete bipartite graph of $2m$ nodes require $\mathcal{O}(m^3)$ operations~\cite{comb_opt}. In our setting the complete bipartite graph formed by the disordered energy levels $G(K,K\oplus 1, E, d_H)$ has $2|\mathbb{S|}$ vertices and so finding $U^*_\text{cool}$ has an algorithmic complexity of $\mathcal{O}(|\mathbb{S}|^3)$ and since $0 \leq |\mathbb{S}| \leq 2^{n}$ depending on the structure of $\Gamma$, where $n$ is the number of machine qubits, the complexity of this task can be anywhere from polynomial to exponential. 

Note that before obtaining $U^*_\text{cool}$ via minimum weight perfect matching an agent need necessarily acquire knowledge of $\mathbb{S}$ by examining at most $2^{n}$ inequalities. Thus the complexity of cooling is more accurately of order $\mathcal{O}(|\mathbb{S}|) + \mathcal{O}(|\mathbb{S}|^3)$ where the first contribution stems from determining the disordered populations by verifying the inequalities described in prior sections. 

\subsection{Automorphisms of the Hypercube Graph}

A graph automorphism is an operation on the graph that preserves the vertex-edge connectivity structure in the graph. The automorphism group of a graph $\mathrm{Aut}(G)$ is the set of all such operations. Whilst $U_\text{cool}$ is generally not an element of the automorphism group studying the structure of this group can give us insights into which gates are needed for which cooling scenarios. 

The automorphism group of the hypercube graph $\mathrm{Aut}(Q_n)$ is the hyperoctohedral group~\cite{godsil01} which can be expressed as $\mathbb{Z}^n_2 \rtimes \mathcal{S}_n$ where $\mathbb{Z}^n_2$ is the group of operations that add a bit string $x \in \{0,1\}^n$ modulo 2 to the bit strings corresponding to the vertices of the graph effectively given by the possible bit flips and $\mathcal{S}_n$ is the group of permutations of $n$ objects. 

\paragraph*{The Two Qubit Machine Example - Aut$(Q_3)$}
In the given example, we have $z_x \in \mathbb{Z}^3_2$ such that $z_x(v_j) = j_0j_1j_2 \oplus x_0x_1x_2$ for all bit strings $j$ and since there are $2^3$ such bit strings $x$ we can add then $|\mathbb{Z}^3_2|= 2^3$. And $S_3 = \{(1)(2)(3),(12)(3),(13)(2),(123),(1)(23),(132)\}$ in cycle notation such that $\pi_{123} \ket{abc} = \ket{cab}$. Therefore $\mathrm{Aut}(Q_3)$ has $48 = 2^3 \times 3!$ elements given by the product of these group which in this case can easily be visualised as the possible rotations on the cube (including inversions). In this way $\mathbb{Z}^n_2 \rtimes S_n$ in general has order $2^n \times n!$ . Of course this does not give one the full set of $2^n!$ permutations of the diagonal elements of the density matrix but this because we are restricting ourselves to Hamming weight 1 increasing operations in a given step and as such it is only in very many steps that we can approach the full set of permutation unitaries.

\begin{figure}
    \centering
    \includegraphics[width = 0.8\linewidth]{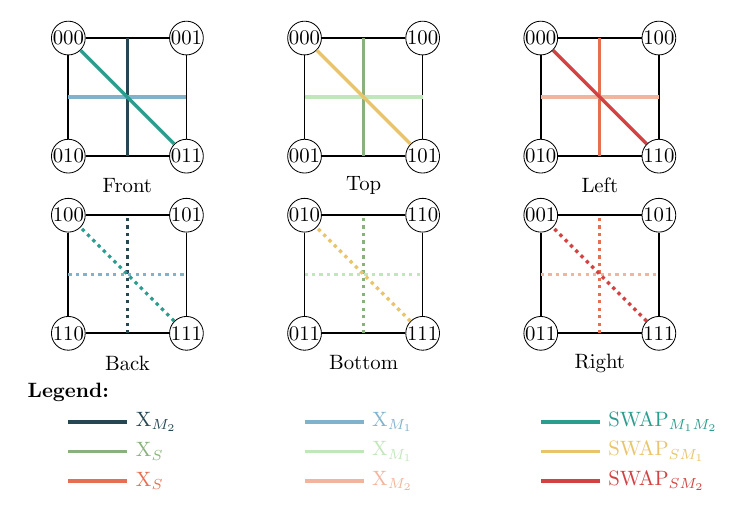}
    \caption{Here we visually represent the elements of the automorphism group of the cube graph $Q_3$ in terms of quantum gates. A $\sigma_x$ operation applied to one of these three qubits generates the subgroup $\mathbb{Z}^3_2$ by reflecting two opposite faces of the cube together through a vertical or horizontal line. We use $\mathrm{X}_i$ to denote the action of $\sigma_x$ on the $i$th qubit. The $\mathrm{SWAP}$ gate generates the permutation subgroup $S_3$ whose action reflects two opposite faces of the cube through a diagonal plane. Controlled Operations act on one face only and so do not preserve the symmetry of $\mathrm{Aut}(Q_3)$. We denote the action of controlled operations with dotted lines and anti-controlled (act conditioned on 0) using solid lines.}
    \label{fig:cube_ops}
\end{figure}

\subsubsection*{Can you cool a qubit with access to two qubits optimally without a Toffoli?}
\label{sec:toff}
Now that we know that we have these group theoretic tools at our disposal - can we use them to show that $G_0 =\{\mathrm X, \mathrm CX\}$ are not enough to cool a qubit to a temperature lower than that of a machine qubit? In the case that $\ket{0_S01} \leftrightarrow \ket{1_S10}$ and $\ket{0_S11} \leftrightarrow \ket{1_S 00}$ are both cooling energy exchanges we know that $G_0$ is enough since $U_\text{cool} = \mathtt{SWAP}_{S,M_{2}}$ which can be generated by $\mathtt{CX}$. But in the case that $U_\text{cool} = \ketbra{0_S11}{1_S00} + \ketbra{1_S00}{0_S11} + \mathbb{1}_\text{Rest} = S_{011,100}$ we can show that this cannot be generated by $G_0$ as no $S_{j,k}$ can be generated by $G_0$. 

\begin{proposition}
    It is not possible to cool a qubit $S$ with gap $\omega$ and temperature $T_S$ given access to two qubits with temperature $T_M$ and gaps $\gamma_1, \gamma_2$ without a $\mathtt{Toffoli}$ gate if 
    $$\gamma_2\leq \frac{T_M}{T_S}\omega + \gamma_1.$$
\end{proposition}
\begin{proof}
Firstly, if $\gamma_2\leq \frac{T_M}{T_S}\omega + \gamma_1$ then the only energy level exchange which cools is $\ket{0_S11} \leftrightarrow \ket{1_S00}$ that is $U_\text{cool} = S_{011,100}$. But no non-trivial $S_{j,k}$ can be produced by $\mathtt{X}$ and $\mathtt{CX}$. To see this,
note that the application of $\mathtt X$ and $\mathtt{SWAP}$ give proper subgroups of $\mathrm{Aut}(Q_3)$ and so give a coset group structure which partitions $\mathrm{Aut}(Q_3)$. We can represent the action of the subgroups whose semiproduct gives the automorphism group on the cube using the quantum gates $\mathtt X$ and $\mathtt{SWAP}$ as \begin{align*}\mathbb{Z}^3_2&=\{\mathtt{I, X_1, X_2, X_3, X_1X_2, X_1 X_3,X_2 X_3, X_1X_2X_3}\}, \\
S_3&=\{\mathtt{I, SWAP_{1,2}, SWAP_{2,3}, SWAP_{1,3}, SWAP_{1,2}SWAP_{2,3}, SWAP_{1,3}SWAP_{2,3}}\}. \end{align*}
As such $xS_3$ for every $x \in \mathbb{Z}^3_2$ partitions the action of the group into 8 orbits each of size 6 and similarly we have the partitioning $s\mathbb{Z}^3_2$ for every $s \in S_3$ which gives 6 orbits each of size 8. A basis ket permutation $S_{j,k}$ flips a single edge of the cube breaking the edge connectivity of the graph and so is not an automorphism of $Q_3$ and does not respect the symmetry of $\mathrm{Aut}(Q_3)$. As such it cannot be carried out using products of $\mathtt X$ and $\mathtt{SWAP}$. $\mathtt{CX}$ and $\mathtt{CSWAP}$ are similarly not automorphisms of $Q_3$ as they act on a single face of the cube and break edge-connectivity so perhaps they could generate an $S_{j,k}$. But notice that the actions $\mathtt{CX}$ and $\mathtt{CSWAP}$ act within the orbits of the partitions $xS_3$ and $s\mathbb{Z}^3_2$ respectively. Focusing on $\mathtt{CX}$ this operation acts on a single face of the cube flipping two edges at once. Say one desires one of the edge flips and not the other to give one $S_{j,k}$. Then, one could carry out $S_{j,k}$ if the action of the unwanted flip could be reversed under the action of another $\mathtt{X}$ or $\mathtt{CX}$ but this is not possible since both $\mathtt{X}$ and $\mathtt{CX}$ act partitions $xS_3$ and so reversing the action of an unwanted flip of a given $\mathtt{CX}_i$ is not possible using $\mathtt{X}_j$ or $\mathtt{CX}_j$ from another orbit since they form a partition. Therefore, it is not possible to generate $S_{j,k}$ using $\{\mathtt{X, CX, SWAP}\}$, as $\mathtt{SWAP}$ can be generate by $\mathtt{CX}$. As a result in the scenario where $U_\text{cool} = S_{011,100}$ one cannot optimally cool using $\{\mathrm{X, CX, SWAP}\}$. The necessity of $\mathtt{CCX}$ is then apparent by noting any Gray code from $011$ to $100$ e.g. $011,111, 110,100$ each step of which can be done with $\mathtt{X}$ and $\mathtt{CCX}$. In summary, we have shown that $\mathtt{X,CX}$ are not sufficient to generate $S_{011,100}$ and that $\mathtt{X, CCX}$ are sufficient. 

An alternative way to see this is that to optimally cool the system in this case an odd permutation is needed, as only a single pair of elements is permuted. As $\mathtt{X}$ and $\mathtt{CX}$ apply even permutations, and combinations of even permutations are also even permutations, no combination these gates can implement the necessary odd permutation to optimally cool. However, when restricting two $3$ qubits $\mathtt{CCX}$ is able to implement this odd permutation. Hence, given access to $\mathtt{CCX}$ one can optimally cool in this case. 
\end{proof}
Examining how the general automorphism group of the $n$-bit hypercube graph partitions under $G=\{\mathtt{X,CX,CCX, \dots C^{n-1}X}\}$ and deriving general no-go theorems in different $k$-reducibility scenarios could be an interesting direction for future research. 

\section{The Swappable Set}
\label{app:swappable_set}
                      
The problem of cooling a qubit with access to $n$ other qubits reveals a mathematical structure on the set of $n$ bit strings in terms of two partial orders. The goal of this section is to introduce these partial orders and in doing so to characterize the \emph{swappable set} $\mathbb{S} := \left\{ j_M\in \{0,1\}^{\times n} | \frac{1}{2}(\frac{T_{M}}{T_S}\omega + E_{\text{Max}}) < E(j_M) \right\}$, where $E(j_M) := j_M\cdot\Gamma$,  which contains the bit-strings denoting energy levels with populations that need to be swapped in order to maximally cool the system. 

One can combine these partial orders to reveal a forced structure on the energetics of the system and machine. This mathematical structure can be used to obtain some results that stand at the root of problem of how to identify which states need to be swapped and which should not. Fascinatingly, these results hold no matter the energy structure of the system and machine, as they arise directly from the Hilbert space structure of multiple qubits. We can therefore identify two parts of the general problem of identifying which are the elements of the swappable set: 1. the mathematical structure imposed by the Hilbert space structure of qubits 2. the specification of the energetic structure of the system and machine.

\subsection{Lexicographic Partial Order}
By definition of the swappable set, we can note that if for two bit-stings $j_M$ and $k_M$ we have $E(j_M) \geq E(k_M)$ then $k_M\in \mathbb{S} \implies j_M\in \mathbb{S}$. That is, if $k_M$ corresponds to an energy level that should be exchanged then $j_M$ should also be exchanged to cool $S$ -- in this $\mathbb{S}$ is ordered.

With knowledge of the refined energy structure of $\Gamma$, energy imposes a total order on the bitstrings. That is, that the map $E : \{0,1\}^{\times n}\rightarrow\mathbb{R}$ defines a total order on the bit-strings of the swappable set. However, we can weaken the total order defined by $E$ to find a complete partial order that is independent of the detailed structure of $\Gamma$ and only its global properties  i.e. that it is formed of non-decreasing energetic gaps. In what follows, we will describe and examing these partial orders on $\mathbb{S}$. To begin, we will denote the partial order relation between bit-strings (in text) by the symbol $\succeq$ and between kets (in figures) by $\rightarrow$.\\

\begin{definition}
    \emph{Hamming weight partial order.} $j_M,k_M\in\{0,1\}^{\times n}$ are $j_M\succeq_\# k_M$ if and only if $j_i \geq k_i$ for all $i=1,\dots,n$.
\end{definition}
This partial order encodes the action of flipping $0$s of a bit-string into $1$s: if we can generate $j_M$ by changing $0$s of $k_M$ into $1$s then we have $j_M\succeq_\# k_M$ and $\ket{k_M}\rightarrow_\# \ket{j_M}$. The three defining properties of a weak partial order are: 1. \emph{reflexivity} $a\succeq a$, \emph{antisymmetry} if $a\succeq b$ and $b\succeq a$ then $a = b$, and \emph{transitivity} if $a\succeq b$ and $b\succeq c$ then $a\succeq c$. These three properties are satisfied naturally by $\succeq_\#$. Furthermore the condition, if $j_M\succeq_\# k_M$ and $\exists~i$ such that $j_i > k_i$ renders this partial order a strict or strong partial order : $j_M\succ_\# k_M$. Naturally, this relation carries over to the energy of the corresponding ket
$$E(j_M) = \sum_{i=1}^n j_i\gamma_i \geq \sum_{i=1}^n k_i\gamma_i = E(k_M)~,$$
where we used the positivity of the gaps. We can therefore conclude that the total order induced by $E(\cdot)$ follows from the Hamming weight partial order with additional ordering constraints. 

Graphically, for $n=3$ we can map this order on the kets to obtain the following directed graph.
\begin{equation*}
    \begin{tikzpicture}[node distance=0.5cm and 0.7cm]
\setlength{\fboxsep}{1pt} 
\node (a) at (0,0) {$\ket{000}$};
\node (b2) [right=of a] {$\ket{010}$};
\node (b1) [above=of b2] {$\ket{100}$};
\node (b3) [below=of b2] {$\ket{001}$};
\node (c1) [right=of b1] {$\ket{110}$};
\node (c2) [right=of b2] {$\ket{101}$};
\node (c3) [right=of b3] {$\ket{011}$};
\node (d) [right=of c2] {$\ket{111}$};

\draw[arrow]  (a) -- node[pos=0.9,below=-1pt] {$_\#$} (b1);
\draw[arrow] (a) -- node[pos=0.95,below=-2pt] {$_\#$} (b2);
\draw[arrow] (a) -- node[pos=0.93,below=-2pt] {$_\#$} (b3);

\draw[arrow] (b1) -- node[pos=0.95,below=-2pt] {$_\#$} (c1);
\draw[arrow] (b1) -- node[pos=0.93,below=-2pt] {$_\#$} (c2);
\draw[arrow] (b2) -- node[pos=0.93,below=-1pt] {$_\#$} (c1);
\draw[arrow] (b2) -- node[pos=0.93,below=-2pt] {$_\#$} (c3);
\draw[arrow] (b3) -- node[pos=0.95,below=-2pt] {$_\#$} (c3);
\draw[arrow] (b3) -- node[pos=0.93,below=-1pt] {$_\#$} (c2);

\draw[arrow] (c1) -- node[pos=0.93,below=-2pt] {$_\#$} (d);
\draw[arrow] (c2) -- node[pos=0.95,below=-2pt] {$_\#$} (d);
\draw[arrow] (c3) -- node[pos=0.93,below=-1pt] {$_\#$} (d);

\end{tikzpicture}
\end{equation*}
This graph is simply a directed version of $Q_3$. At this point we can see that this graph seems to suggest to group the bit-strings by column, i.e. $\{0,1\}^{\times n} = \bigcup_{k=0}^n \mathbb{B}_k \, : \, \bigcap^n_{k=0} \mathbb{B}_k = \emptyset$ where $\mathbb{B}_k$ contains all the bit-strings with fixed Hamming weight $k$--i.e $j_M\in \{0,1\}^{\times n}$ is an element of $\mathbb{B}_k$ if and only if its number of ones $\#(j_M)$ is equal $k$. The Hamming weight partial order cannot compare two elements with the same Hamming weight. We can therefore use the fact that we can always re-order the machine gaps so that they are increasing $\gamma_1\leq \gamma_2\leq \dots\leq \gamma_n$ to obtain a partial order within each $\mathbb{B}_k$.\\

For $j_M \in \mathbb{B}_\ell$ let us label the indices where $j_M$ has a $1$ as $j_M(i)$ denoting the $i$th bit of $j_M$ is a 1.
\begin{definition}
    \emph{Positional partial order.} $j_M, k_M \in \mathbb{B}_\ell$ are $j_M\succeq_\text{p} k_M$ if and only if $j_M(i) \geq k_M(i)$ for all $i\in [1,\ell]$.
\end{definition}
This partial order encodes the action of swapping a $10$ for a $01$ in the bit-string. Therefore if we can construct a $j_M$ from a $k_M$ via this action, then we have $j_M\succeq_{\text p} k_M$ and $\ket{k_M}\rightarrow_{\text p} \ket{j_M}$. Similarly to $\succeq_\#$, $\succeq_{\text p}$ satisfies the defining properties of a weak partial order, can be made into the strong partial order $\succ_{\text p}$ by requiring $\exists~i$ such that $j_M(i) > k_M(i)$, and the order carries to the energies of the kets
$$E(j_M) = \sum_{i=1}^\ell \gamma_{j_M(i)} \geq \sum_{i=1}^\ell \gamma_{k_M(i)} = E(k_M)~,$$
where we used that we can rewrite the energy as $E(k_M) = \sum_{j=1}^n k_j\gamma_j= \sum_{i=1}^\ell \gamma_{k_M(i)}$ since $k_{j} = 1$ for $j\in\{k_M(i)\}_{i=1}^\ell$ and $k_{j} = 0$ if $j\notin\{k_M(i)\}_{i=1}^\ell$. This allows us to conclude that the total order induced by $E$ is follows from the positional partial order with additional conditions. Conceptually, this partial order compares which out of two bit strings $i_M,j_M$  \textit{has more 1s on the right} and so features higher energetic gaps contributing to their energies $E(i_M)$ and $E(j_M)$.

If we map this order onto a directed graph of kets we obtain a directed graph whose structure is analogous to that of an antisymmetric $\ell$-tensor over an $n$-dimensional space. For $n=5$ and $\ell=2$ we obtain
\begin{equation*}
    \begin{tikzpicture}[node distance=0.5cm and 0.75cm]
\setlength{\fboxsep}{1pt} 
\node (a1) at (0,0) {$\ket{11000}$};
\node (b1) [below=of a1] {$\ket{10100}$};
\node (b2) [right=of b1] {$\ket{01100}$};
\node (c1) [below=of b1] {$\ket{10010}$};
\node (c2) [below=of b2] {$\ket{01010}$};
\node (c3) [right=of c2] {$\ket{00110}$};
\node (d1) [below=of c1] {$\ket{10001}$};
\node (d2) [below=of c2] {$\ket{01001}$};
\node (d3) [below=of c3] {$\ket{00101}$};
\node (d4) [right=of d3] {$\ket{00011}~.$};

\draw[arrow]  (a1) -- node[pos=0.65,left=-1pt] {$_\text{p}$} (b1);
\draw[arrow]  (b1) -- node[pos=0.65,left=-1pt] {$_\text{p}$} (c1);
\draw[arrow]  (c1) -- node[pos=0.65,left=-1pt] {$_\text{p}$} (d1);
\draw[arrow]  (b2) -- node[pos=0.65,left=-1pt] {$_\text{p}$} (c2);
\draw[arrow]  (c2) -- node[pos=0.65,left=-1pt] {$_\text{p}$} (d2);
\draw[arrow]  (c3) -- node[pos=0.65,left=-1pt] {$_\text{p}$} (d3);

\draw[arrow] (b1) -- node[pos=0.93,below=-1pt] {$_\text{p}$} (b2);
\draw[arrow] (c1) -- node[pos=0.93,below=-1pt] {$_\text{p}$} (c2);
\draw[arrow] (c2) -- node[pos=0.93,below=-1pt] {$_\text{p}$} (c3);
\draw[arrow] (d1) -- node[pos=0.93,below=-1pt] {$_\text{p}$} (d2);
\draw[arrow] (d2) -- node[pos=0.93,below=-1pt] {$_\text{p}$} (d3);
\draw[arrow] (d3) -- node[pos=0.93,below=-1pt] {$_\text{p}$} (d4);
\end{tikzpicture}
\end{equation*}
This graph can be obtained by assigning a \textit{direction} or degree of freedom for each $1$ and moving along the designated direction when the corresponding $1$ is shifted, resulting in a graph that expands in $\ell$ directions--e.g. for $\ell=3$ vertically for the third $1$, horizontally ($x$) for the second, and horizontally ($y$) for the first.\\

Since the Hamming weight partial order only compares elements across two different $\mathbb{B}_\ell$ and the positional partial order only compares within a single $\mathbb{B}_\ell$, it is very natural to try to combine the two by finding a single partial order that does not depend on the \textit{fine details of the} energy gaps and extends both of these.

\begin{definition}
    \emph{Lexicographic partial order.} $j_M, k_M \in \{0,1\}^{\times n}$ with $J=\#(j_M)$ and $K=\#(k_M)$. Then $j_M \succeq k_M$ if and only if $J\geq K$ and $j_M(i+J-K) \geq k_M(i)$ for all $i=1,\dots,K$.
\end{definition}
Reflexivity and antisymmetry follow from $\succeq_{\text{p}}$ being a partial order. For transitivity we only need to note that if $j_M\succeq k_M$ and $k_M\succeq l_M$, then we can shift $i+J-K$ by $K-L$ to obtain $j_M(i+J-L) \geq k_M(i+K-L) \geq l_M(i)$ for all $i=1,\dots,L=\#(l_M)$.

Similarly to the previous two orders, we can make the lexicographic order strict: $j_M\succ k_M$ if and only if $j_M\succeq k_M$ and either $J > K$ or $j_M\succ_\text p k_M$. This order is truly an extension of both the positional order and the Hamming weight order: we can see this because we can define $l_M\in \mathbb{B}_{K}$ such that $l_M(i)=j_M(i+J-K)$ for $i=0,\dots,K$, then by construction we have $l_M\succeq_\text p k_M$ and $j_M\succeq_\#l_M$. This intermediary element also allows us to conclude that the total order induced by $E$ is an extension of the lexicographic order with further ordering constraints. By mapping it onto a graph for $n=3$ and $n=4$ we obtain
\begin{equation*}
    \begin{tikzpicture}[node distance=0.5cm and 0.7cm]
\setlength{\fboxsep}{1pt} 
\node (a) at (0,0) {$\ket{000}$};
\node (b2) [right=of a] {$\ket{010}$};
\node (b1) [above=of b2] {$\ket{100}$};
\node (b3) [below=of b2] {$\ket{001}$};
\node (c1) [right=of b1] {$\ket{110}$};
\node (c2) [right=of b2] {$\ket{101}$};
\node (c3) [right=of b3] {$\ket{011}$};
\node (d) [right=of c2] {$\ket{111}$~,};

\draw[arrow] (a) --(b1);
\draw[arrow] (a) -- (b2);
\draw[arrow] (a) -- (b3);

\draw[arrow] (b1) -- (c1);
\draw[arrow] (b1) -- (c2);
\draw[arrow] (b2) -- (c1);
\draw[arrow] (b2) -- (c3);
\draw[arrow] (b3) -- (c3);
\draw[arrow] (b3) -- (c2);

\draw[arrow] (c1) -- (d);
\draw[arrow] (c2) -- (d);
\draw[arrow] (c3) -- (d);

\draw[arrow] (b1) -- (b2);
\draw[arrow] (b2) -- (b3);
\draw[arrow] (c1) -- (c2);
\draw[arrow] (c2) -- (c3);

\end{tikzpicture}
\end{equation*}

\begin{equation*}
    \begin{tikzpicture}[node distance=0.4cm and 1.2cm]
\setlength{\fboxsep}{1pt} 
\node (a) at (0,0) {$\ket{0000}$};

\node (b2) [right=of a, yshift=0.4cm] {$\ket{0100}$};
\node (b1) [above=of b2] {$\ket{1000}$};
\node (b3) [below=of b2] {$\ket{0010}$};
\node (b4) [below=of b3] {$\ket{0001}$};

\node (c2) [right=of b1] {$\ket{1010}$};
\node (c1) [above=of c2] {$\ket{1100}$};
\node (c3) [right=of b2] {$\ket{1001}$};
\node (c4) [right=of b3] {$\ket{0110}$};
\node (c5) [right=of b4] {$\ket{0101}$};
\node (c6) [below=of c5] {$\ket{0011}$};

\node (d1) [right=of c2] {$\ket{1110}$};
\node (d2) [right=of c3] {$\ket{1101}$};
\node (d3) [right=of c4] {$\ket{1011}$};
\node (d4) [right=of c5] {$\ket{0111}$};

\node (e) [right=of d2, yshift=-0.4cm] {$\ket{1111}$~.};

\draw[arrow] (a) -- (b1);
\draw[arrow] (a) -- (b2);
\draw[arrow] (a) -- (b3);
\draw[arrow] (a) -- (b4);

\draw[arrow] (b1) -- (c1);
\draw[arrow] (b1) -- (c2);
\draw[arrow] (b1) -- (c3);
\draw[arrow] (b2) -- (c1);
\draw[arrow] (b2) -- (c4);
\draw[arrow] (b2) -- (c5);
\draw[arrow] (b3) -- (c2);
\draw[arrow] (b3) -- (c4);
\draw[arrow] (b3) -- (c6);
\draw[arrow] (b4) -- (c3);
\draw[arrow] (b4) -- (c5);
\draw[arrow] (b4) -- (c6);

\draw[arrow] (c1) -- (d1);
\draw[arrow] (c1) -- (d2);
\draw[arrow] (c2) -- (d1);
\draw[arrow] (c2) -- (d3);
\draw[arrow] (c3) -- (d2);
\draw[arrow] (c3) -- (d3);
\draw[arrow] (c4) -- (d1);
\draw[arrow] (c4) -- (d4);
\draw[arrow] (c5) -- (d2);
\draw[arrow] (c5) -- (d4);
\draw[arrow] (c6) -- (d3);
\draw[arrow] (c6) -- (d4);

\draw[arrow] (d1) -- (e);
\draw[arrow] (d2) -- (e);
\draw[arrow] (d3) -- (e);
\draw[arrow] (d4) -- (e);

\draw[arrow] (b1) -- (b2);
\draw[arrow] (b2) -- (b3);
\draw[arrow] (b3) -- (b4);
\draw[arrow] (d1) -- (d2);
\draw[arrow] (d2) -- (d3);
\draw[arrow] (d3) -- (d4);

\draw[arrow] (c1) -- (c2);
\draw[arrow] (c2) -- (c3);
\draw[arrow] (c4) -- (c5);
\draw[arrow] (c5) -- (c6);

\draw[arrow] (c2) to[bend right=40] (c4);
\draw[arrow] (c3) to[bend left=40] (c5);

\end{tikzpicture}
\end{equation*}

We underline here that the lexicographic order is a partial order and not a total order as non-comparable elements still remain, for example $\ket{0001}$ and $\ket{1100}$, unlike in an ordering of these bitstrings with complete knowledge of the energetic structure which induces a total order. We can therefore apply this lexicographic order to the original observation to obtain the following lemma.
\begin{lemma}\label{lem:partial_order}
    Let $j_M,k_M\in\{0,1\}^{\times n}$ such that $j_M \succeq k_M$ and $k_M\in \mathbb{S}$. Then $j_M \in \mathbb{S}$.\\
\end{lemma}
\emph{Proof.} Since the total order induced by $E(\cdot)$ is an extension of $\succeq$, we also have $E(j_M)\geq E(k_M)$. Therefore, by the definition of $S$, we have $j_M\in \mathbb{S}$. \qed \\

From this lemma we can conclude that the characterization of $S$ without knowledge of each individual $\gamma_j \in \Gamma$ comes down to finding a small collection of bit-strings that are incomparable to each other from which we can generate all the other elements of $S$.

\subsection{Always and Never Swappable Sets}
The structure induced by the partial order allows us to note some interesting aspects about the characterization of the swappable set that do not depend on the energetic structure of the machine. In particular we seek to analyse if there are some elements of $\{0,1\}^{\times n}$ that are always or never part of $\mathbb{S}$. Which is to say, ground state subspace energy level which when exchanged for excited subspace levels seem to always cool or heat $S$ regardless of the structure of  $M$. In order to do so, it is worth noticing that the definition of $\mathbb S$ is dependent on the system gap $\omega$: as we make $\omega$ larger then we can reduce the size of $\mathbb S$ to the point that it becomes the empty set for $\omega \geq E_{\text{Max}}$. Therefore, we will consider for a moment the set $\mathbb{S}_0 := \left\{ j_M\in \{0,1\}^{\times n} | \frac{1}{2}E_{\text{Max}} < E(j_M) \right\}$. Note that, by definition, $\mathbb{S} \subseteq \mathbb{S}_0$.

To characterize $\mathbb{S}_0$ we are trying to find which energy levels are always larger or smaller than the median $\frac{1}{2}E_{\text{Max}}$. Since $|\{0,1\}^{\times n}|= 2^n$, which is even, the median $\frac{1}{2}E_{\text{Max}}$ is not defined by a specific energy level $E(j_M)$, but from the average of two energy levels $E(j_M)$ and $E(j_M\oplus 1)$. Therefore, almost everywhere for $\Gamma \in (\mathbb R^+)^n$, there is no $j_M$ such that $E(j_M) = \frac{1}{2}E_{\text{Max}}$. In turn this implies that, $|\mathbb{S}_0| = 2^{n-1}$ by definition of the median, for almost every energy structure. It is important to note that counter-examples can easily be found to this statement: it is sufficient to construct an energy structure such that for some $j_M$ we have $E(j_M) = E(j_M\oplus 1)$, however it is easy to see that these examples are not robust to arbitrarily small noise on the energy gaps. Furthermore, it is also worth noting that $\mathbb{S}_0$ corresponds to the swaps that \emph{need} to be performed in order to fully cool the system. However, swaps the on states corresponding to the bit-strings such that $E(j_M) = E(j_M\oplus 1) = \frac{1}{2}E_{\text{Max}}$ neither cool nor heat the system. Therefore, when $\omega = 0$, one can always bring the cardinality of the set of states that are being swapped to $2^{n-1}$.\\

Two main observations are at play for the characterization of which states are guaranteed to be in (or out of) $\mathbb{S}_0$, first is \lemref{lem:partial_order}, and the second is the existence of \emph{incompatible bit-strings}. We say that two elements of $\{0,1\}^{\times n}$ are incompatible when the belonging of one in $\mathbb{S}_0$ is equivalent to the exclusion of the other in $\mathbb{S}_0$. A natural example of such two elements can always be constructed with any $j_M\in \{0,1\}^{\times n}$ by picking its conjugate $\tilde j_M:=j_M\oplus1$: by definition of $E(\cdot)$ we have $E(\tilde j_M) = E_{\text{Max}}-E(j_M)$, which allows us to conclude that  $j_M \in \mathbb{S}_0 \Rightarrow \tilde j_M\notin \mathbb{S}_0$. If we further assume that $E(j_M) \neq \frac{1}{2}E_{\text{Max}}$ then we obtain the other direction: $j_M\notin \mathbb{S}_0 \Rightarrow \tilde j_M\in \mathbb{S}_0$. Therefore there are three scenarios: (i) $j_M\in \mathbb{S}_0$ and $\tilde j_M\notin \mathbb{S}_0$ (ii) $j_M\notin \mathbb{S}_0$ and $\tilde j_M\in \mathbb{S}_0$ (iii) $j_M\notin \mathbb{S}_0$ and $\tilde j_M\notin \mathbb{S}_0$, with (iii) happening if and only if $E(j_M) = \frac{1}{2}E_{\text{Max}}$. Then it is worth noting that if for a bit-string $k_M$ there exists $j_M$ such that $j_M \succeq k_M$ and $\tilde j_M \succeq k_M$ then $k_M \notin \mathbb{S}_0$ regardless of the energy structure. And, almost always, we also will have $\tilde k_M\in \mathbb{S}_0$. This is particularly interesting for cases where $\tilde k_M \succeq k_M$, where we can pick $j_M = k_M$ and reach the same conclusion. It turns out this is the defining and only characteristic of the bit-strings that are either never or (almost) always in $\mathbb{S}_0$ and that it generalizes to the entire set $\mathbb{S}$.

With these observations, we outline with the following theorem what are the bit-strings that correspond to the states that always need to be swapped and those that never need to be swapped.

\begin{theorem}[Always and never swappable levels]\label{thm:S0} We define the following sets
    \begin{enumerate}[label=(\alph*)]
        \item the $\ell$-minimal swappable set $\mathbb{S}_{\text{min}}^\ell =\{j_M \in \{0,1\}^{\times n} |~j_M \succeq j_M\oplus 1 \text{ and }\#(j_M) > \frac{n}{2}+\ell\}$, with $\ell=-1,0,1,\dots,\lfloor\frac{n-1}{2}\rfloor$,
        \item the never swappable set $\widetilde{\mathbb{S}}_{\text{n}} := \{j_M\in \{0,1\}^{\times n}|~j_M\oplus1\succeq j_M\}$.
    \end{enumerate}
    Then the following statements hold true
    \begin{enumerate}[label=(\roman*)]
        \item $\widetilde{\mathbb{S}}_n$ is the largest set such that $\widetilde{\mathbb{S}}_n \cap \mathbb{S} = \varnothing$ for all $0<\gamma_1\leq \gamma_2\leq\dots\leq \gamma_n$ and $0\leq \frac{T_M}{T_S}\omega$,
        \item $\mathbb{S}_{\text{min}}^\ell$ is the largest set such that $\mathbb{S}_{\text{min}}^\ell\subseteq \mathbb  S$ for all $0<\gamma_1\leq \gamma_2\leq\dots\leq \gamma_n$ and $0\leq \frac{T_M}{T_S}\omega<\sum_{i=1}^{L}\gamma_i$ for $\ell\geq0$,
        \item $\mathbb{S}_{\text{min}}^{-1}$ is the largest set such that $\mathbb{S}_{\text{min}}^{-1}\subseteq \mathbb S$ for all $0<\gamma_1\leq\dots\leq \gamma_n$ and $\frac{T_M}{T_S}\omega<\min(\gamma_1+\gamma_2,\gamma_n-\gamma_1)$ for $n$ even,
    \end{enumerate}
    where $L=2(\ell+1)$ for $n$ even and $L=2\ell+1$ for $n$ odd.
\end{theorem}
\emph{Proof of (ii) and (iii).} We can rephrase statement (ii) as $j_M\in \mathbb S$ for all energy structures with $\frac{T_M}{T_S}\omega < \sum_{i=1}^{L}\gamma_i$ if and only if $j_M\in \mathbb{S}_\text{min}^\ell$ ($\ell\geq 0$). ($\Leftarrow$ of ii) Assuming $j_M\in \mathbb{S}_\text{min}^\ell$, we have $j_M\geq \tilde j_M$ and $J=\#(j_M)>\frac{n}{2}+\ell$. By definition of the lexicographic partial order and of conjugate bit-strings, we necessarily have
\begin{equation}
    E(j_M) = \sum_{i=1}^J \gamma_{j_M(i)} \geq E(\tilde j_M)+\sum_{i=1}^{2J-n}\gamma_{j_M(i)} \geq E(\tilde j_M) +\sum_{i=1}^{2J-n}\gamma_i~.
\end{equation}
By using $E(\tilde j_M) = E_{\text{Max}}-E(j_M)$, we find $E(j_M)\geq \frac{1}{2}(E_{\text{Max}}+\sum_{i=1}^{2J-n}\gamma_i)$.
For $n$ even we have $J\geq \frac{n}{2}+1+\ell$, and therefore $2J-n \geq 2(\ell+1)$. While for $n$ odd we have $J\geq \frac{n+1}{2}+\ell$, and therefore $2J-n \geq 2\ell+1$. Therefore $j_M\in \mathbb S$ for $\frac{T_M}{T_S}\omega < \sum_{i=1}^{L}\gamma_i$.\\
Similarly to (ii), We can rephrase statement (iii) $j_M\in \mathbb S$ for all energy structures with $\frac{T_M}{T_S}\omega\leq \textnormal{min}(\gamma_1+\gamma_2,\gamma_n-\gamma_1)$ if and only if $j_M\in \mathbb{S}_\text{min}^{-1}$ when $n$ is even. ($\Leftarrow$ of iii) It is sufficient to extend the argument for ($\Leftarrow$ of ii) with $\ell=0$ to the elements of $\mathbb{S}_{\text{min}}^{-1}$ that do not belong in $\mathbb{S}_{\text{min}}^{0}$. Assuming $j_M\in \mathbb{S}_{\text{min}}^{-1}$ and $j_M\notin \mathbb{S}_{\text{min}}^{0}$ we have $\#(j_M)=\#(\tilde j_M)=\frac{n}{2}$. Furthermore we also have that $j_M(i)>\tilde j_M(i)$ for all $i=1,\dots,\frac{n}{2}=\#(j_M)=\#(\tilde j_M)$, therefore $j_M(\frac{n}{2}) = n$ and $\tilde j_M(1) =1$, which implies $E(j_M) \geq E(\tilde j_M) +\gamma_n-\gamma_1$. Similarly to the argumentation in ($\Leftarrow$ of ii), we get $E(j_M) \geq \frac{1}{2}(E_{\text{Max}} +\gamma_n-\gamma_1)$ and therefore $j_M\in \mathbb S$ for $\frac{T_M}{T_S}\omega < \gamma_n-\gamma_1$. Concluding the proof of ($\Leftarrow$ of iii).\\
($\Rightarrow$ of ii and iii) Conversely, if we now assume $j_M\notin \mathbb{S}_\text{min}^\ell$, then either $j_M\in \mathbb{S}_\text{min}^{\ell-1}$ (for $\ell\geq0$) or $j_M\not\succeq\tilde j_M$. In both cases it is sufficient to show that there exists an energy structure (that satisfies the requirements of the corresponding part of the theorem) such that $j_M\notin \mathbb S$. We start with the part that is only required for the proof of (ii), assuming $j_M\in \mathbb{S}_\text{min}^{\ell-1}$ and $j_M\notin \mathbb{S}_\text{min}^{\ell}$ (with $\ell\geq0$), then $\#(j_M)=\lfloor\frac{n}{2}\rfloor+\ell$. For the energy structure $\gamma_1=\gamma_2=\dots=\gamma_n$ this means $E(j_M) = \frac{1}{2}(E_{\text{Max}}+ 2\ell\gamma_1)$ for $n$ even and $E(j_M) = \frac{1}{2}(E_{\text{Max}} + (2\ell-1)\gamma_1)$ for $n$ odd. By picking $\frac{T_M}{T_S}\omega = 2\ell\gamma_1$ we have $\frac{T_M}{T_S}\omega < (2\ell+1)\gamma_1 =\sum_{i=1}^{2\ell+1}\gamma_i$ and $E(j_M) \leq \frac{1}{2}(E_{\text{Max}}+\frac{T_M}{T_S}\omega)$, therefore $j_M\notin\mathbb S$.\\
Assuming $j_M\not\succeq \tilde j_M$ we have that either (a) $J<\frac{n}{2}$ or (b) $\exists~\ell\in\{1,\dots,n-J\}$ such that $\tilde j_M(\ell) > j_M(\ell+2J-n)$. For (a) we pick $\gamma_1=\gamma_2=\dots=\gamma_n$, then it is immediate that $E(\tilde j_M) > E(j_M)$, therefore $\frac{1}{2}E_{\text{Max}} > E(j_M)$ and $j_M\notin S$. For (b) we pick, $\gamma_1=\dots=\gamma_{\tilde j_M(\ell)-1}$ and $\gamma_{\tilde j_M(\ell)}=\dots=\gamma_n>\gamma_1$. This leads to $E(\tilde j_M) = (\ell-1)\gamma_1+ (n-J-\ell+1)\gamma_n$. While for $j_M$ we have
\begin{equation}
    E(j_M) = \sum_{i=1}^J\gamma_{j_M(i)} = (\ell+2J-n)\gamma_1 + \sum_{i=\ell+2J-n+1}^J \gamma_{j_M(i)} \leq (\ell+2J-n)\gamma_1+(n-J-\ell)\gamma_n~.
\end{equation}
Therefore we can pick $\gamma_n$ large enough such that $E(\tilde j_M)> E(j_M)$. Therefore $j_M\notin \mathbb S$ for all $\omega$. Thus concluding the proof (ii) and (iii).\\

\emph{Proof of (i).} The statement can be rephrased as $j_M\notin \mathbb S$ for all energy structures if and only if $j_M\in\widetilde{\mathbb{S}}_{\text{n}}$. By definition of the lexicographic order, $j_M\in \tilde S_\text{n}$ satisfies $\#(j_M)\leq\#(\tilde j_M)$. Therefore it is immediate that $\widetilde{\mathbb{S}}_{\text{n}} = \{j_M\oplus1|~j_M\in \mathbb{S}_\text{min}^{-1}\}$. We obtain ($\Leftarrow$) immediately from \lemref{lem:partial_order} and incompatibility of $j_M$ and $j_M\oplus 1$.\\
($\Rightarrow$) Conversely, we now assume $\tilde j_M\not\succeq j_M$. By definition of the lexicographic order, either (a) $J=\#(j_M)>n/2$ or (b) there exists an $\ell \in \{1,\dots,J\}$ such that $\tilde j_M(\ell+n-2J) < j_M(\ell)$. In both cases we construct examples of energy structures where $j_M\in \mathbb S$: for (a) $J>n/2$ we can pick $\gamma_1=\dots=\gamma_n$, from which it is immediate that $E(j_M) > E(\tilde j_M)$ and therefore $E(j_M) > \frac{1}{2}E_{\text{Max}}$. For (b) we can pick $\gamma_1=\dots=\gamma_{j_M(\ell)-1}$ and $\gamma_{j_M(\ell)}=\dots=\gamma_n>\gamma_1$. Similarly to the proof of (ii) and (iii) we obtain $E(j_M) = \sum_{i=1}^J\gamma_{j_M(i)} = (J-\ell+1)\gamma_n + (\ell-1)\gamma_1$ and $E(\tilde j_M)= \sum_{i=1}^{n-J}\gamma_{\tilde j_M(i)} \leq (J-\ell)\gamma_n +(\ell+n-2J)\gamma_1$. Therefore we can choose $\gamma_n$ large enough such that $E(j_M) > E(\tilde j_M)$. Therefore we conclude that if $j_M\notin \widetilde{\mathbb{S}}_{\text{n}}$ then there are energy structures such that $E(j_M) > \frac{1}{2}E_{\text{Max}}$ and therefore for which $j_M\in \mathbb S$. Thus concluding the proof of (i) and the theorem.\qed\\

Let us briefly unpack the statements made by this theorem by considering first the never-swappable set, as it is more intuitive. Starting from observations made above that use the lexicographic partial order, we are able to find a set of states that, regardless of the energy structure of the machine, will not need to be swapped. Furthermore, the set $\widetilde{\mathbb{S}}_{\text{n}}$ is maximal, in the sense that for any bit-string $j_M\notin \widetilde{\mathbb{S}}_{\text{n}}$ then we can find an energy structure where one needs to act on the population of the corresponding sate $\ket{0_Sj_M}$ to cool the system.

Ideally, we would want to make a similarly simple analysis on the set of states that are always need to be acted upon in order to cool the system. However, there is a complication to be taken into account: no state is \emph{always} worth swapping as one can simply pick a value of $\omega$ that is large enough so that it is not worth swapping. Therefore one has to work around this problem by keeping track of how cold is already the initial system. We can start by considering the case where $\omega=0$, for which the set states that always need to be swapped, regardless of the energy structure of the machine, is given by the $\mathbb S_{min}^{-1}$. Incidentally, and intuitively, the elements of this set are the conjugate elements of the never swappable set. Similarly to the never swappable set, $\mathbb S_{min}^{-1}$ is maximal in the sense that for any bit-string $j_M\notin \mathbb S_{min}^{-1}$ we can find a corresponding energy structure (that satisfies $\omega=0$) such that the state $\ket{0_Sj_M}$ does not need to be swapped in order to cool the system. As one increases the value of $\omega$ there will be some states that stop being always in the swappable set. Therefore one needs to ``adjust'' the mathematical statement accordingly. In the theorem we do for every value of $\omega$ until the machine becomes completely useless--which happens when the system starts so cold that the machine cannot be used to cool it further.

\subsubsection{Constructing the $\ell$-minimal swappable set}
We now move on to consider which are the elements of $\mathbb S_{min}^{\ell}$. For the practicality of this section, we will decompose this set into its hamming weight sectors: $A_k :=\{j_M \in \{0,1\}^{\times n} |~j_M \succeq j_M\oplus 1 \text{ and }\#(j_M) =k\}$, so that $$\mathbb S_{min}^{\ell} = \bigcup_{k>\frac{n}{2}+\ell} A_k~.$$

The defining property of the lexicographic partial order is that an bit-string $j_M$ is larger than a bit-string $k_M$ when for each $1$ in $k_M$ we can identify a unique $1$ in $j_M$ that is ``further to the right''. Therefore the condition $j_M \succeq j_M\oplus 1$ means that for every $0$ in $j_M$ there is a $1$ uniquely associated to this $0$ further to the right in the bit-string. Very naturally, we can immediately see that $A_\ell$ is empty for $\ell<\frac{n}{2}$ since there are more zeros than ones in those bit-strings. However the problem becomes less trivial for $\ell\geq\frac{n}{2}$ (which we will assume for the rest of this section), as there are bit-strings such as $1110$ which cannot be part of the set $A_3$ for $n=4$.

In practice one can simply brute-force check every element, and therefore one would obtain the list of elements of $A_k$ for a given $n$.  However one can do better than this--and even obtain analytical expressions that can be used for bounds--by using the property that $j_M \succeq k_M \Leftrightarrow k_M\oplus 1 \succeq j_M\oplus 1$ (this follows immediately from the fact that $\oplus1$ swaps the zeros and ones). In particular, let us consider $j_M,k_M\in\mathbb B_\ell$ such that $j_M\succeq k_M$. Then it is immediate from the above property (and transitivity) that $k_M\in A_\ell \Rightarrow j_M\in A_\ell$. 

Let us consider the decreasing operation of the positional partial order: within the bit-string we change a $01$ for a $10$. Since for any element of $A_k$ one can find a unique $1$ on the right of each $0$, then we can apply the decreasing operation to the bit-string until we obtain $a_M^k=1^{2k-n}(01)^{n-k}$. We can note that $a_M^k\in A_k$ and that if we apply any decreasing operation to $a_M^k$ then the resulting bit-string is not an element of $A_k$. Therefore $a_M^k$ is the smallest element of $A_k$ by the positional partial order. Since $A_k$ has a smallest element, by using the property mentioned above, we can construct all of $A_k$ by applying iteratively the increasing operation $10\rightarrow01$ to $a_M^k$. Furthermore, it is also worth noting that $0^{n-k}1^k$ is an element of $A_k$ by construction and is the largest element of $\mathbb B_k$, and therefore is also the largest element of $A_\ell$.\\

Finally, we can also note that if we change any $0$ of $a_M^k$ for a $1$ then we obtain an element that is larger or equal to $a_M^{k+1}$. The same remains true for any other element of $A_k$, since they are larger than $a_M^k$. But we also have $a_M^{k+1}\succeq a_M^k$, which allows to extend our generative method from one $A_k$ to all of them. Therefore we can simplify the definition of the $\ell$-minimal swappable set to be all the elements that are larger than $a_M^{\lfloor\frac{n}{2}+\ell+1\rfloor}$ by the lexicographic partial order
\begin{equation}
    \mathbb S^\ell_{min} = \{j_M \in \{0,1\}^{\times n} |~j_M \succeq a_M^{\lfloor\frac{n}{2}+\ell+1\rfloor}\}~.
\end{equation}
And one can generate this set by considering all the elements obtained by starting from $a_M^{\lfloor\frac{n}{2}+\ell+1\rfloor}$ and either changing a $0$ into a $1$ or changing the sequence $10$ to $01$ in the bit-string.

We can therefore rephrase the statements (ii) and (iii) of the theorem as\\
\emph{For $\sum_{i=1}^{2(\ell-1)-n}\gamma_i \leq \frac{T_M}{T_S}\omega <\sum_{i=1}^{2\ell-n}\gamma_i$ (with  $\ell = \lceil \frac{n+1}{2}\rceil,\dots,n$) the largest subset of $\mathbb{S}$ that is agnostic of $\Gamma$ is given by $\{j_M \in \{0,1\}^{\times n} |~j_M \succeq 1^{2\ell-n}(01)^{n-\ell}\}$. Furthermore, when $n$ is even and $\frac{T_M}{T_S}\omega<\min(\gamma_1+\gamma_2,\gamma_n-\gamma_1)$ we can expand this set to $\{j_M \in \{0,1\}^{\times n} |~j_M \succeq (01)^{n/2}\}$.}

\end{document}